\documentclass{article}

\usepackage{authblk}
\usepackage[square]{natbib}
\usepackage{yhmath}
\usepackage{mathtools, amsmath, amsthm, amssymb, amsbsy, amsfonts, amscd}
\usepackage{euscript}
\usepackage{bm}
\usepackage{color}
\usepackage{empheq}
\usepackage{rotating}
\usepackage{enumerate}
\usepackage{enumitem}
\usepackage{bbold}

\usepackage{ftnxtra}
\usepackage{fnpos}
\usepackage{appendix}

\usepackage{graphicx}
\usepackage{pstool}
\usepackage{psfrag}

\usepackage{comment}

\numberwithin{equation}{section}

\theoremstyle{plain}	
 \newtheorem{thm}{Theorem}[section]

 \newtheorem{prop}[thm]{Proposition}
\theoremstyle{definition}	
 \newtheorem{defi}{Definition}[section]
 \newtheorem{remark}{Remark}[section]
 \newtheorem{example}{Example}[section]

\usepackage{caption2}

\usepackage{hyperref}
\hypersetup{colorlinks=true, linkcolor=blue}
\hypersetup{colorlinks=true,citecolor=blue}

\DeclareMathAlphabet{\mathpzc}{OT1}{pzc}{m}{it}

\usepackage{mathrsfs}

\definecolor{lighter_purple_mathematica}{rgb}{0.6666666666,0.33333333333,0.666666666666}

\makeatletter
\newsavebox{\@brx}
\newcommand{\llangle}[1][]{\savebox{\@brx}{\(\m@th{#1\langle}\)}%
  \mathopen{\copy\@brx\mkern2mu\kern-0.9\wd\@brx\usebox{\@brx}}}
\newcommand{\rrangle}[1][]{\savebox{\@brx}{\(\m@th{#1\rangle}\)}%
  \mathclose{\copy\@brx\mkern2mu\kern-0.9\wd\@brx\usebox{\@brx}}}%
\let\oldabs\abs
\def\abs{\@ifstar{\oldabs}{\oldabs*}}
\makeatother

\usepackage{accents}
\newcommand{\Fe}{\accentset{e}{\mathbf{F}}}
\newcommand{\Fa}{\accentset{a}{\mathbf{F}}}

\newcommand{\cFa}{\accentset{a}{F}}

\usepackage{graphicx}
\newcommand{\SOne}{\accentset{\scalebox{0.4}{(1)}}{\mathbf{S}}}
\newcommand{\SOneC}{\accentset{\scalebox{0.4}{(1)}}{S}}

\newcommand{\N}{\accentset{1}{\boldsymbol{\mathsf{N}}}}
\newcommand{\NN}{\accentset{2}{\boldsymbol{\mathsf{N}}}}
\newcommand{\NNN}{\accentset{3}{\boldsymbol{\mathsf{N}}}}

\newcommand{\Nj}{\accentset{j}{\boldsymbol{\mathsf{N}}}}

\newcommand{\n}{\accentset{1}{\boldsymbol{\mathsf{n}}}}
\newcommand{\nn}{\accentset{2}{\boldsymbol{\mathsf{n}}}}
\newcommand{\nnn}{\accentset{3}{\boldsymbol{\mathsf{n}}}}
\newcommand{\nj}{\accentset{j}{\boldsymbol{\mathsf{n}}}}

\newcommand{\sigmaR}{\accentset{R}{\boldsymbol{\sigma}}}
\newcommand{\sigmaRc}{\accentset{R}{\sigma}}

\newcommand{\Lin}{\operatorname{Lin}}

\usepackage{xcolor}
\definecolor{darkred}{RGB}{180,0,0}

\begin{document}

\title{\textbf{Rational Mechanics of Material Strength \\in Brittle Solids
}}

\author[1,2]{Arash Yavari\thanks{Corresponding author, e-mail: arash.yavari@ce.gatech.edu}}
\author[1]{Aditya Kumar}
\affil[1]{\small \textit{School of Civil and Environmental Engineering, Georgia Institute of Technology, Atlanta, GA 30332, USA}}
\affil[2]{\small \textit{The George W. Woodruff School of Mechanical Engineering, Georgia Institute of Technology, Atlanta, GA 30332, USA}}

\maketitle

\begin{abstract}
Material strength is a classical concept that has recently found renewed applications in fracture mechanics, especially in models for crack nucleation in brittle solids. In this paper, we formulate material strength in the setting of finite elasticity and examine its geometric, constitutive, and symmetry-theoretic foundations. 
Spatial covariance requires the physical strength criterion to be independent of the stress measure used to represent it. We show that a criterion depending only on one stress measure generally becomes a function of both stress and strain when expressed using another stress measure. This motivates taking the pair $(\text{stress},\text{strain})$ as the natural domain of a representation-independent formulation. Classical stress-based strength criteria remain admissible as a special class in which the strain dependence is absent in the chosen representation.
We discuss the covariance of strength functions under arbitrary spatial diffeomorphisms and use this to relate representations in terms of the first Piola--Kirchhoff, second Piola--Kirchhoff, and Cauchy stresses. Restricting attention to the stress-based criteria that appear in the existing literature, we define the associated strength hypersurface as a subset of the constitutively admissible stress manifold, distinguish constitutive admissibility from fracture, and analyze the geometric and topological properties of the corresponding safe domain. 
We show that, for stress-based strength functions satisfying standard regularity conditions and the requirement that sufficiently large stresses are inadmissible, the strength surface is a smooth compact hypersurface of the constitutively admissible stress manifold.
We show that the safe domain is star-shaped under a natural proportional-reduction hypothesis. We then extend the formulation to anelastic brittle solids and examine the effects of residual stresses and anelastic distortions on material strength. We also discuss the action of material symmetry on strength functions for anisotropic solids. 
\end{abstract}

\begin{description}
\item[Keywords:] Material strength, brittle solids, finite elasticity, covariance, strength hypersurface, anisotropy, residual stress, anelasticity.
\end{description}

\tableofcontents

\section{Introduction}

Brittle solids offer in principle the simplest setting to formulate a theory of fracture. Brittle solids are those that are nominally elastic and only dissipate energy through creation of new surfaces \citep{Griffith1921}. Despite this apparent simplicity, developing a comprehensive theory that is predictive across all brittle materials and loading conditions has proved to be challenging. Brittle solids encompass not only classical linear elastic brittle materials, such as glass and ceramics, which fail at small strains, but also elastomers that can undergo large deformations before brittle failure. Brittle behavior is often regarded as an idealization, as microscopic crack-growth mechanisms typically introduce a degree of inelasticity prior to failure. However, the influence of such inelastic effects diminishes with increasing structural length scale. Moreover, brittle materials can exhibit measurable inelasticity under compressive loading. Nevertheless, a rigorous understanding of the idealized brittle state remains foundational to the development of fracture theories in more complex settings.

Experimental studies of brittle solids emphasize two canonical problems: (i) failure of the solids under uniform stress/strain, and (ii) crack growth and failure from a pre-existing large crack. The study of the former dates back at least to the time of Galileo Galilei (1564-1642) \citep{Yu2002}, but picked up steam in the 19th century with the experimental efforts of \citet{lame1833} and others. The first study of the latter problem is largely credited to \citet{Griffith1921}. These two experimental problems are used to define the two key material properties that dictate the fracture mechanics of brittle solids: strength and toughness. Out of these two properties, strength has proved to be both a harder property to measure and define in so much that a clear definition is rarely laid out and often implicitly misunderstood. The goal of this work is to present a more complete understanding of strength.

Strength has been defined in two subtly distinct ways as a material property. The conventional view treats it as the elastic limit in an elastic brittle material. Peak or maximum attainable stress is used to define this elastic limit in conventional hard brittle materials \citep{Yu2004}, whereas in elastomers, the peak strain or strain energy is often utilized \citep{hamdi2006, volokh2007, chen2017}, although many works report both peak stress and strain \citep{smith1964, smith1969}. 
The response remains elastic below this threshold and stress or strain exceeding it are deemed inadmissible. 
However, this definition is frequently conflated with a failure criterion, although the two are distinct, even in the absence of pre-existing cracks or sharp notches that would induce stress singularities.
To illustrate, consider a structure containing a U-shaped notch. The associated stress concentration depending on the radius of the notch causes the local elastic limit to be reached at relatively low applied loads; yet, the onset of a macroscopic crack may occur only at substantially higher loads.
Another illustrative example is the Brazilian fracture test, commonly used to estimate the tensile strength of rock-like materials \citep{li2013brazilian}. In this test, a thin circular disk without any pre-existing flaw is compressed between two rigid platens until it fractures. A widespread misinterpretation is that failure occurs when the maximum tensile stress reaches the material's tensile strength. However, more careful analyses reveal that, due to the non-uniformity of the stress field, fracture can occur at significantly higher loads \citep{KLDLP24}.
Thus, a point-wise violation of the strength limit does not, by itself, signify material failure, and certainly does not imply structural failure. 
In reality, the material is often presumed to undergo a softening process, wherein its load-carrying capacity degrades progressively until sufficient energy is available for crack formation. However, this post-peak softening behavior is generally not directly accessible from experiments. Consequently, defining strength purely as an elastic limit does not lead to a robust definition of a material property.

A second definition of strength is obtained by identifying it with the onset of failure. Experiments show that this occurs in homogeneous bodies subjected to spatially uniform stress states.  Motivated by this observation, \cite{KumarLopezPamies2020} and \citet{KumarBourdinFrancfortLopezPamies2020} defined strength as follows: ``the strength of an elastic brittle material is the set of all critical stresses at which the material fractures when it is subjected to a state of monotonically increasing, spatially uniform, but otherwise arbitrary stress." This is a more meaningful definition as it relates directly to experimental observations of failure. However, there are two aspects of this definition that need careful scrutiny. 

The first concerns the exclusive characterization of strength as a function of stress, rather than as a function of strain, strain energy, or a combined stress–strain description. Stress-based definitions of strength, represented in the form of a strength surface, are well accepted in the literature and are experimentally consistent \citep{sharma1965experimental, sato1987graphite}. However, as stated before, experimental practice in elastomers often favors strain- or energy-based measures. This preference is largely pragmatic: due to the pronounced strain-stiffening exhibited by elastomers at large deformations, it is often more reliable to identify a critical stretch or strain energy than the corresponding critical stress. Moreover, there is an outsized focus in experimental literature on the uniaxial tensile state. Nevertheless, \citet{KFLP18} and \citet{KumarBourdinFrancfortLopezPamies2020} presented a counterexample illustrating the limitations of strain- or energy-based definitions. In the case of cavitation in an incompressible elastomer subjected to uniform hydrostatic tension, the internal constraint prevents any increase in strain or strain energy, while the stress continues to grow until failure occurs (see Example~\ref{Ex:Hydro-Static-Tension}). Motivated by such examples, and by the success of stress-based strength theories in hard brittle materials, they advocated for a definition of strength based solely on stress. However, in finite elasticity, the choice of stress measure is not unique. Consequently, it is not immediately clear which stress measure is most appropriate for defining strength surfaces and whether purely stress-based definitions of strength are appropriate in finite elasticity. Clarifying this issue constitutes one of the central objectives of the present work.

The second aspect concerns the notion of spatially uniform stress. In linear elasticity, the meaning of this is usually unambiguous, but in finite elasticity one must distinguish carefully between different stress measures and between homogeneous stress and homogeneous deformation. In particular, a homogeneous deformation induces a spatially uniform stress field in a homogeneous compressible hyperelastic (or even Cauchy elastic) solid, but the converse is not true, in general. Thus, the phrase ``spatially uniform stress'' requires further clarification before it can be used as the basis for a mathematical definition of material strength. We discuss this further in this work.

The present study of strength is timely. Recent developments in computational fracture mechanics, in particular the phase-field method \citep{Bourdin00, Bourdin08, KumarBourdinFrancfortLopezPamies2020}, have provided unprecedented ability to simulate precise crack physics. However, these developments also underscore the need to reconcile strength-based and toughness-based descriptions of fracture within a unified framework and with it the need to understand strength as a material property.
Strength and toughness theories of fracture have seen parallel development over the last century. In many materials communities, strength has been effectively sidelined, to the extent that toughness-based approaches are often referred to as \emph{the fracture mechanics theory} \citep{anderson2005}, directly contrasted with strength-based viewpoints. This perspective is, in part, pragmatically justified: in some materials, failure is typically governed by pre-existing microscopic defects introduced during processing, rendering toughness the controlling parameter. Yet this is not universally applicable across all brittle materials. More fundamentally, any comprehensive theory of fracture must account for both the nucleation and propagation of cracks. This has been implicitly realized with the development of cohesive zone models \citep{barenblatt1959, dugdale1960} that account for strength. However, such approaches are often treated as distinct frameworks, rather than as models that bridge and unify strength and toughness. There have been efforts recently to understand strength from a bottom-up approach \citep{bonacci2026stochastic, lamont2025cohesive}. Strength is the macroscopic manifestation of the presence of microscopic defects, i.e., the ``weakest links" in the material. This makes strength inherently stochastic; recent work has explored incorporation of strength as a stochastic parameter in the phase field models \citep{zeng2025stochastic}.

\paragraph{Contributions of this paper.}
In this paper, we formulate material strength for brittle solids within the framework of finite elasticity and examine its geometric, constitutive, and symmetry-theoretic foundations. Spatial covariance requires the physical strength criterion to be independent of the stress measure used to represent it. We show that a criterion depending solely on one stress measure generally acquires dependence on the corresponding strain when expressed using another stress measure. The pair $(\text{stress},\text{strain})$ therefore provides the natural domain for a representation-independent formulation, while classical stress-based criteria remain admissible as a special class. The main contributions of this work can be summarized as follows:

\begin{itemize}[topsep=0pt,noitemsep, leftmargin=10pt]

\item We provide a covariant formulation of material strength in nonlinear elasticity and show that a criterion depending solely on one stress measure generally acquires strain dependence when expressed using another stress measure.

\item We define the strength hypersurface as a subset of the manifold of constitutively admissible stresses and distinguish clearly between constitutive admissibility and fracture.

\item We analyze the geometric and topological properties of the strength hypersurface and show that, under a proportional-reduction hypothesis, the associated safe domain is star-shaped.

\item For isotropic solids, we derive invariant representations of strength functions in terms of stress invariants and study the symmetry properties of the corresponding strength surfaces in principal stress space.

\item We extend the formulation to anelastic brittle solids and show that eigenstrains, through their modification of the material metric, alter the stress invariants and simultaneously affect the manifold of constitutively admissible stresses.

\item We analyze the action of material symmetry on strength functions for anisotropic solids and provide invariant representations for several symmetry classes.

\end{itemize}

This paper is organized as follows. In \S\ref{Sec:Elasticity}, we review the necessary background from nonlinear elasticity, including material symmetry, anisotropic elasticity, and constitutive equations of implicit elasticity. \S\ref{Sec:Material-Strength} formulates the notion of material strength in finite elasticity, discusses the choice of stress measure, loading conditions, the strength function and the associated strength hypersurface, and examines their topological properties, as well as the role of internal constraints. \S\ref{Material-Strength-Isotropic-Solids} specializes the discussion to isotropic solids and studies several classical isotropic strength surfaces and the geometry of their safe domains. \S\ref{Material-Strength-Residual-Stresses} discusses material strength in the presence of residual stresses and introduces the corresponding formulation in anelasticity. \S\ref{Material-Strength-Anisotropic-Solids} turns to anisotropic solids and discusses how material symmetry acts on strength functions. Finally, \S\ref{Sec:Conclusions} summarizes the main conclusions and comments on possible directions for future work.

\section{Nonlinear Elasticity} \label{Sec:Elasticity}

This section reviews the basic framework of nonlinear elasticity used throughout this work. Although much of this material is standard, we include it to establish notation, fix conventions, and make the paper reasonably self-contained. The covariance arguments developed later rely on this formulation and on the geometric interpretation of the deformation and stress measures. We begin with the kinematics of deformation, then review commonly used stress measures and material symmetry, and conclude with the framework of implicit elasticity.

\subsection{Kinematics}

Consider a body whose reference configuration is identified with an embedded submanifold $\mathcal{B}$ of a Euclidean ambient space $\mathcal{S}$. We denote by $\mathbf{g}$ the Euclidean metric of $\mathcal{S}$ and by $\mathbf{G}=\mathbf{g}\big|_{\mathcal{B}}$ the induced metric on the body in its undeformed configuration. A deformation is a map $\varphi:\mathcal{B}\to\mathcal{C}\subset\mathcal{S}$, where $\mathcal{C}=\varphi(\mathcal{B})$ is the current configuration. Its tangent map is the deformation gradient $\mathbf{F}=T\varphi$, which is independent of any metric structure. At each material point $X\in\mathcal{B}$, the deformation gradient is a linear map $\mathbf{F}(X):T_X\mathcal{B}\to T_{\varphi(X)}\mathcal{C}$ where $T_X\mathcal{B}$ and $T_{\varphi(X)}\mathcal{C}$ are the tangent spaces of $\mathcal{B}$ and $\mathcal{C}$ at $X\in\mathcal{B}$ and $\varphi(X)\in\mathcal{C}$, respectively.
If $\{X^A\}$ and $\{x^a\}$ are coordinate charts on $\mathcal{B}$ and $\mathcal{C}$, respectively, then the components of the deformation gradient are $F^a{}_A=\frac{\partial \varphi^a}{\partial X^A}$.
Deformation gradient has the following local coordinate representation
\begin{equation}
	\mathbf{F}=F^a{}_A\,\frac{\partial }{\partial x^a}\otimes dX^A\,.
\end{equation}
Its adjoint $\mathbf{F}^\star: T_{\varphi(X)}^*\mathcal{C}\to T_X^*\mathcal{B}$, where $T_X^*\mathcal{B}$ and $T_{\varphi(X)}^*\mathcal{C}$ are cotangent spaces of $\mathcal{B}$ and $\mathcal{C}$ at $X\in\mathcal{B}$ and $\varphi(X)\in\mathcal{C}$, respectively,\footnote{A cotangent space (space of covectors or $1$-forms) is dual to its corresponding tangent space.} is defined such that $\langle \mathbf{F}^\star\boldsymbol{\alpha},\mathbf{U} \rangle=\langle \boldsymbol{\alpha},\mathbf{F}\mathbf{U} \rangle$, $\forall \boldsymbol{\alpha}\in T_{\varphi(X)}^*\mathcal{C}\,,\mathbf{U}\in T_X\mathcal{B}$, where $\langle .,.\rangle$ is the natural pairing of covectors ($1$-forms) and vectors. The adjoint deformation gradient has the following local coordinate representation
\begin{equation}
	\mathbf{F}^\star=F^a{}_A\, dX^A\otimes \frac{\partial }{\partial x^a}\,.
\end{equation}
While deformation gradient maps a vector in the reference tangent space to a vector in the deformed tangent space, adjoint deformation gradient maps a covector in the deformed cotangent space to a covector in the reference cotangent space. It should be noted that the adjoint operator is metric independent.

The deformation gradient transpose $\mathbf{F}^{\mathsf{T}}: T_{\varphi(X)}\mathcal{C}\to T_X\mathcal{B}$ is defined such that $\llangle \mathbf{F}\mathbf{U},\mathbf{u} \rrangle_{\mathbf{g}}=\llangle \mathbf{U} , \mathbf{F}^{\mathsf{T}}\mathbf{u}\rrangle_{\mathbf{G}}$, $\forall \mathbf{U}\in T_X\mathcal{B}\,, \mathbf{u}\in T_{\varphi(X)}\mathcal{C}$, where $\llangle .,. \rrangle_{\mathbf{g}}$ and $\llangle .,. \rrangle_{\mathbf{G}}$ are the inner products induced by $\mathbf{g}$ and $\mathbf{G}$, respectively. Thus, $\mathbf{F}^{\mathsf{T}}=\mathbf{G}^\sharp \mathbf{F}^\star\,\mathbf{g}$ and in components are $(F^{\mathsf{T}})^A{}_{a}=G^{AB}\,g_{ab}\,F^b{}_B$.

The right Cauchy--Green strain tensor is defined as $\mathbf{C}=\mathbf{F}^{\mathsf{T}}\mathbf{F}$. In components, $C^A{}_B=(F^{\mathsf{T}})^A{}_a\,F^a{}_B$, and hence $C_{AB}=(g_{ab}\circ\varphi)\,F^a{}_A\,F^b{}_B$.
Thus, the covariant right Cauchy--Green tensor is simply the pull-back of the spatial metric, i.e., $\mathbf{C}^\flat=\varphi^*\mathbf{g}$, where $\flat$ denotes lowering indices with $\mathbf{G}$.
The left Cauchy--Green strain tensor is defined by $\mathbf{B}^{\sharp}=\varphi^*(\mathbf{g}^{\sharp})$. Its components are $B^{AB}=F^{-A}{}_a\,F^{-B}{}_b\,g^{ab}$, where $F^{-A}{}_a$ are the components of $\mathbf{F}^{-1}$.
The spatial counterparts of these tensors are obtained by pushing forward the reference metric and its inverse. In particular, the tensor $\mathbf{c}^\flat=\varphi_*\mathbf{G}=\mathbf{F}^{-\star} \mathbf{G}\mathbf{F}^{-1}$ has components $c_{ab}=F^{-A}{}_a\,F^{-B}{}_b\,G_{AB}$, while $\mathbf{b}^{\sharp}=\varphi_*(\mathbf{G}^{\sharp})$ has components $b^{ab}=F^a{}_A\,F^b{}_B\,G^{AB}$. Recall that $\mathbf{b}=\mathbf{c}^{-1}$.
The tensors $\mathbf{C}$ and $\mathbf{b}$ (with components $C^A{}_B$ and $b^a{}_b$) have the same principal invariants. These are \citep{Ogden1984,MarsdenHughes1983}
\begin{equation} \label{Principal-Invariants}
\begin{aligned}
	I_1&=\operatorname{tr}\mathbf{b}=b^{ab}g_{ab}
	=\lambda_1^2+\lambda_2^2+\lambda_3^2\,,\\
	I_2&=\frac{1}{2}\left(I_1^2-\operatorname{tr}\mathbf{b}^2\right)=\frac{1}{2}\left(I_1^2-b^{ab}b^{cd}g_{ac}g_{bd}\right)
	=\lambda_1^2\lambda_2^2+\lambda_2^2\lambda_3^2+\lambda_3^2\lambda_1^2\,,\\
	I_3&=\det\mathbf{b}=\lambda_1^2\lambda_2^2\lambda_3^2=J^2\,.
\end{aligned}
\end{equation}
where $\lambda_1, \lambda_2$, and $\lambda_3$ are the principal stretches.

Nanson's formula is written as \citep{MarsdenHughes1983}
\begin{equation} \label{Nanson}
	\mathbf{n}^{\flat}\,da=J\,\mathbf{F}^{-\star}\,\mathbf{N}^{\flat}\,dA\,,\qquad
	n_a\,da=J\,F^{-A}{}_a\,N_A\,dA
	\,,
\end{equation}
where the Jacobian of deformation is defined by $dv=J\,dV$ and is explicitly given as
\begin{equation} \label{Jacobian-Def}
	J=\sqrt{\frac{\det\mathbf{g}}{\det\mathbf{G}}}\,\det\mathbf{F}\,.
\end{equation}
Raising the spatial index, one obtains $n^a\,da=J\,g^{ab}\,F^{-B}{}_b\,G_{BA}\,N^A\,dA$. Equivalently, Nanson's formula can be written in vector form as
\begin{equation} \label{Nanson-Vector}
	\mathbf{n}\,da=J\,\mathbf{F}^{-\mathsf{T}}\,\mathbf{N}\,dA
	\,.
\end{equation}

The polar decomposition of the deformation gradient is written as
\begin{equation}\label{Polar-Decomposition}
	\mathbf{F}=\mathbf{R}\mathbf{U}=\mathbf{V}\mathbf{R}\,,
\end{equation}
where $\mathbf{U}$ and $\mathbf{V}$ denote the material and spatial stretch tensors, respectively, and should not be confused with the material velocity. The tensor $\mathbf{R}:T\mathcal{B}\to T\mathcal{C}$ is a $(\mathbf{G},\mathbf{g})$-orthogonal tensor field \citep{SimoMarsden1984}, i.e.,\footnote{Equivalently, one may write $\mathbf{G}^\sharp \mathbf{R}^\star (\mathbf{g}\circ\varphi)\,\mathbf{R}=\mathbf{R}^{\mathsf{T}}\mathbf{R}=\operatorname{id}_{T\mathcal{B}}$.}
\begin{equation} \label{G-g-Orthogonal}
	\mathbf{R}^\star (\mathbf{g}\circ\varphi)\,\mathbf{R}=\mathbf{G}\,.
\end{equation}
In components, this relation reads $R^a{}_A\,(g_{ab}\circ\varphi)\,R^b{}_B=G_{AB}$. The polar decomposition itself is written in components as $F^a{}_A=R^a{}_B\,U^B{}_A=V^a{}_b\,R^b{}_A$.
Equation \eqref{G-g-Orthogonal} implies that $(\det\mathbf{R})^2\det\mathbf{g}=\det\mathbf{G}$, while \eqref{Polar-Decomposition} implies that $\det\mathbf{U}=\det\mathbf{V}$. 
From \eqref{Jacobian-Def} one can see that $J=\det\mathbf{U} =\det\mathbf{V}$.
The material stretch tensor $\mathbf{U}:T_X\mathcal{B}\to T_X\mathcal{B}$ and the spatial stretch tensor $\mathbf{V}:T_x\mathcal{C}\to T_x\mathcal{C}$ are related to the right and left Cauchy--Green deformation tensors by
\begin{equation} \label{C-U-b_v}
\begin{aligned}
	\mathbf{C}
	&=\mathbf{F}^{\mathsf{T}}\mathbf{F}
	=(\mathbf{R}\mathbf{U})^{\mathsf{T}}\mathbf{R}\mathbf{U}
	=\mathbf{G}^\sharp(\mathbf{R}\mathbf{U})^\star\mathbf{g}\mathbf{R}\mathbf{U}
	=\mathbf{G}^\sharp\mathbf{U}^\star\mathbf{R}^\star\mathbf{g}\mathbf{R}\mathbf{U} 
	=\mathbf{G}^\sharp\mathbf{U}^\star\mathbf{G}\mathbf{U}
	=\mathbf{U}^2\,,\\
	\mathbf{b}
	&=\mathbf{F}\mathbf{F}^{\mathsf{T}}
	=\mathbf{V}\mathbf{R}(\mathbf{V}\mathbf{R})^{\mathsf{T}}
	=\mathbf{V}\mathbf{R}\mathbf{G}^\sharp(\mathbf{V}\mathbf{R})^\star\mathbf{g}
	=\mathbf{V}\mathbf{R}\mathbf{G}^\sharp\mathbf{R}^\star\mathbf{V}^\star\mathbf{g} 
	=\mathbf{V}\mathbf{g}^\sharp\mathbf{V}^\star\mathbf{g}
	=\mathbf{V}^2\,.
\end{aligned}
\end{equation}
Equivalently, one may write
\begin{equation}
	\mathbf{C}^\flat=\mathbf{U}^\star\mathbf{G}\mathbf{U}\,,\qquad
	\mathbf{b}^\sharp=\mathbf{V}\mathbf{g}^\sharp\mathbf{V}^\star\,.
\end{equation}
In components, these relations read $C_{AB}=U^M{}_A\,G_{MN}\,U^N{}_B$ and $b^{ab}=V^a{}_m\,g^{mn}\,V^b{}_n$. The relations \eqref{C-U-b_v} are commonly written as $\mathbf{U}=\sqrt{\mathbf{C}}$ and $\mathbf{V}=\sqrt{\mathbf{b}}$.
It is straightforward to show that $U_{AB}=G_{AM} U^M{}_B=U_{BA}$.

Let $\{\N,\NN, \NNN\}$ be a $\mathbf{G}$-orthonormal eigenbasis of $\mathbf{U}$, with corresponding principal stretches $\lambda_j$, $j=1,2,3$. Using the spectral decomposition of $\mathbf{C}$ and $\mathbf{U}$, one has \citep{Ogden1984}
\begin{equation} \label{C-Spectral}
	\mathbf{C}^\sharp= \lambda_1^2 \,\N \otimes\N+\lambda_2^2 \,\NN \otimes\NN
	+\lambda_3^2 \,\NNN \otimes\NNN
	\,,\qquad
	\mathbf{U}^\sharp= \lambda_1 \,\N \otimes\N+\lambda_2 \,\NN \otimes\NN
	+\lambda_3 \,\NNN \otimes\NNN\,,
\end{equation}
where $\lambda_1,\lambda_2,\lambda_3$ are the principal stretches and $\N,\NN,\NNN$ are the corresponding principal directions. 
Recall that $\N\otimes\N+\NN\otimes\NN+\NNN\otimes\NNN=\mathbf{G}^\sharp$.
The representation \eqref{C-Spectral}$_2$ is equivalent to
\begin{equation}
	\mathbf{U}= \lambda_1 \,\N \otimes\N^\flat+\lambda_2 \,\NN \otimes\NN^\flat+\lambda_3 \,\NNN \otimes\NNN^\flat\,,
\end{equation}
and therefore
\begin{equation}
	\mathbf{U}\Nj=\lambda_j\Nj\,,\qquad\text{~(no summation~on~}j)\,.
\end{equation}
Using the polar decomposition $\mathbf{F}=\mathbf{R}\mathbf{U}=\mathbf{V}\mathbf{R}$, one obtains
\begin{equation}
	\mathbf{F}\Nj=\mathbf{R}\mathbf{U}\Nj=\lambda_j \mathbf{R}\Nj=\mathbf{V}\mathbf{R}\Nj
	\,,\qquad j=1,2,3\,.
\end{equation}
Thus, if $\{\n,\nn,\nnn\}$ is the eigenbasis of $\mathbf{V}$, then \citep{Ogden1984}
\begin{equation}
	\nj=\mathbf{R}\Nj\,,\qquad j=1,2,3\,.
\end{equation}
Knowing that the eigenvalues of $\mathbf{b}$ and $\mathbf{V}$ are $\lambda_j^2$ and $\lambda_j$, respectively, one concludes that the Finger and spatial stretch tensors admit the spectral representations
\begin{equation}
	\mathbf{b}^\sharp= 
	\lambda_1^2 \,\n \otimes\n+\lambda_2^2 \,\nn \otimes\nn+\lambda_3^2 \,\nnn \otimes\nnn\,,\qquad
	\mathbf{V}^\sharp= 
	\lambda_1 \,\n \otimes\n+\lambda_2 \,\nn \otimes\nn+\lambda_3 \,\nnn \otimes\nnn
	\,.
\end{equation}
Note that
\begin{equation} 
	\sum_{j=1}^3 \nj\otimes \nj = \sum_{j=1}^3 \mathbf{R}\Nj\otimes \mathbf{R}\Nj
	= \mathbf{R} \Big(\sum_{j=1}^3 \Nj\otimes \Nj\Big) \mathbf{R}^\star
	= \mathbf{R} \mathbf{G}^\sharp \mathbf{R}^\star =\mathbf{g}^\sharp \,.
\end{equation}

Suppose $f:\mathbb{R}\to\mathbb{R}$ is a smooth monotone function such that $f(1)=0$ and $f'(1)=1$. Hill's strain measures are defined by applying $f$ to the principal stretches of $\mathbf{V}$, namely \citep{Hill1968,Hill1970,Hill1978}
\begin{equation}
	f(\mathbf{V}^\sharp)= 
	f(\lambda_1) \,\n \otimes\n+ f(\lambda_2) \,\nn \otimes\nn+f(\lambda_3) \,\nnn \otimes\nnn
	\,.
\end{equation}
A particularly important example is Hencky's logarithmic strain \citep{Neff2016}, which is given by
\begin{equation}
	\mathbf{h}^\sharp= \log\mathbf{V}^\sharp= 
	\log \lambda_1 \,\n \otimes\n+\log\lambda_2 \,\nn \otimes\nn+\log\lambda_3 \,\nnn \otimes\nnn
	\,.
\end{equation}
The material analogue of Hencky's logarithmic strain is defined using the spectral decomposition of $\mathbf{U}$. Using \eqref{C-Spectral}$_2$, one defines
\begin{equation}
	\log \mathbf{U}^\sharp=
	\log \lambda_1 \,\N \otimes\N+\log \lambda_2 \,\NN \otimes\NN+\log \lambda_3 \,\NNN \otimes\NNN
	\,.
\end{equation}
Equivalently, one may write
\begin{equation}
	\log \mathbf{U}
	=\log \lambda_1 \,\N \otimes\N^\flat+\log \lambda_2 \,\NN \otimes\NN^\flat+\log \lambda_3 \,\NNN \otimes\NNN^\flat\,.
\end{equation}
Thus, $\log \mathbf{U}$ has the same eigenvectors as $\mathbf{U}$, and its eigenvalues are $\log \lambda_j$, $j=1,2,3$.

\subsection{Measures of stress}

We briefly discuss a few commonly used stress tensors in nonlinear elasticity. Although continuum mechanics admits infinitely many stress measures, the tensors reviewed here are among the most useful in applications.

Consider an area element $da$ in the deformed configuration $\mathcal{C}$ with $\mathbf{g}$-unit normal $\mathbf{n}$, i.e., $\llangle \mathbf{n},\mathbf{n} \rrangle_{\mathbf{g}}=1$. The corresponding traction vector is defined by $\mathbf{t}=\boldsymbol{\sigma}\mathbf{n}^\flat$, where $\boldsymbol{\sigma}$ is the Cauchy stress and $\mathbf{n}^\flat=\mathbf{g}\mathbf{n}$. Hence, the force acting on this area element is $\mathbf{f}=\mathbf{t}\,da$. In components, one has $t^a=\sigma^{ab}n_b$, where $n_b=g_{bc}n^c$.
Let $dA$ be the corresponding area element in the reference configuration $\mathcal{B}$ with $\mathbf{G}$-unit normal $\mathbf{N}$, i.e., $\llangle \mathbf{N},\mathbf{N} \rrangle_{\mathbf{G}}=1$. The first Piola--Kirchhoff stress tensor $\mathbf{P}$ is defined by requiring that
\begin{equation}
	\mathbf{t}\,da = \mathbf{P}\,\mathbf{N}^\flat\,dA\,.
\end{equation}
Using Nanson's formula \eqref{Nanson}, one obtains
\begin{equation}
	\mathbf{P} = J\,\boldsymbol{\sigma}\,\mathbf{F}^{-\star}\,.
\end{equation}
In components, this reads $P^{aA}=J\,\sigma^{ab}\,F^{-A}{}_b$.

Pulling back the force $\mathbf{f}$ to the reference configuration, one next defines the second Piola--Kirchhoff stress tensor $\mathbf{S}$ by
$\mathbf{F}^{-1}\mathbf{t}\,da=\mathbf{S}\,\mathbf{N}^\flat\,dA$. Thus,
\begin{equation}
	\mathbf{S} = \mathbf{F}^{-1}\mathbf{P} = J\,\mathbf{F}^{-1}\,\boldsymbol{\sigma}\,\mathbf{F}^{-\star}\,.
\end{equation}
In components, one has $S^{AB}=F^{-A}{}_a\,P^{aB}=J\,F^{-A}{}_a\,\sigma^{ab}\,F^{-B}{}_b$. The balance of angular momentum implies that $\mathbf{S}=\mathbf{S}^\star$ or in components $S^{AB}=S^{BA}$. 

The Kirchhoff stress is defined as $\boldsymbol{\tau}=J\,\boldsymbol{\sigma}$, while the convected stress is defined as $\boldsymbol{\Sigma}=\varphi_t^*\boldsymbol{\sigma}=\mathbf{F}^{-1}\boldsymbol{\sigma}\mathbf{F}^{-\star}=J^{-1}\mathbf{S}$.
The rotated stress tensor is defined as \citep{GreenNaghdi1965}
\begin{equation}
	\sigmaR=\mathbf{R}^{*}\boldsymbol{\sigma}
	=\mathbf{R}^{-1}\boldsymbol{\sigma}\mathbf{R}^{-\star}\,,
\end{equation}
i.e., the pull-back of the Cauchy stress by the orthogonal tensor $\mathbf{R}$. In components, this reads $\sigmaRc^{AB}=R^A{}_a\,\sigma^{ab}\,R^B{}_b$.

Recall that $\mathbf{R}:T_X\mathcal{B}\to T_x\mathcal{C}$ is a $(\mathbf{G},\mathbf{g})$-orthogonal tensor field with components $R^a{}_A$. The components of its inverse $\mathbf{R}^{-1}:T_x\mathcal{C}\to T_X\mathcal{B}$ are denoted by $R^A{}_a$ and are given by $R^A{}_a=G^{AM}\,R^m{}_M\,g_{ma}$.

The Biot stress $\SOne$ is the material stress tensor that is work-conjugate to the material stretch tensor $\mathbf{U}$. Note that 
\begin{equation}
	\frac{1}{2}\mathbf{S}\!:\!\dot{\mathbf{C}}^\flat
	=\frac{1}{2} \left( \mathbf{S}\!:\!\mathbf{G}\dot{\mathbf{U}}\mathbf{U} 
	+  \mathbf{S}\!:\!\mathbf{G}\mathbf{U}\dot{\mathbf{U}} \right)
	=\frac{1}{2} \left( \mathbf{U}\mathbf{S} + \mathbf{S}\mathbf{U}^\star \right)\!:\!\dot{\mathbf{U}}^\flat
	=\SOne\!:\!\dot{\mathbf{U}}^\flat\,,
\end{equation}
where 
\begin{equation}
\begin{aligned}
	& \SOne =\operatorname{sym}\!\left(\mathbf{U}\mathbf{S}\right)
	=\frac{1}{2}\Big(\mathbf{U}\mathbf{S}+(\mathbf{U}\mathbf{S})^\star\Big)
	=\frac{1}{2}\left(\mathbf{U}\mathbf{S}+\mathbf{S}^\star\mathbf{U}^\star\right)
	=\frac{1}{2}\left(\mathbf{U}\mathbf{S}+\mathbf{S}\mathbf{U}^\star\right)\,,\\
	& \SOneC^{AB} =\frac{1}{2}\left(U^A{}_M\,S^{MB}+S^{AM}\,U^B{}_M\right)\,.
\end{aligned}
\end{equation}
Using the polar decomposition $\mathbf{F}=\mathbf{R}\mathbf{U}$ and the relation $\mathbf{S}=\mathbf{F}^{-1}\mathbf{P}$, one concludes that $\mathbf{U}\mathbf{S}=\mathbf{U}\mathbf{F}^{-1}\mathbf{P}=\mathbf{R}^{-1}\mathbf{P}$.
Therefore,
\begin{equation}
	\SOne
	=\frac{1}{2}\left(\mathbf{R}^{-1}\mathbf{P}+\mathbf{P}^\star\mathbf{R}^{-\star}\right)
	=\operatorname{sym}\!\left(\mathbf{R}^{-1}\mathbf{P}\right)\,,\qquad
	\SOneC^{AB}
	=\frac{1}{2}\left(R^A{}_a\,P^{aB}+R^B{}_a\,P^{aA}\right)\,.
\end{equation}

\paragraph{Eigenvalues of Biot stress for isotropic solids.}
Let $\{\mathbf{E}_A\}$ be a $\mathbf{G}$-orthonormal eigenbasis of $\mathbf{U}$, i.e.,
\begin{equation}
	\mathbf{U}\mathbf{E}_A=\lambda_A\,\mathbf{E}_A\,,
	\quad\text{~(no~summation~on~}A)\,,\qquad
	\llangle \mathbf{E}_A,\mathbf{E}_B \rrangle_{\mathbf{G}}=\delta_{AB}\,.
\end{equation}
For an isotropic solid $\mathbf{U}$ and $\mathbf{S}$ are coaxial, i.e., share the same eigenvectors; this property does not hold in general for anisotropic solids. Hence,
\begin{equation}
	\mathbf{S}\mathbf{E}_A=S_A\,\mathbf{E}_A\,,
	\quad\text{~(no~summation~on~}A)\,.
\end{equation}
Knowing that $\mathbf{U}$ is $\mathbf{G}$-symmetric, i.e., $U_{AB}=U_{BA}$, one has
\begin{equation}
	\llangle \mathbf{U}\mathbf{X},\mathbf{Y}\rrangle_{\mathbf{G}}
	=\llangle \mathbf{X},\mathbf{U}\mathbf{Y}\rrangle_{\mathbf{G}}\,,
	\qquad \forall\,\mathbf{X},\mathbf{Y}\in T_X\mathcal{B}\,.
\end{equation}
Hence,
\begin{equation}
\begin{aligned}
	\llangle \mathbf{S}\mathbf{U}^\star \mathbf{E}_A,\mathbf{E}_B\rrangle_{\mathbf{G}}
	&=\llangle \mathbf{U}^\star \mathbf{E}_A,\mathbf{S}\mathbf{E}_B\rrangle_{\mathbf{G}}
	=S_B\,\llangle \mathbf{U}^\star \mathbf{E}_A,\mathbf{E}_B\rrangle_{\mathbf{G}}
	=S_B\,\llangle \mathbf{E}_A,\mathbf{U}\mathbf{E}_B\rrangle_{\mathbf{G}}\\
	&=S_B\lambda_B\,\llangle \mathbf{E}_A,\mathbf{E}_B\rrangle_{\mathbf{G}}
	=S_B\lambda_B\,\delta_{AB}\,,
	\quad\text{~(no~summation~on~}A~\text{or~}B)\,.
\end{aligned}
\end{equation}
Therefore,
\begin{equation}
	\mathbf{S}\mathbf{U}^\star \mathbf{E}_A=\lambda_A S_A\,\mathbf{E}_A\,,
	\quad\text{~(no~summation~on~}A)\,.
\end{equation}
Since also
\begin{equation}
	\mathbf{U}\mathbf{S}\mathbf{E}_A =\lambda_A S_A\,\mathbf{E}_A\,,
	\quad\text{~(no~summation~on~}A)\,,
\end{equation}
it follows that
\begin{equation}
	\SOne \mathbf{E}_A
	=\frac{1}{2}(\mathbf{U}\mathbf{S}+\mathbf{S}\mathbf{U}^\star)\mathbf{E}_A
	=\lambda_A S_A\,\mathbf{E}_A\,, \quad\text{~(no~summation~on~}A)\,.
\end{equation}
Hence, $\mathbf{U}$, $\mathbf{S}$, and $\SOne$ are coaxial and share the same eigenvectors, and the eigenvalues of the Biot stress are given by
\begin{equation}
	\beta_A=\lambda_A S_A\,,\quad\text{~(no~summation~on~}A)\,.
\end{equation}
Moreover, the eigenvalues of the Cauchy stress are related to those of the Biot stress as follows. Using $\boldsymbol{\sigma}=J^{-1}\,\mathbf{F}\mathbf{S}\mathbf{F}^\star=J^{-1}\,\mathbf{R}\mathbf{U}\mathbf{S}\mathbf{U}^\star\mathbf{R}^\star$, one sees that $\boldsymbol{\sigma}$ and $J^{-1}\mathbf{U}\mathbf{S}\mathbf{U}^\star$ have the same eigenvalues. Since
\begin{equation}
	\mathbf{U}\mathbf{S}\mathbf{U}^\star \mathbf{E}_A 
	=\lambda_A^2 S_A\,\mathbf{E}_A\,,
	\quad\text{~(no~summation~on~}A)\,,
\end{equation}
it follows that
\begin{equation} \label{Eigenvalues-sigma-B-S-iso}
	\sigma_A=J^{-1}\lambda_A^2 S_A =J^{-1}\lambda_A\beta_A\,,
	\quad\text{~(no~summation~on~}A)\,.
\end{equation}

Another useful stress measure is the Mandel stress $\mathbf{M}=\mathbf{C}\mathbf{S}$. It is work-conjugate to the material Hencky strain $\log\mathbf{U}$ \citep{Hoger1987,Xiao1997} in the sense that the stress power may be written as $\mathbf{M}\!:\!\mathring{\overline{\log\mathbf{U}}}$, where for any material tensor $\mathbf{A}$ the corotational rate is defined by $\mathring{\mathbf{A}}\coloneqq\dot{\mathbf{A}}+\boldsymbol{\Omega}\mathbf{A}-\mathbf{A}\boldsymbol{\Omega}$, and $\boldsymbol{\Omega}=\mathbf{R}^{-1}\dot{\mathbf{R}}$ is the material spin associated with the rotation tensor $\mathbf{R}$. Although the Mandel stress is widely used in finite-strain plasticity, we will not use it in this work.

\subsection{Material symmetry group}

For a hyperelastic solid, consider a stored energy density of the form $W=W(X,\mathbf{F},\mathbf{G},\mathbf{g})$, where $\mathbf{g}$ is the metric of the Euclidean ambient space and $\mathbf{G}$ is the induced metric on the body in the reference configuration. In the absence of eigenstrains, $\mathbf{G}$ is the material metric.

The material symmetry group at a point $X\in\mathcal{B}$, relative to the Euclidean reference configuration $(\mathcal{B},\mathbf{G})$, is the subgroup $\mathcal{G}_X\leqslant \mathrm{Orth}(\mathbf{G})$ consisting of all linear maps $\mathbf{K}:T_X\mathcal{B}\to T_X\mathcal{B}$ such that
\begin{equation} \label{W-Material-Symmetry}
	W(X,\mathbf{F}\mathbf{K},\mathbf{G},\mathbf{g})
	= W(X,\mathbf{F},\mathbf{G},\mathbf{g})\,,\qquad
	\forall\,\mathbf{F}\,,~
	\forall\,\mathbf{K}\in\mathcal{G}_X\leqslant \mathrm{Orth}(\mathbf{G})\,,
\end{equation}
where
\begin{equation}
	\mathrm{Orth}(\mathbf{G})
	=
	\left\{
	\mathbf{Q}:T_X\mathcal{B}\to T_X\mathcal{B}\ \big|\ 
	\mathbf{Q}^{\star}\mathbf{G}\mathbf{Q}=\mathbf{G}
	\right\}\,.
\end{equation}
Thus, the stored energy is invariant under the action of $\mathcal{G}_X$.

\paragraph{Isotropic solids.}
For an isotropic solid $\mathcal{G}_X=\mathrm{Orth}(\mathbf{G})$, the stored energy depends only on the principal invariants $I_1,I_2,I_3$, and hence $W=W(I_1,I_2,I_3)$.

\subsection{Anisotropic elasticity} \label{Sec:Anisotropic-Elasticity}

For a hyperelastic anisotropic solid, the stored energy density per unit reference volume is written as
\begin{equation}\label{eneg}
	W=\hat{W}(\mathbf{C}^\flat,\mathbf{G},\boldsymbol{\zeta}_1,\dots,\boldsymbol{\zeta}_n)\,,
\end{equation}
where $\boldsymbol{\zeta}_i$, $i=1,\dots,n$, are structural tensors characterizing the material symmetry group. The role of the structural tensors is to represent the energy as an isotropic scalar-valued function of its arguments. By Hilbert's theorem, for any finite collection of tensors there exists a finite integrity basis for the algebra of isotropic invariants generated by that collection. Therefore, if $I_j$, $j=1,\dots,m$, is an integrity basis associated with the arguments in \eqref{eneg}, then the energy may be written as $W=W(X,I_1,\dots,I_m)$. The Doyle--Ericksen formula \citep{Doyle1956,MarsdenHughes1983,Yavari2006} then gives us
\begin{equation}\label{invse1}
	\mathbf{S}=2\frac{\partial \hat{W}}{\partial\mathbf{C}^\flat}
	=\sum_{j=1}^{m}2W_j\frac{\partial I_j}{\partial\mathbf{C}^\flat}\,,\qquad
	W_j=\frac{\partial W}{\partial I_j}\,,\qquad j=1,\dots,m\,.
\end{equation}

\paragraph{Transversely isotropic solids.}
A transversely isotropic solid has, at each material point, a single preferred direction orthogonal to the plane of isotropy. Let $\mathbf{N}(X)$ denote this preferred direction at $X\in\mathcal{B}$. A choice for the structural tensor is $\mathbf{A}=\mathbf{N}\otimes\mathbf{N}$, and the energy is written as $W=W(\mathbf{G},\mathbf{C}^\flat,\mathbf{A})$ \citep{Doyle1956,Spencer1982,Lu2000}. In this case the energy depends on five independent invariants,
\begin{equation} 
	I_1=\mathrm{tr}\,\mathbf{C}\,,\qquad 
	I_2=\mathrm{det}\,\mathbf{C}\,\mathrm{tr}\,\mathbf{C}^{-1}\,,\quad 
	I_3=\mathrm{det}\,\mathbf{C}\,,\qquad 
	I_4=\mathbf{N}\cdot\mathbf{C}\cdot\mathbf{N}\,,\qquad 
	I_5=\mathbf{N}\cdot\mathbf{C}^2\cdot\mathbf{N}\,,
\end{equation}
i.e., $W=W(I_1,\hdots,I_5)$.

\paragraph{Orthotropic solids.}
An orthotropic solid possesses reflection symmetry with respect to three mutually orthogonal planes. Let $\mathbf{N}_1(X)$, $\mathbf{N}_2(X)$, and $\mathbf{N}_3(X)$ be $\mathbf{G}$-orthonormal vectors specifying the orthotropic axes at $X$. One may choose the structural tensors $\mathbf{A}_1=\mathbf{N}_1\otimes\mathbf{N}_1$, $\mathbf{A}_2=\mathbf{N}_2\otimes\mathbf{N}_2$, and $\mathbf{A}_3=\mathbf{N}_3\otimes\mathbf{N}_3$. Since $\mathbf{A}_1+\mathbf{A}_2+\mathbf{A}_3=\mathbf{I}=\operatorname{id}_{T_X\mathcal{B}}$, only two of these are independent. Accordingly, one writes $W=W(\mathbf{G},\mathbf{C}^\flat,\mathbf{A}_1,\mathbf{A}_2)$ \citep{Doyle1956,Spencer1982,Lu2000}. A convenient integrity basis is given by the following seven invariants
\begin{equation} \label{Orthotropic-Invariants}
\begin{aligned}
	& I_1=\mathrm{tr}\,\mathbf{C}\,,&& I_2=\mathrm{det}\,\mathbf{C}\,\mathrm{tr}\,\mathbf{C}^{-1}\,,
	&&  I_3=\mathrm{det}\,\mathbf{C}\,,\\
	& I_4=\mathbf{N}_1\cdot\mathbf{C}\cdot\mathbf{N}_1\,,&&
	I_5=\mathbf{N}_1\cdot\mathbf{C}^2\cdot\mathbf{N}_1\,,&&
	I_6=\mathbf{N}_2\cdot\mathbf{C}\cdot\mathbf{N}_2\,,&&
	I_7=\mathbf{N}_2\cdot\mathbf{C}^2\cdot\mathbf{N}_2\,.
\end{aligned}
\end{equation}
Thus, $W=W(I_1,\hdots,I_7)$.

\paragraph{Monoclinic solids.}
A monoclinic solid is characterized by three material preferred directions represented by three unit vectors $\{\mathbf{N}_1,\mathbf{N}_2,\mathbf{N}_3\}$, where $\mathbf{N}_1\cdot\mathbf{N}_2\neq 0$ and $\mathbf{N}_3$ is normal to the plane spanned by $\mathbf{N}_1$ and $\mathbf{N}_2$ \citep{merodio2020finite}. The energy function of a monoclinic solid depends on nine invariants \citep{Spencer1986}. Seven of these coincide with the orthotropic invariants in \eqref{Orthotropic-Invariants}, while the remaining two are
\begin{equation}
	I_8=\mathfrak{g}\,\mathbf{N}_1\cdot\mathbf{C}\cdot\mathbf{N}_2,\qquad
	I_9=\mathfrak{g}^2\,,
\end{equation}  
where $\mathfrak{g}=\mathbf{N}_1\cdot\mathbf{N}_2$. Thus, $W=W(I_1,\hdots,I_9)$.

\subsection{Constitutive equations of implicit elasticity}

In the literature of elasticity it is often tacitly assumed that stress is given explicitly as a function of strain, as in Cauchy elasticity \citep{Cauchy1828,YavariGoriely2025} and hyperelasticity \citep{Truesdell1952} (Green elasticity \citep{Green1838,Green1839,Spencer2015}). In these theories one writes constitutive equations of the form $\boldsymbol{\sigma}=\hat{\boldsymbol{\sigma}}(\mathbf{b},\mathbf{g})$. However, this is only one possible constitutive setting. One may also consider elastic materials for which strain is determined by stress.
In his work on \emph{controllable states of stress}, \citet{Carroll1973} considered constitutive laws of the form
\begin{equation} \label{invertible-stress}
	\mathbf{b}^\sharp=\xi_0\,\mathbf{g}^\sharp+\xi_1 \boldsymbol{\sigma}
	+\xi_2 \boldsymbol{\sigma}^2\,.
\end{equation}
Here, the scalar response functions $\xi_0$, $\xi_1$, and $\xi_2$ depend on the principal invariants of the Cauchy stress. Carroll defined a stress field to be controllable if it satisfies the equilibrium equations in the absence of body forces and if the associated strain is compatible for arbitrary choices of the response functions $\xi_0$, $\xi_1$, and $\xi_2$. He proved that, for isotropic elastic solids with constitutive equations of the form \eqref{invertible-stress}, every controllable stress field is necessarily homogeneous.\footnote{See \citet{Carroll1973ControllableIncompressible} for similar discussions for incompressible solids.} At the same time, he observed that, for a given material in this class, not every homogeneous stress field is constitutively admissible as its corresponding $\mathbf{b}^\sharp$ may not be compatible.

More generally, one may consider elastic materials whose constitutive response is described by an implicit relation between stress and strain measures. This broader class is usually referred to as implicit elasticity \citep{Morgan1966,Rajagopal2003,Rajagopal2007}. In such theories, the constitutive equation is not written in the explicit form $\boldsymbol{\sigma}=\hat{\boldsymbol{\sigma}}(\mathbf{b},\mathbf{g})$ or $\mathbf{b}=\hat{\mathbf{b}}(\boldsymbol{\sigma},\mathbf{g})$, but rather as an implicit tensor equation of the following form
\begin{equation} \label{ImplicitStressStrain}
	\boldsymbol{\mathsf{f}}(\boldsymbol{\sigma},\mathbf{b},\mathbf{g})=\mathbf{0}\,.
\end{equation}
Cauchy elastic and Green elastic solids are special cases of this more general framework.

\begin{remark}
This observation is important for the notion of material strength. If strength is defined as the set of critical homogeneous stresses at which fracture nucleates (see \S\ref{Sec:Material-Strength}), then its definition does not require stress to be an explicit function of strain. What is required instead is that the stress state be constitutively admissible. In an implicit theory this means that the pair $(\boldsymbol{\sigma},\mathbf{b})$ must lie in the constitutive stress-strain manifold. Thus, the notion of material strength extends naturally beyond explicit constitutive theories. At the same time, this suggests that a general notion of strength should not be restricted to purely stress-based descriptions, since fracture is expected to depend on both stress and strain. In the following sections we will see that covariance considerations naturally lead to strength functions that depend on both stress and strain measures.
\end{remark}

It is important to keep in mind that, in elasticity, stress and strain are on the same footing. In other words, it does not make sense to say that strain causes stress or that stress causes strain. Constitutive equations simply specify a relation between appropriate measures of stress and strain, whether written explicitly or implicitly. This viewpoint, together with the consequences of spatial covariance, also motivates our formulation of material strength. In general, strength need not be stress-based. Rather, strength is controlled by the pair $(\text{stress},\text{strain})$.

\begin{example}[Rajagopal bar]
Motivated by the examples provided by \citet{Rajagopal2003}, consider a one-dimensional body whose response is described by an elastic bar connected in parallel with an inextensible string. Let $\lambda$ denote the stretch, and let $\hat{\sigma}(\lambda)$ be the elastic response of the bar. Let $\lambda_c>1$ denote the limiting stretch at which the inextensible string becomes taut, and define $\sigma_c=\hat{\sigma}(\lambda_c)$. The constitutive relation is written as (we assume that for $\lambda<1$ the body deforms elastically, neglecting buckling)
\begin{equation}
	\begin{cases}
		\sigma=\hat{\sigma}(\lambda)\,, & \lambda<\lambda_c\,,\\[4pt]
		\lambda=\lambda_c\,, & \sigma\ge\sigma_c\,.
	\end{cases}
\end{equation}
An equivalent implicit representation of this response is
\begin{equation}
	(\lambda-\lambda_c)\big(\sigma-\hat{\sigma}(\lambda) \big)=0\,,\qquad
	\lambda\le\lambda_c\,,\quad	\sigma\ge\sigma_c\,.
\end{equation}
For $\lambda<\lambda_c$ the string is slack, and the response is governed solely by the elastic relation $\sigma=\hat{\sigma}(\lambda)$. Once the stretch reaches $\lambda_c$, the string becomes taut and enforces the constraint $\lambda=\lambda_c$. In this regime, the stretch remains fixed while the stress can increase or change arbitrarily, subject only to $\sigma\ge\sigma_c$. Thus, at $\lambda=\lambda_c$ there is a single admissible stretch but infinitely many admissible stresses. This response is non-dissipative, but the stress cannot be expressed globally as a single-valued function of the stretch. Hence, it does not fall within the classical framework of Cauchy elasticity or Green elasticity, and is instead viewed as belonging to the broader class of implicit constitutive theories.
Upon unloading from a state with $\lambda=\lambda_c$ and $\sigma>\sigma_c$, the response is uniquely determined by the requirement of non-dissipation. The body first unloads along the vertical branch $\lambda=\lambda_c$ until $\sigma=\sigma_c$. At this point the inextensible string becomes slack, and the constraint is no longer active. Further unloading follows the elastic relation $\sigma=\hat{\sigma}(\lambda)$. Thus, the loading and unloading paths coincide, and the response is fully reversible with no hysteresis. This highlights that, although the stress-stretch relation is not single-valued, the admissible path in any process is uniquely determined by the requirement that the elastic branch be followed whenever it is accessible.
\end{example}

\subsection{Covariance, objectivity, and material symmetry in continuum mechanics} \label{Sec:Elasticity-Covariance}

In this section we review the notion of covariance in continuum mechanics and nonlinear elasticity. In later sections, we will use these ideas to formulate a covariant theory of material strength. As we will see in \S\ref{Strength-Hypersurface}, covariance provides the framework for relating strength functions corresponding to different measures of stress.

The modern notion of covariance in physics is most closely associated with Einstein's theory of relativity, where the requirement that the governing equations of a given field theory preserve their form under a prescribed class of transformations became a central organizing principle. In particular, general covariance emerged in the development of general relativity as invariance under arbitrary smooth changes of coordinates, although the underlying mathematical ideas trace back to nineteenth-century differential geometry and tensor calculus. The historical and conceptual status of covariance, especially its relation to symmetry and background independence, has been discussed extensively in the literature \citep{Norton1993,Giulini2007}.

The notion of covariance in continuum mechanics appears in several closely related contexts involving invariance under transformation groups acting on the spatial and material configurations. Three important manifestations are the following.
\begin{itemize}[topsep=0pt,noitemsep, leftmargin=10pt]
\item When the ambient space is Euclidean, constitutive equations are postulated to be invariant under isometries of Euclidean space, a requirement usually referred to as objectivity or material frame indifference \citep{TruesdellNoll2004}. For a non-flat ambient space, however, there may be no nontrivial isometries. In this case, the natural extension of material frame indifference is the spatial covariance of the energy function, namely invariance under arbitrary diffeomorphisms of the ambient space \citep{HuMa1977,MarsdenHughes1983}.

\item In the classical setting of Euclidean ambient spaces, \citet{Green64} showed that the balance laws of nonlinear elasticity may be derived by postulating the balance of energy together with its invariance under superposed rigid body motions of the ambient space, while \citet{Noll1963} gave a related passive interpretation in terms of time-dependent changes of spatial frame. This idea was later extended to hyperelasticity in Riemannian ambient spaces by \citet{HuMa1977}, who postulated covariance of the energy balance under arbitrary ambient space diffeomorphisms and showed that this gives the balance laws of hyperelasticity together with the Doyle--Ericksen formula. The covariance-based formulation was later extended to Cauchy elasticity \citep{YavariGoriely2025}; see also \citep{MarsdenHughes1983,Yavari2006,YavariMarsden2012}.

\item Material covariance of constitutive equations is a generalization of material isotropy. In the classical Euclidean setting, isotropy is understood as invariance of the constitutive equation under isometries of the material manifold. For a general Riemannian material manifold, however, there may be no nontrivial isometries, and this classical characterization is no longer adequate. In that case, the natural generalization of isotropy is covariance of the constitutive equation under arbitrary diffeomorphisms of the material manifold \citep{MarsdenHughes1983}. See also \citep{Lu2000,Lu2012,YavariSozio2023} for extensions of material covariance to anisotropic elasticity and anelasticity.
\end{itemize}
\smallskip

A diffeomorphism is a map $\varphi:\mathcal{B}\to\mathcal{C}$ that is smooth, invertible, and whose inverse $\varphi^{-1}:\mathcal{C}\to\mathcal{B}$ is also smooth. Equivalently, it is a smooth bijection with a smooth inverse.
In continuum mechanics, a diffeomorphism represents an admissible deformation, i.e., a smooth one-to-one mapping that excludes interpenetration of matter, tearing, and the formation of voids.

The configuration space of nonlinear elasticity is the set of all embeddings of the body $\mathcal{B}$ into the ambient space $\mathcal{S}$. We denote this set by $\mathcal{E}=\operatorname{Emb}(\mathcal{B},\mathcal{S})$ \citep{EbinMarsden1970,Simo1988}. Let $\xi:\mathcal{S}\to\mathcal{S}$ and $\Xi:\mathcal{B}\to\mathcal{B}$ be spatial and material diffeomorphisms, respectively. The collections of spatial and material diffeomorphisms form groups, denoted by $\operatorname{Diff}(\mathcal{S})$ and $\operatorname{Diff}(\mathcal{B})$, respectively.\footnote{For example, $\operatorname{Diff}(\mathcal{S})$ is a group under composition: the identity element is $\operatorname{id}_{\mathcal{S}}$, the inverse of any diffeomorphism $\xi$ is $\xi^{-1}$, which is well-defined and unique, and associativity follows from the associativity of composition. The same holds for $\operatorname{Diff}(\mathcal{B})$.}

\begin{remark}\label{Diff-Emb}
Since $\mathcal{B}\subset \mathcal{S}$ is an embedded submanifold, let $\iota:\mathcal{B}\hookrightarrow \mathcal{S}$ denote the inclusion map. For any $\Xi\in \operatorname{Diff}(\mathcal{B})$, $\Xi:\mathcal{B}\to\mathcal{B}$, and hence $\iota\circ \Xi:\mathcal{B}\to\mathcal{S}$. Since $\Xi$ is a diffeomorphism, $\iota\circ \Xi$ is an embedding, and thus $\iota\circ \Xi\in \mathcal{E}=\operatorname{Emb}(\mathcal{B},\mathcal{S})$. Therefore, $\iota\circ \operatorname{Diff}(\mathcal{B}) \subset \mathcal{E}$. 
For any $\xi\in \operatorname{Diff}(\mathcal{S})$, the restriction $\xi\big|_{\mathcal{B}}:\mathcal{B}\to \mathcal{S}$ is an embedding, and hence $\xi\big|_{\mathcal{B}}\in \mathcal{E}$. 
In particular, any deformation $\varphi:\mathcal{B}\to\mathcal{C}\subset\mathcal{S}$ is an embedding, and thus $\varphi\in \mathcal{E}$. Since $\mathcal{B}$ is compact, any embedding $\varphi\in \mathcal{E}$ can be extended to a diffeomorphism $\xi\in \operatorname{Diff}(\mathcal{S})$ such that $\varphi=\xi\big|_{\mathcal{B}}$. This is a consequence of the isotopy extension theorem \citep{Hirsch1976}.
\end{remark}

Let us define the left translation of $\varphi\in\mathcal{E}$ as
\begin{equation} \label{Left-Action}
\begin{aligned}
	L_{\xi}:\operatorname{Diff}(\mathcal{S})\times\mathcal{E} &\longrightarrow \mathcal{E} \\
	(\xi,\varphi) &\longmapsto L_{\xi}(\varphi)=\xi\circ\varphi\,.
\end{aligned}
\end{equation}
Similarly, the right translation of $\varphi\in\mathcal{E}$ is defined as
\begin{equation} \label{Right-Action}
\begin{aligned}
	R_{\Xi}: \mathcal{E} \times\operatorname{Diff}(\mathcal{B}) &\longrightarrow \mathcal{E} \\
	(\xi,\varphi) &\longmapsto R_{\Xi}(\varphi)= \varphi\circ\Xi\,.
\end{aligned}
\end{equation}
In order to define the left and right actions on spatial and material tensor fields, one first computes the tangent maps associated with $L_{\xi}$ and $R_{\Xi}$. To this end, consider a $1$-parameter family of maps $\varphi_{\epsilon}\in \mathcal{E}$ such that $\varphi_{\epsilon}\big|_{\epsilon=0}=\varphi$. Then
\begin{equation}
	\delta\varphi(X)\coloneqq \frac{d}{d\epsilon} \varphi_{\epsilon}(X)\Big|_{\epsilon=0}	\,,
\end{equation}
where $\delta\varphi$ is a vector field along $\varphi$.\footnote{A vector field along $\varphi$ is a map $\mathbf{Z}:\mathcal{B}\to T\mathcal{S}$ such that $\mathbf{Z}(X)\in T_{\varphi(X)}\mathcal{S}$ for every $X\in\mathcal{B}$, i.e., it assigns to each material point a vector based at its image under $\varphi$.} 
If $\epsilon$ is interpreted as time, then $\delta\varphi$ coincides with the material velocity field $\mathbf{V}$.
Now, the tangent map of $L_{\xi}$ is computed as
\begin{equation}
	(T_{\varphi}L_{\xi})\cdot \delta\varphi(X) 
	= \frac{d}{d\epsilon}\Bigg|_{\epsilon=0} L_{\xi}\varphi_{\epsilon}(X) \\
	= \frac{d}{d\epsilon}\Bigg|_{\epsilon=0} \xi(\varphi_{\epsilon}(X)) \\
	= T\xi(\varphi(X))\cdot \delta\varphi(X)
	\,,
\end{equation}
where in the last equality the chain rule was used. Note that $\delta\varphi(X)\in T_{\varphi(X)}\mathcal{S}$, while $T\xi(\varphi(X)):T_{\varphi(X)}\mathcal{S}\to T_{\xi(\varphi(X))}\mathcal{S}$. Therefore, $(T_{\varphi}L_{\xi})\cdot \delta\varphi$ is a vector field along $\xi\circ\varphi$, and one may write
\begin{equation}
	(T_{\varphi}L_{\xi})\cdot \delta\varphi
	=(T\xi\circ\varphi)\cdot \delta\varphi\,,\qquad \forall \xi\in \operatorname{Diff}(\mathcal{S})
	\,.
\end{equation}
Similarly, the tangent map of $R_{\Xi}$ is computed as
\begin{equation}
	(T_{\varphi}R_{\Xi})\cdot \delta\varphi(X)
	= \frac{d}{d\epsilon}\Bigg|_{\epsilon=0} R_{\Xi}\varphi_{\epsilon}(X)
	= \frac{d}{d\epsilon}\Bigg|_{\epsilon=0} \varphi_{\epsilon}(\Xi(X))
	= \delta\varphi(\Xi(X))
	\,.
\end{equation}
Note that $\delta\varphi(\Xi(X))\in T_{\varphi(\Xi(X))}\mathcal{S}$. Therefore, $(T_{\varphi}R_{\Xi})\cdot \delta\varphi$ is a vector field along $\varphi\circ\Xi$, and one may write
\begin{equation}
	(T_{\varphi}R_{\Xi})\cdot \delta\varphi
	=\delta\varphi\circ\Xi\,,\qquad \forall \Xi\in \operatorname{Diff}(\mathcal{B})
	\,.
\end{equation}

The geometric meaning of the left and right actions defined above is illustrated in Fig.~\ref{Fig:Covariance}. A deformation $\varphi$ is transformed by a material relabeling $\Xi$ and a spatial reparametrization $\xi$ to give $\varphi’=\xi\circ\varphi\circ\Xi^{-1}$, where both mappings describe the same physical deformation process. This diagram clarifies the roles of $\Xi$ and $\xi$ and the distinction between material and spatial covariance.
\begin{figure}[t!]
\centering
\includegraphics[width=0.5\textwidth]{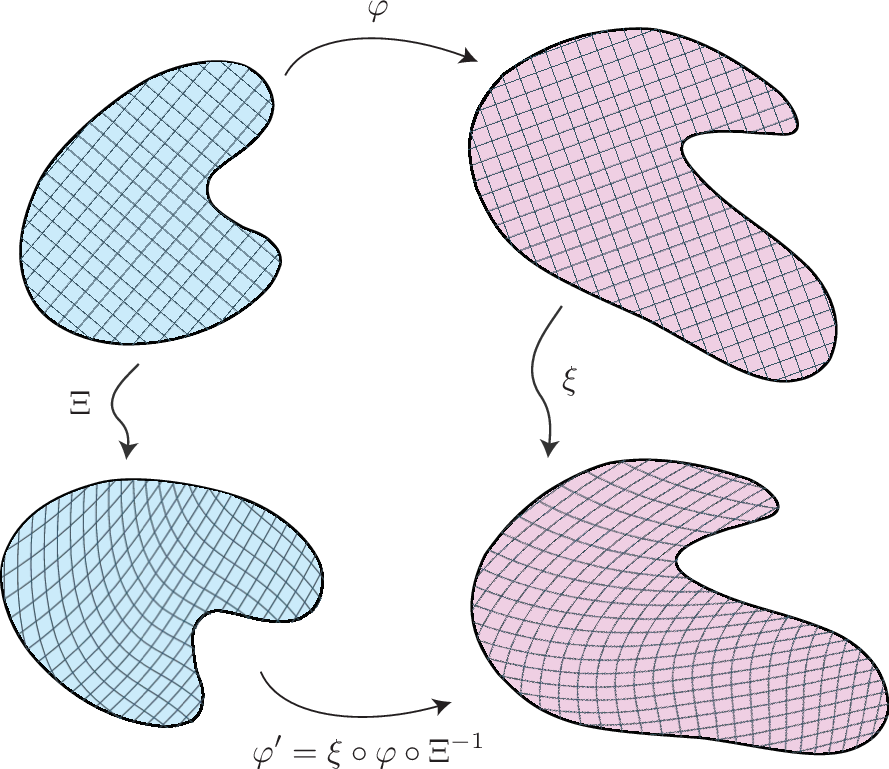}
\vspace*{0.10in}
\caption{Schematic illustration of spatial and material covariance. The mappings $\varphi$ and $\varphi’=\xi\circ\varphi\circ\Xi^{-1}$ describe the same physical deformation process, and are related by a material relabeling $\Xi$ and a spatial reparametrization $\xi$. The blue configurations represent the undeformed, stress-free body, while the pink configurations represent the same deformed, stressed body. Restricting to $\Xi=\mathrm{id}_{\mathcal{B}}$ gives spatial covariance, while restricting to $\xi=\mathrm{id}_{\mathcal{S}}$ gives material covariance. The map $\xi$ may be viewed as a warped lens through which the same deformed body is described. Covariance of a continuum field theory means that any tensor field, including scalar and vector fields, transforms under these maps via push-forward or pull-back.}
\label{Fig:Covariance}
\end{figure}

\subsubsection{Spatial covariance}

An energy function $W=W(X,\mathbf{F},\mathbf{G},\mathbf{g}\circ\varphi)$ is \textit{spatially covariant} if it is left invariant under the action of the group $\operatorname{Diff}(\mathcal{S})$. This can be informally written as $L_{\xi}W=W$. More formally, one has 
\begin{equation} \label{W-Spatial-Covariance}
	W(X,\xi_*\mathbf{F},\mathbf{G},\xi_*\mathbf{g})
	=W(X,\mathbf{F},\mathbf{G},\mathbf{g}\circ\xi\circ\varphi)\,.
\end{equation}
This is more specifically written as
\begin{equation} 
	W\big(X,T\xi\cdot\mathbf{F},\mathbf{G},(T\xi)^{-\star}\cdot\mathbf{g}\cdot(T\xi)^{-1}\big)
	=W(X,\mathbf{F},\mathbf{G},\mathbf{g}\circ\xi\circ\varphi)\,.
\end{equation}
Note that under a spatial diffeomorphism material tensors remain unchanged. 

Next we show that, by choosing $\xi\in\operatorname{Diff}(\mathcal{S})$ to be a diffeomorphic extension of $\varphi^{-1}$, spatial covariance implies that the dependence of the energy function on $\mathbf{F}$ and $\mathbf{g}$ reduces to that on $\mathbf{C}^\flat$.
Since $\varphi:\mathcal{B}\to\mathcal{C}\subset\mathcal{S}$ is an embedding and $\mathcal{B}$ is compact, $\varphi$ is a diffeomorphism onto its image $\mathcal{C}$. Hence, $\varphi^{-1}:\mathcal{C}\to\mathcal{B}$ is a diffeomorphism. Because $\mathcal{C}$ is compact, the isotopy extension theorem (see Remark~\ref{Diff-Emb}) implies that there exists $\xi\in\operatorname{Diff}(\mathcal{S})$ such that $\xi \big|_{\mathcal{C}}=\varphi^{-1}$. Thus, in the spatial covariance relation one may choose $\xi$ to be any diffeomorphic extension of $\varphi^{-1}$ to $\mathcal{S}$. Therefore,
\begin{equation}
	W(X,\mathbf{F},\mathbf{G},\mathbf{g}\circ\varphi)
	= W\big(X,\mathbf{F}^{-1}\cdot\mathbf{F},\mathbf{G},\mathbf{F}^\star \mathbf{g}\,\mathbf{F}\big)
	= W\big(X,\operatorname{id}_{T_X\mathcal{B}},\mathbf{G},\mathbf{C}^\flat\big)
	= \hat{W}\big(X,\mathbf{G},\mathbf{C}^\flat\big)
	\,,
\end{equation}
where
\begin{equation}
	\hat{W}\big(X,\mathbf{G},\mathbf{C}^\flat\big)
	\coloneqq	W\big(X,\operatorname{id}_{T_X\mathcal{B}},
	\mathbf{G},\mathbf{C}^\flat\big)	\,.
\end{equation}

\subsubsection{Material covariance}

An energy function $W=W(X,\mathbf{F},\mathbf{G},\mathbf{g}\circ\varphi)$ is \textit{materially covariant} if it is right invariant under the action of the group $\operatorname{Diff}(\mathcal{B})$. This can be informally written as $R_{\Xi}W=W$. More formally, one has 
\begin{equation} \label{W-Material-Covariance}
	W(\Xi(X),\Xi_*\mathbf{F},\Xi_*\mathbf{G},\mathbf{g}\circ\varphi\circ\Xi)
	=W\big(\Xi(X),\mathbf{F}\cdot T\Xi,(T\Xi)^{-\star}\mathbf{G}\,(T\Xi)^{-1},\mathbf{g}\circ\varphi\circ\Xi\big)
	=W(X,\mathbf{F},\mathbf{G},\mathbf{g}\circ\varphi)\,.
\end{equation}
Using the representation resulting from spatial covariance, the material covariance can be expressed as
\begin{equation}
	\hat{W}\big(\Xi(X),\Xi_*\mathbf{G},\Xi_*\mathbf{C}^\flat\big)
	=\hat{W}\big(\Xi(X),(T\Xi)^{-\star}\mathbf{G}\,(T\Xi)^{-1},(T\Xi)^{-\star}\mathbf{C}^\flat\,(T\Xi)^{-1}\big)
	= \hat{W}\big(X,\mathbf{G},\mathbf{C}^\flat\big)
	\,.
\end{equation}

Next we show that, by choosing $\Xi\in\operatorname{Diff}(\mathcal{B})$ to be the material diffeomorphism induced by $\varphi$, material covariance implies that the dependence of $\hat{W}$ on $\mathbf{G}$ and $\mathbf{C}^\flat$ reduces to that on $\mathbf{c}^\flat$ and $\mathbf{g}$, i.e., material isotropy.\footnote{Dependence on $\mathbf{c}^\flat$ and $\mathbf{g}$ implies that the energy is expressed purely in terms of spatial tensors and their invariants, and is therefore independent of any preferred material directions, which is precisely isotropy.} As $\varphi:\mathcal{B}\to\mathcal{C}\subset\mathcal{S}$ is an embedding and $\mathcal{B}$ is compact, $\varphi$ is a diffeomorphism onto its image $\mathcal{C}$. Therefore, one may use $\varphi$ itself as the material relabeling map in the covariance relation, i.e., one may choose $\Xi=\varphi$. From material covariance one obtains
\begin{equation}
	\hat{W}\big(X,\mathbf{G},\mathbf{C}^\flat\big)
	=\hat{W}\big(\varphi(X),\varphi_*\mathbf{G},\varphi_*\mathbf{C}^\flat\big)
	=\hat{W}\big(\varphi(X),\mathbf{F}^{-\star}\mathbf{G}\mathbf{F}^{-1},\varphi_*\varphi^*\mathbf{g}\big)
	=\hat{W}\big(x,\mathbf{c}^\flat,\mathbf{g}\big)
	\,.
\end{equation}

In the next section we show that spatial covariance plays a central role in rewriting a strength function in terms of different stress measures.

\section{Material Strength} \label{Sec:Material-Strength}

Fracture nucleation and evolution in brittle solids under quasi-static loading are governed by three macroscopic (continuum) material properties: (i) the elastic response of the material, characterized by elastic constants in linear elasticity and by a strain-energy function in hyperelasticity,\footnote{For Cauchy elastic solids, the elastic response is characterized by constitutive response functions rather than by a strain-energy function.} (ii) the fracture toughness, or equivalently a critical energy release rate, and (iii) the material strength. The present discussion precludes dynamic loading which would require additional analysis.

As stated in the Introduction, we adopt the following definition for material strength \citep{KumarLopezPamies2020,KumarBourdinFrancfortLopezPamies2020}: 
``the strength of an elastic brittle material is the set of all critical stresses $\mathbf{S}$ at which the material fractures when it is subjected to a state of monotonically increasing,\footnote{The phrase ``monotonically increasing stress" is part of the definition quoted from \citep{KumarLopezPamies2020,KumarBourdinFrancfortLopezPamies2020}. 
This phrase is intended only to exclude loading paths that involve unloading or cyclic variations of the applied stress state. No ordering of general tensor-valued stress states is intended. In the mathematical formulation used here, material strength is represented by the boundary of the safe domain in the space of spatially uniform, constitutively admissible stresses. The loading qualification indicates that this boundary is approached from within the safe domain before the first occurrence of fracture. For proportional loading, the loading path is a straight ray emanating from the origin, as in \citep{Malmeister1966}.}
spatially uniform, but otherwise arbitrary stress."
In other words, the strength surface controls when crack formation occurs under uniform stresses. However, in general, the strength surface violation acts as a necessary but not sufficient condition for crack evolution. Material toughness plays an important role under non-uniform stresses. A complete mathematical understanding of how strength and toughness couple with each other is still lacking. However, it has been hypothesized with some evidence \citep{lopezpamieskamarei2025, ward2025} that the strength surface acts as a constraint on the variational statement of Griffith fracture as put forward by \citet{FrancfortMarigo1998} and extended to nonlinear elastic materials by \citet{francfort2005nonlinear}. 
Mathematically, this may be stated as follows: under Dirichlet boundary conditions, at a discrete time $t_k$ the pair of deformation field and the crack set  $\left( \varphi_k := \varphi_k(X, t_k), \; \Gamma_k := \Gamma(t_k) \right)$ minimizes the functional
\begin{equation}
	\mathcal{E}(\varphi,\Gamma) \coloneqq 
	\int_{\Omega_0 \setminus \Gamma} W(\mathbf{F},\mathbf{G},\mathbf{g}) \, \mathrm{d}V 
	+ G_c \, \mathcal{A}(\Gamma)\,,
	\label{Variational}
\end{equation}
among all admissible pairs $(\varphi,\Gamma = \Gamma_{k-1} \cup \Delta \Gamma)$ such that
\begin{equation}
 \Delta \Gamma \subset \mathcal{B}_f(t)\,,
\end{equation}
where
\begin{equation}
	\mathcal{B}_f(t) = \left\{ X\in\mathcal{B} : 
	\mathsf{F}(\mathbf{P},\mathbf{F},\mathbf{g},\mathbf{G}) \ge 0 \right\}\,,
\end{equation}
is the failed subset of the body, while $\mathsf{F}(\mathbf{P},\mathbf{F},\mathbf{g},\mathbf{G})=0$ defines the strength hypersurface. Many challenging mathematical questions remain regarding this formulation and its extensions, for instance, to account for boundary loads \citep{chambolle2009, larsen2021, larsen2024}. 
Still, a phase field model corresponding to this formulation has been developed by \citet{KumarBourdinFrancfortLopezPamies2020} and \citet{KumarLopezPamies2020} and has proven successful in predicting crack nucleation and propagation in a wide variety of brittle materials under arbitrary loading conditions, including finite deformations \citep{Kamarei2026}. The model is not developed as a direct regularization of the above constrained variational formulation; however, it contains the essential features. We note here also other recent modeling efforts to incorporate strength into the classical variational formulation \citep{bourdin2025variational}.

As discussed in the Introduction, a critical question to evaluate in this definition is the phrase ``spatially uniform stress''.
To see the difficulty, recall that for a homogeneous compressible hyperelastic solid a homogeneous deformation, or equivalently a uniform strain, induces a spatially uniform stress field. However, a spatially uniform stress field may also be supported by an inhomogeneous deformation, so that uniform stress need not correspond to a homogeneous strain field. This can already be seen for a compressible hyperelastic material with energy function $W=W(J)$, for which the Cauchy stress is purely hydrostatic and depends only on $J$ (a hyperelastic fluid). Consider the inhomogeneous simple shear deformation $\varphi(X_1,X_2,X_3)=\big(X_1\,,X_2+f(X_1)\,,X_3\big)$ with a non-affine function $f$. This deformation is isochoric, i.e., $J=1$, while the strain is inhomogeneous because the deformation gradient is written as
\begin{equation}
	\mathbf{F}=\begin{bmatrix}
	1 & 0 & 0 \\
	f'(X_1) & 1 & 0\\
	0 & 0 & 1
	\end{bmatrix}\,,
\end{equation}
and $f'(X_1)$ varies with $X_1$. Nonetheless, the Cauchy stress is spatially uniform, namely $\boldsymbol{\sigma}=W'(1)\,\mathbf{g}^\sharp$. 
It is known more generally that homogeneous Cauchy stress in isotropic hyperelasticity does not necessarily imply homogeneous deformations \citep{MihaiNeff2017,MihaiNeff2018}. The examples constructed in these works are piecewise homogeneous deformations whose deformation gradients are rank-one connected across planar interfaces, in the spirit of the classical Ball--James theory of microstructure \citep{BallJames1987}. However, if the stored-energy function is strictly rank-one convex, then such non-homogeneous rank-one connected deformations cannot support a homogeneous Cauchy stress, since in that case the Cauchy stress is injective along rank-one connected lines \citep{NeffMihai2017}.

These examples show that, in finite elasticity, the notion of a spatially uniform stress state is more subtle than in the linear theory. This naturally leads to the question of what measure of stress should be used in defining material strength.

\subsection{What measure of stress should be used in defining material strength?} 

If one wishes to define material strength as the set of critical spatially uniform stresses, one must first decide which stress measure is to be used and under what assumptions a uniform stress state is regarded as representing a homogeneous loading state. This issue is especially important in finite elasticity, where different stress measures live in different configurations and encode different geometric information.

Let us consider a homogeneous undeformed body made of some hyperelastic solid. 
Let us define the material strength of this solid as the set of critical spatially uniform stresses under which the body fractures.
For an arbitrary material point $X\in\mathcal{B}$ and given any element of area $dA$ with unit normal $\mathbf{N}(X)$, traction in the deformed configuration is $\mathbf{t}=\mathbf{P}\mathbf{N}^\flat$, where $\mathbf{P}$ is the first Piola--Kirchhoff stress. Traction $\mathbf{t}$ acts on the element of area $da$ with unit normal $\mathbf{n}\circ\varphi(X)$ in the current configuration. By Nanson's formula one has $\mathbf{n}^{\flat}\,da=J\,\mathbf{F}^{-\star}\,\mathbf{N}^{\flat}\,dA$, or equivalently $\mathbf{n}\,da=J\mathbf{F}^{-\mathsf{T}}\mathbf{N}\,dA$.
These considerations suggest that, if one wants to characterize strength in terms of boundary loading, the first Piola--Kirchhoff stress is a natural stress measure to use. Indeed, $\mathbf{P}$ directly determines the traction vector associated with a given oriented area element in the reference configuration, which is the materially homogeneous stress-free configuration of the body. In this sense, $\mathbf{P}$ is convenient for prescribing and characterizing spatially homogeneous loading by boundary tractions. 

On the other hand, if one wishes to regard homogeneity in a purely spatial sense, then the Cauchy stress $\boldsymbol{\sigma}$ is a more natural stress measure. Thus, even at the level of definitions, one must distinguish between uniformity relative to the reference configuration and uniformity relative to the current configuration. However, uniformity of the Cauchy stress, in general, is not equivalent to uniformity of the first Piola--Kirchhoff stress.
It should be emphasized that the various stress measures are mechanically equivalent in the sense that, when paired with the corresponding area element in the appropriate configuration, they represent the same physical traction vector. Nevertheless, the notion of strength depends, in general, on which stress measure is assumed to be homogeneous at the onset of fracture. Once strength is defined with respect to one stress measure, one may then ask how the corresponding strength function is represented in terms of any other stress measure.
We examine each stress measure in turn---namely the first and second Piola--Kirchhoff stresses and the Cauchy stress---and compare the corresponding notions of uniformity in the reference and spatial configurations.

\paragraph{Uniform first Piola--Kirchhoff stress.}
Let us follow \citep{KumarLopezPamies2020,KumarBourdinFrancfortLopezPamies2020} and assume that the first Piola--Kirchhoff stress is spatially uniform at the onset of fracture. For the two-point tensor $\mathbf{P}$, spatial uniformity means that its spatial covariant derivative vanishes, i.e., in components $P^{aA}{}_{|b}=0$. Recall that \citep{MarsdenHughes1983}
\begin{equation}
	P^{aA}{}_{|B} = P^{aA}{}_{,B} +\gamma^a{}_{cd}\,F^d{}_B\,P^{cA}
	+\Gamma^A{}_{CB}\,P^{aC}\,,
\end{equation}
where $\gamma^a{}_{bc}$ are the Christoffel symbols of the spatial metric $\mathbf{g}$ in a coordinate chart $\{x^a\}$ and $\Gamma^A{}_{BC}$ are the Christoffel symbols of the material metric $\mathbf{G}$ in a coordinate chart $\{X^A\}$. We also know that $P^{aA}{}_{|B}=F^b{}_B\,P^{aA}{}_{|b}$. Since the deformation gradient is invertible, it follows that
\begin{equation}
	P^{aA}{}_{|b}=0 \quad\Leftrightarrow\quad P^{aA}{}_{|B}=0\,.
\end{equation}
In other words, material and spatial uniformity of the first Piola-Kirchhoff stress are equivalent.

The notion of spatial or material uniformity for a two-point tensor such as the first Piola--Kirchhoff stress requires some care. Unlike a purely spatial or purely material tensor, $\mathbf{P}$ maps material covectors to spatial vectors, i.e., $\mathbf{P}:T_X^*\mathcal{B}\to T_x\mathcal{C}$, and therefore depends simultaneously on material and spatial positions. For this reason, it does not make sense to define uniformity by saying that the components of $\mathbf{P}$ are constant in a given coordinate system. The appropriate notion of uniformity is covariant constancy. The condition $P^{aA}{}_{|b}=0$ means that $\mathbf{P}$ does not vary from point to point in the spatial manifold once the geometry of the current configuration is taken into account, while $P^{aA}{}_{|B}=0$ means that $\mathbf{P}$ does not vary from point to point in the material manifold once the geometry of the reference configuration is taken into account. Since these two conditions are related by $P^{aA}{}_{|B}=F^b{}_B\,P^{aA}{}_{|b}$ and the deformation gradient is invertible, spatial and material uniformity are equivalent. 
Thus, for a two-point tensor, homogeneity means that the mapping from material covectors to spatial vectors is the same everywhere in the body, after accounting for the geometries of the reference and current configurations.

We next examine the implications of spatial uniformity of the first Piola--Kirchhoff stress for the second Piola--Kirchhoff and Cauchy stresses.

\paragraph{Uniform second Piola--Kirchhoff stress.}
The material covariant derivative of the second Piola--Kirchhoff stress with respect to the material metric $\mathbf{G}$, $S^{AB}{}_{|C}$, is computed as follows. Recall that $S^{AB}=F^{-A}{}_a\,P^{aB}$. Hence,
\begin{equation}
	S^{AB}{}_{|C} = F^{-A}{}_{a|C}\,P^{aB} +F^{-A}{}_a\,P^{aB}{}_{|C}= F^{-A}{}_{a|C}\,P^{aB}\,.
\end{equation}
Thus, uniformity of the first Piola--Kirchhoff stress does not, in general, imply material uniformity of the second Piola--Kirchhoff stress. The latter holds only under the additional condition $F^{-A}{}_{a|C}=0$, i.e., for homogeneous deformations.

\begin{remark}
A clarification is needed here. When the balance of linear momentum is written in terms of the second Piola--Kirchhoff stress, the operator $\operatorname{Div}\mathbf{S}$ is the divergence with respect to the metric $\mathbf{C}^\flat$, not the material metric $\mathbf{G}$. Thus, the equilibrium equation for $\mathbf{S}$ involves the Levi--Civita connection of $\mathbf{C}^\flat$.
This should not be confused with the notion of material uniformity. Uniformity of $\mathbf{S}$ is defined using the material covariant derivative associated with $\mathbf{G}$, since uniformity is a geometric notion on the material manifold $(\mathcal{B},\mathbf{G})$. In elasticity, $\mathbf{G}$ is the induced metric on the reference configuration, namely $\mathbf{G}=\mathbf{g}\big|_{\mathcal{B}}$.
Therefore, two different metrics appear in the discussion of $\mathbf{S}$. The metric $\mathbf{C}^\flat$ enters the equilibrium equation through the divergence operator, whereas the metric $\mathbf{G}$ enters the definition of material uniformity. These two notions are distinct and should not be conflated.
\end{remark}

\paragraph{Uniform Cauchy stress.}
Recall that $\sigma^{ab}=J^{-1}P^{aA}F^b{}_A$. Taking the spatial covariant derivative, and assuming that the first Piola--Kirchhoff stress has vanishing covariant derivative, one obtains
\begin{equation}
	\sigma^{ab}{}_{|c} = (J^{-1})_{,c}\,P^{aA}F^b{}_A +J^{-1}P^{aA}F^b{}_{A|c}\,.
\end{equation}
Now, since $J_{,c}=F^{-C}{}_c\,J_{,C}$ and $J_{,C}=J\,F^{-A}{}_d\,F^d{}_{A|C}$, one finds $(J^{-1})_{,c}=-J^{-1}F^{-C}{}_c\,F^{-A}{}_d\,F^d{}_{A|C}$. Also, $F^b{}_{A|c}=F^b{}_{A|C}\,F^{-C}{}_c$. Therefore,
\begin{equation}
\begin{aligned}
	\sigma^{ab}{}_{|c} &= -J^{-1}F^{-C}{}_c\,F^{-A}{}_d\,F^d{}_{A|C}\,\sigma^{ab} 
	+J^{-1}P^{aA}F^b{}_{A|C}\,F^{-C}{}_c \\
	&= J^{-1}F^{-C}{}_c
	\left( P^{aA}F^b{}_{A|C} - F^{-A}{}_d\,F^d{}_{A|C}\,\sigma^{ab}	\right)\,.
\end{aligned}
\end{equation}
Thus, vanishing covariant derivative of the first Piola--Kirchhoff stress does not, in general, imply vanishing covariant derivative of the Cauchy stress. For homogeneous deformations, however, $F^a{}_{A|B}=0$, and hence $\sigma^{ab}{}_{|c}=0$, i.e., the Cauchy stress is spatially uniform.

In summary, uniformity of one measure of stress does not, in general, imply uniformity of another. In particular, for the first Piola--Kirchhoff stress, spatial and material uniformity are equivalent. However, uniformity of the first Piola--Kirchhoff stress does not, in general, imply either spatial uniformity of the Cauchy stress or material uniformity of the second Piola--Kirchhoff stress. For homogeneous deformations, on the other hand, spatial uniformity of the Cauchy stress, material uniformity of the second Piola--Kirchhoff stress, and uniformity of the first Piola--Kirchhoff stress are equivalent.

\begin{remark}
Local invertibility of the stress-strain relation is not sufficient to conclude that a spatially uniform stress field necessarily corresponds to a spatially uniform strain field. As a matter of fact, local invertibility only guarantees uniqueness of strain within a neighborhood of a given state, and therefore does not exclude the possibility that the same uniform stress may correspond to strain states lying on different branches of the constitutive equation. In contrast, if the constitutive equation is globally invertible, then a uniform stress field uniquely determines a uniform strain field. In particular, let us assume that the constitutive equation $\boldsymbol{\sigma}=\hat{\boldsymbol{\sigma}}(\mathbf{b}^\sharp,\mathbf{g})$ is invertible. Let $\boldsymbol{\sigma}_0$ be a homogeneous Cauchy stress. Then $\mathbf{b}^\sharp=\mathbf{b}_0^\sharp=\hat{\boldsymbol{\sigma}}^{-1}(\boldsymbol{\sigma}_0,\mathbf{g})$ is uniform. This gives a unique uniform spatial stretch tensor that corresponds to a homogeneous deformation \citep{Ciarlet1988,MihaiNeff2018}.
In this case, the first Piola--Kirchhoff, Cauchy, and second Piola--Kirchhoff stresses are uniform simultaneously.
\end{remark}

\begin{remark}[Invertibility of the linear isotropic stress-strain map]
The linear isotropic constitutive equations $\boldsymbol{\sigma}=2\mu\,\mathrm{dev}\,\boldsymbol{\varepsilon}+\kappa\,\mathrm{tr}(\boldsymbol{\varepsilon})\,\mathbf{g}$ ($\mu$ is shear modulus and $\kappa$ is bulk modulus) are invertible if and only if $\mu\neq 0$ and $\kappa\neq 0$. Note that taking the trace, one finds $\mathrm{tr}(\boldsymbol{\sigma})=3\kappa\,\mathrm{tr}(\boldsymbol{\varepsilon})$, and hence $\mathrm{tr}(\boldsymbol{\varepsilon})=\dfrac{1}{3\kappa}\,\mathrm{tr}(\boldsymbol{\sigma})$, provided $\kappa\neq 0$. Taking the deviatoric part, one finds $\mathrm{dev}\,\boldsymbol{\sigma}=2\mu\,\mathrm{dev}\,\boldsymbol{\varepsilon}$, and hence $\mathrm{dev}\,\boldsymbol{\varepsilon}=\dfrac{1}{2\mu}\,\mathrm{dev}\,\boldsymbol{\sigma}$, provided $\mu\neq 0$. Therefore, both $\mathrm{tr}(\boldsymbol{\varepsilon})$ and $\mathrm{dev}\,\boldsymbol{\varepsilon}$ are uniquely determined by $\boldsymbol{\sigma}$. Using the decomposition $\boldsymbol{\varepsilon}=\mathrm{dev}\,\boldsymbol{\varepsilon}+\dfrac{1}{3}\,\mathrm{tr}(\boldsymbol{\varepsilon})\,\mathbf{g}$, one concludes that $\boldsymbol{\varepsilon}$ is uniquely determined by $\boldsymbol{\sigma}$.
\end{remark}

\begin{remark}[Non-invertibility of $\mathbf{P}=\hat{\mathbf{P}}(\mathbf{F},\mathbf{G},\mathbf{g})$ with respect to $\mathbf{F}$]
Let us consider a Cauchy elastic solid for which $\mathbf{P}=\hat{\mathbf{P}}(\mathbf{F},\mathbf{G},\mathbf{g})$. It is a classic result that constitutive equations cannot be invertible when written with respect to the deformation gradient \citep{Schweickert2018}.
Objectivity of $\hat{\mathbf{P}}$ implies that for every $\mathbf{Q}\in SO(3)$, one has
\begin{equation}
	\hat{\mathbf{P}}(\mathbf{Q}\mathbf{F},\mathbf{G},\mathbf{g})
	=\mathbf{Q}\hat{\mathbf{P}}(\mathbf{F},\mathbf{G},\mathbf{g})\,.
\end{equation}
Assume now, for the sake of contradiction, that $\hat{\mathbf{P}}$ is injective in $\mathbf{F}$ for fixed $\mathbf{G}$ and $\mathbf{g}$. Then
\begin{equation}
	\hat{\mathbf{P}}(\mathbf{F}_{1},\mathbf{G},\mathbf{g})
	=\hat{\mathbf{P}}(\mathbf{F}_{2},\mathbf{G},\mathbf{g})
	\quad\Longrightarrow\quad
	\mathbf{F}_{1}=\mathbf{F}_{2}\,.
\end{equation}
Consider the stress-free reference configuration corresponding to the inclusion map $\iota:\mathcal{B}\hookrightarrow\mathcal{S}$, and denote its tangent map by $\mathbf{I}=T\iota$. Assume that $\hat{\mathbf{P}}(\mathbf{I},\mathbf{G},\mathbf{g}) =\mathbf{0}$. Using material frame indifference with $\mathbf{F}=\mathbf{I}$, for any $\mathbf{Q}\in SO(3)$, one obtains
\begin{equation}
	\hat{\mathbf{P}}(\mathbf{Q},\mathbf{G},\mathbf{g})
	=\hat{\mathbf{P}}(\mathbf{Q}\mathbf{I},\mathbf{G},\mathbf{g})
	=\mathbf{Q}\hat{\mathbf{P}}(\mathbf{I},\mathbf{G},\mathbf{g})
	=\mathbf{Q}\mathbf{0}
	=\mathbf{0}
	=\hat{\mathbf{P}}(\mathbf{I},\mathbf{G},\mathbf{g})\,.
\end{equation}
Since $\hat{\mathbf{P}}$ is injective in $\mathbf{F}$, it follows that $\mathbf{Q}=\mathbf{I}$, $\forall\,\mathbf{Q}\in SO(3)$, which is absurd. Therefore, a constitutive equation satisfying material frame indifference cannot be injective in $\mathbf{F}$. In particular, $\mathbf{P}=\hat{\mathbf{P}}(\mathbf{F},\mathbf{G},\mathbf{g})$ cannot be inverted uniquely to recover $\mathbf{F}$.
\end{remark}

When written in terms of the Cauchy stress, constitutive equations may be invertible with respect to the strain measure $\mathbf{b}^\sharp$, i.e., $\boldsymbol{\sigma}=\hat{\boldsymbol{\sigma}}(\mathbf{b}^\sharp,\mathbf{g})$ is invertible. In such cases, a uniform state of Cauchy stress uniquely determines $\mathbf{b}^\sharp$. Consequently, a uniform state of Cauchy stress corresponds to a unique homogeneous deformation up to isometries of the ambient space.

\begin{example}
Consider the standard compressible neo-Hookean constitutive equation written in terms of the Cauchy stress:
\begin{equation}
	\boldsymbol{\sigma}
	=\frac{\mu}{J}\,(\mathbf{b}^\sharp-\mathbf{g}^\sharp)
	+\frac{\kappa \ln J}{J}\,\mathbf{g}^\sharp\,,
\end{equation}
where $J>0$ and $\det\mathbf{b}^\sharp=J^2(\det\mathbf{g})^{-1}$. Unlike linear elasticity, this is not an affine relation in $\mathbf{b}^\sharp$. Even if $\boldsymbol{\sigma}$ is prescribed, the unknown $\mathbf{b}^\sharp$ appears both explicitly and implicitly through $J$, since $J$ itself depends on $\mathbf{b}^\sharp$. Rearranging the constitutive equation, one obtains
\begin{equation}
	\mathbf{b}^\sharp
	=\frac{J}{\mu}\,\boldsymbol{\sigma}
	+\left(1-\frac{\kappa}{\mu}\ln J\right)\mathbf{g}^\sharp\,.
\end{equation}
This is not yet a solution for $\mathbf{b}^\sharp$, because $J=J(\mathbf{b}^\sharp,\mathbf{g})$. Substituting this expression into the determinant relation leads to a nonlinear scalar equation for $J$. In general, it is not clear a priori whether this equation has one solution, several solutions, or none. Thus, in the general case, global invertibility of the constitutive mapping $\mathbf{b}^\sharp\mapsto\boldsymbol{\sigma}$ is a genuinely nontrivial question.
To see concretely what can happen, let us restrict attention to purely dilatational deformations and assume $\mathbf{b}^\sharp=\alpha\,\mathbf{g}^\sharp$, $\alpha>0$. Then, $\det\mathbf{b}^\sharp =\alpha^3 \det\mathbf{g}^\sharp =\alpha^3(\det\mathbf{g})^{-1}$. Comparing with $\det\mathbf{b}^\sharp=J^2/\det\mathbf{g}$, one obtains $J=\alpha^{\frac{3}{2}}$. Substituting this into the constitutive equation gives us
\begin{equation}
	\boldsymbol{\sigma}
	=\frac{\mu}{J}\,(\alpha\,\mathbf{g}^\sharp-\mathbf{g}^\sharp)
	+\frac{\kappa \ln J}{J}\,\mathbf{g}^\sharp
	=f(\alpha)\,\mathbf{g}^\sharp\,,\qquad
	f(\alpha)=\frac{\mu(\alpha-1)+\frac{3\kappa}{2}\ln\alpha}{\alpha^{\frac{3}{2}}}\,.
\end{equation}
Thus, along purely dilatational deformations, the question of global invertibility reduces to the injectivity of the scalar function $\alpha\mapsto f(\alpha)$. Now observe that
\begin{equation}
	\lim_{\alpha\to 0^+} f(\alpha)=-\infty\,,
	\qquad f(1)=0\,, \qquad \lim_{\alpha\to\infty} f(\alpha)=0^+\,.
\end{equation}
Moreover,
\begin{equation}
	f'(1)=\frac{2\mu+3\kappa}{2}>0\,.
\end{equation}
Hence $f(\alpha)$ is increasing near $\alpha=1$, but tends to $0^+$ as $\alpha\to\infty$. Therefore, it must attain a positive maximum at some $\alpha_\ast>1$. Consequently, for any $f$ satisfying $0<f<f(\alpha_\ast)$, there exist two distinct values $\alpha_1\neq \alpha_2$ such that $f(\alpha_1)=f(\alpha_2)=f$. Therefore,
\begin{equation}
	\hat{\boldsymbol{\sigma}}(\alpha_1\,\mathbf{g}^\sharp,\mathbf{g})
	= \hat{\boldsymbol{\sigma}}(\alpha_2\,\mathbf{g}^\sharp,\mathbf{g})\,,  \qquad	\alpha_1\neq \alpha_2\,,
\end{equation}
and hence the constitutive mapping $\mathbf{b}^\sharp\mapsto\boldsymbol{\sigma}$ is not globally injective. 
\end{example}

\subsection{Loading conditions for defining strength}

As we have seen, the notion of strength is based on considering homogeneous states of stress and identifying those that cause a homogeneous body to fracture. One may then ask how such a stress state is realized physically. Are only boundary tractions applied, or are body forces also involved? If one uses the first Piola--Kirchhoff stress as the stress measure in defining strength, and assumes quasistatic loading, then the balance of linear momentum reads $\operatorname{Div}\mathbf{P}+\rho_0\mathbf{B}=\mathbf{0}$, where $\rho_0$ and $\mathbf{B}$ are the material mass density and the body force per unit mass, respectively. Since $P^{aA}{}_{|B}=0$, one immediately has $P^{aA}{}_{|A}=0$, and hence $\mathbf{B}=\mathbf{0}$. Thus, in this setting, the definition of material strength corresponds to loading a homogeneous body only through boundary tractions.

\begin{remark}
It should be noted that a homogeneous state of stress need not be constitutively admissible. For compressible isotropic solids, one may therefore ask which stress states, in the absence of body forces, can be maintained for an arbitrary material in this class. This is precisely the question posed and answered by \citet{Carroll1973}. In analogy with Ericksen's universal deformations,\footnote{Universal deformations are deformations that can be maintained, in the absence of body forces, for every material within a given class \citep{Ericksen1954,Ericksen1955}. Generalizations of this notion to anisotropic solids, and in particular to fiber-reinforced solids, can be found in \citep{YavariGoriely2021,YavariGoriely2023Universal,Yavari2025Universal}.} we call such stress states \textit{universal stresses}. \citet{Carroll1973} showed that universal stresses must be homogeneous (he called them controllable stresses). The converse, however, does not hold: not every homogeneous stress state is constitutively admissible, as the corresponding strain field may be incompatible.
\end{remark}

\begin{remark}
From this point of view, the notion of strength is not tied to the existence of an explicit stress--strain map, but rather to the geometry of admissible stress states. The strength surface is the subset of constitutively admissible homogeneous stress states at which fracture first occurs. This viewpoint separates two distinct questions: the constitutive question of admissibility, determined by the elastic response, and the fracture question of which admissible stress states are critical. 
This distinction becomes particularly important for implicit constitutive theories, for materials with internal constraints, and, more generally, whenever the stress-strain relation is given implicitly.
\end{remark}

Let us consider a hyperelastic brittle solid with energy function $W=W(X,\mathbf{F},\mathbf{G},\mathbf{g})$. In defining the strength function, we restrict attention to uniform stress states produced solely by boundary tractions. This may have a connection with universal deformations. Stress and strain are related through the constitutive equation $\mathbf{P}= \mathbf{g}^\sharp\, \frac{\partial W}{\partial\mathbf{F}}$.
For a general Cauchy elastic solid we have the constitutive equation $\mathbf{P}=\hat{\mathbf{P}}(X,\mathbf{F},\mathbf{G},\mathbf{g})$.
One may then ask which uniform stress states can be maintained by applying only boundary tractions. A related question was posed by \citet{Ericksen1955}: for a homogeneous body made of an arbitrary compressible isotropic hyperelastic material, which deformations can be maintained by boundary tractions alone in the absence of body forces? Such deformations are called \textit{universal deformations}. Ericksen showed that, for this class of materials, the universal deformations are precisely the homogeneous deformations. However, the present problem is different. Here, we fix a single hyperelastic material and consider all constitutively admissible homogeneous stress states, and then identify those at which the brittle solid fractures. Thus, unlike the problem of universal deformations, the present problem is not to characterize all stress states or deformations sustainable for an entire class of materials, but to determine, for a fixed brittle hyperelastic material, the constitutively admissible homogeneous stress states that are critical for fracture. These critical states define the strength hypersurface.

\subsection{Strength function, strength hypersurface, and covariance} \label{Strength-Hypersurface}

Strength has traditionally been formulated as a stress-based criterion, where a particular stress measure is assumed to be uniform at fracture. However, once such a definition is adopted, expressing the same strength criterion in terms of a different stress measure generally requires knowledge of the deformation, and hence the  strain. Therefore, strain dependence cannot be avoided.
Purely stress-based criteria correspond to the special case where the strength function is independent of strain.

Let us assume that the strength hypersurface is written in terms of the first Piola stress as
\begin{equation}
	\mathsf{F}(\mathbf{P},\mathbf{g}\circ\varphi,\mathbf{G})=0\,.
\end{equation}
Note that the first Piola--Kirchhoff stress $\mathbf{P}$ is a two-point tensor with components $P^{aA}$. More precisely, it maps material covectors to spatial vectors, $\mathbf{P}:T_X^*\mathcal{B}\to T_{\varphi(X)}\mathcal{C}$. Since $\mathbf{P}$ connects objects in two different configurations, a scalar-valued function of $\mathbf{P}$ cannot be formed without the use of the material and spatial metrics to lower and contract indices. For this reason, the strength function must depend explicitly on both the material metric $\mathbf{G}$ and the spatial metric $\mathbf{g}\circ\varphi$ in addition to $\mathbf{P}$.
As a generalization of material-frame-indifference, we assume that the strength function is spatially covariant, i.e., it is left invariant under the action of the group $\operatorname{Diff}(\mathcal{S})$: $L_{\xi}\mathsf{F}=\mathsf{F}$. More specifically, this implies that
\begin{equation}
	\mathsf{F}(\xi_*\mathbf{P},\xi_*\mathbf{g},\mathbf{G})
	=\mathsf{F}\left((T\xi)\cdot\mathbf{P},(T\xi)^{-\star}\cdot\mathbf{g}\cdot(T\xi)^{-1},\mathbf{G}\right)
	=\mathsf{F}(\mathbf{P},\mathbf{g}\circ\varphi,\mathbf{G})
	\,.
\end{equation}
Similar to what was done for the energy function in \S\ref{Sec:Elasticity-Covariance}, in the above spatial covariance relation we may choose $\xi$ to be any diffeomorphic extension of $\varphi^{-1}$ to $\mathcal{S}$. Thus,\footnote{\citet{KumarLopezPamies2020} explicitly state that the strength surface must be a function of an objective stress measure and accordingly use the eigenvalues of the Biot stress.}
\begin{equation}
	\mathsf{F}(\mathbf{P},\mathbf{g}\circ\varphi,\mathbf{G})
	= \mathsf{F}\big(\mathbf{F}^{-1}\mathbf{P},\mathbf{F}^\star \mathbf{g}\,\mathbf{F},\mathbf{G}\big)
	= \hat{\mathsf{F}}(\mathbf{S},\mathbf{C}^\flat,\mathbf{G})
	\,.
\end{equation}
We observe that when a stress-based strength function is formulated in terms of $\mathbf{P}$, its representation in terms of $\mathbf{S}$ necessarily depends on $\mathbf{C}$ as well.

More generally, when using the first Piola--Kirchhoff stress $\mathbf{P}$, the strength function can depend on both $\mathbf{P}$ and deformation gradient: 
\begin{equation} \label{Material-Strength-PF}
	\mathsf{F}(\mathbf{P},\mathbf{F},\mathbf{g},\mathbf{G})=0\,.
\end{equation}
In this case, spatial covariance implies that
\begin{equation}
	\mathsf{F}(\mathbf{P},\mathbf{F},\mathbf{g},\mathbf{G}) 
	=\mathsf{F}(\varphi^*\mathbf{P},\varphi^*\mathbf{F},\varphi^*\mathbf{g},\mathbf{G})
	=\mathsf{F}(\mathbf{S},\mathrm{Id}_{T_X\mathcal{B}},\mathbf{C}^\flat,\mathbf{G})\,.
\end{equation}
Now the strength hypersurface in terms of second Piola--Kirchhoff stress is defined as\footnote{Note that $\mathbf{F}^{-1}\mathbf{F}=\operatorname{id}_{T_X\mathcal{B}}$ and $\mathbf{F}\mathbf{F}^{-1}=\operatorname{id}_{T_{\varphi(X)}\mathcal{C}}$.}
\begin{equation}
	\hat{\mathsf{F}}(\mathbf{S},\mathbf{C}^\flat,\mathbf{G})
	\coloneqq\mathsf{F}(\mathbf{S},\mathrm{Id}_{T_X\mathcal{B}},\mathbf{C}^\flat,\mathbf{G})=0\,.
\end{equation}
In what follows we assume that the strength hypersurface is written in terms of the second Piola--Kirchhoff stress as
\begin{equation}
	\hat{\mathsf{F}}(\mathbf{S},\mathbf{C}^\flat,\mathbf{G})=0\,.
\end{equation}

\begin{remark}
Provided that the derivative of $\hat{\mathsf{F}}$ with respect to $(\mathbf{S},\mathbf{C}^\flat)$ does not vanish on its zero set, the equation $\hat{\mathsf{F}}(\mathbf{S},\mathbf{C}^\flat,\mathbf{G})=0$ defines a hypersurface in the $12$-dimensional space of pairs $(\mathbf{S},\mathbf{C}^\flat)$, since both $\mathbf{S}$ and $\mathbf{C}^\flat$ are symmetric second-order tensors and hence each has six independent components. However, $\mathbf{S}$ and $\mathbf{C}^\flat$ are not independent, as they are related by constitutive equations. In implicit elasticity, one has constitutive equations of the form $\boldsymbol{\mathsf{f}}(\mathbf{S},\mathbf{C}^\flat,\mathbf{G})=\mathbf{0}$, which represent six scalar constraints. Under the regularity condition that the derivative of $\boldsymbol{\mathsf{f}}$ with respect to $(\mathbf{S},\mathbf{C}^\flat)$ has rank six, the regular level-set theorem implies that the constitutively admissible states form a $6$-dimensional submanifold $\mathcal{K}$ of the original $12$-dimensional space. The strength hypersurface should therefore be understood as the zero set of the restriction $\hat{\mathsf{F}}|_{\mathcal{K}}$. Provided that the derivative of this restriction does not vanish on its zero set, the strength hypersurface is a $5$-dimensional submanifold of $\mathcal{K}$. Thus, even in implicit elasticity, the notion of strength is attached not to the full space of stress--strain pairs, but to the subset of pairs that are constitutively admissible.
\end{remark}

Recall that
\begin{equation}
	\varphi_*\mathbf{S}=J\boldsymbol{\sigma},\qquad
	\varphi_*\mathbf{C}^\flat=\mathbf{g},\qquad
	\varphi_*\mathbf{G}=\mathbf{c}^\flat\,,
\end{equation}
where $\boldsymbol{\sigma}$ is the Cauchy stress, $\mathbf{g}$ is the spatial metric, and $\mathbf{c}^\flat$ is the spatial metric induced on the body. Let us push forward the strength function to the current configuration (spatial covariance). Covariance of the strength function implies that
\begin{equation}
	\varphi_*\hat{\mathsf{F}}(\mathbf{S},\mathbf{C}^\flat,\mathbf{G})
	= \hat{\mathsf{F}}(\varphi_*\mathbf{S},\varphi_*\mathbf{C}^\flat,\varphi_*\mathbf{G})
	= \hat{\mathsf{F}}(J\boldsymbol{\sigma},\mathbf{g},\mathbf{c}^\flat)\,.
\end{equation}
Thus, the spatial representation of the strength function is given by
\begin{equation} \label{Spatial-Strength-Function}
	\mathsf{f}(\boldsymbol{\sigma},\mathbf{c}^\flat,\mathbf{g})
	\coloneqq \hat{\mathsf{F}}(J\boldsymbol{\sigma},\mathbf{g},\mathbf{c}^\flat)\,.
\end{equation}

In much of the literature, material strength is assumed to be described by a stress-based criterion. However, this assumption is tied to the particular stress measure being used. For example, suppose that the strength hypersurface is described in terms of the Cauchy stress by $\mathsf{f}(\boldsymbol{\sigma},\mathbf{g})=0$. Since the various stress measures are related through the deformation, or equivalently through the strain, any change of stress measure is necessarily accompanied by the corresponding strain measure. More precisely, assuming that $\mathsf{f}$ is independent of $\mathbf{c}^\flat$, from
$\mathsf{f}(\boldsymbol{\sigma},\mathbf{c}^\flat,\mathbf{g})=
\hat{\mathsf{F}}(J\boldsymbol{\sigma},\mathbf{g},\mathbf{c}^\flat)$,
one concludes that $\hat{\mathsf{F}}=\hat{\mathsf{F}}(J\boldsymbol{\sigma},\mathbf{g})$. Pulling this back to the reference configuration and using covariance, one obtains
$\varphi^*\hat{\mathsf{F}}(J\boldsymbol{\sigma},\mathbf{g})=
\hat{\mathsf{F}}(\mathbf{S},\mathbf{C}^\flat)$.
Similarly, if one assumes that $\hat{\mathsf{F}}=\hat{\mathsf{F}}(\mathbf{S},\mathbf{G})$, then one concludes that $\mathsf{f}=\mathsf{f}(\boldsymbol{\sigma},\mathbf{c}^\flat)$.
Thus, a criterion that is stress-based when written in terms of the Cauchy stress will, in general, depend on both stress and strain when written in terms of the second Piola--Kirchhoff stress. 
In this sense, once a strength function is fixed in terms of a given stress measure, rewriting it in terms of another stress measure generally introduces dependence on the corresponding strain.

\begin{remark}
Experimental evidence for viscoelastic elastomers shows that strength depends on both stress and strain, or equivalently, on stress and deformation \citep{Smith1958,Smith1963,Smith1964a,Smith1964b,Knauss1967}. In particular, the experimental results show that failure cannot, in general, be characterized by a unique critical stress, a unique critical strain, or a unique critical strain-energy density. Rather, the critical stress-strain pair along the loading path is what determines fracture. Consistent with this experimental evidence, \citet{KamareiBreedloveLopezPamies2025} used a strength function depending on both stress and strain in modeling fracture of viscoelastic elastomers. 
Here, we have shown that, even for elastic solids, a strength criterion expressed solely in terms of one stress measure generally acquires strain dependence when rewritten using another stress measure.
\end{remark}

\begin{example}
The Drucker--Prager strength function reads \citep{KumarBourdinFrancfortLopezPamies2020}
\begin{equation} \label{Drucker-Prager-Strength-Function}
	\mathcal{F}(\mathbf{S})
	=\sqrt{\mathcal{J}_2}
	+\frac{s_{\mathrm{c}s}-s_{\mathrm{t}s}}{\sqrt{3}\left(s_{\mathrm{c}s}+s_{\mathrm{t}s}\right)}
	\,\mathcal{I}_1
	-\frac{2s_{\mathrm{c}s}s_{\mathrm{t}s}}{\sqrt{3}\left(s_{\mathrm{c}s}+s_{\mathrm{t}s}\right)}
	=0\,,
\end{equation}
where $s_{\mathrm{t}s}$ and $s_{\mathrm{c}s}$ denote the uniaxial tensile and compressive strengths, respectively. The invariants $\mathcal{I}_1$ and $\mathcal{J}_2$ are defined in terms of the principal values $\beta_1,\beta_2,\beta_3$ of the Biot stress by
\begin{equation}
	\mathcal{I}_1 = \beta_1+\beta_2+\beta_3\,,\qquad
	\mathcal{J}_2 = \frac{1}{3}\left[ (\beta_1+\beta_2+\beta_3)^2-\beta_1^2-\beta_2^2-\beta_3^2 \right]
	=\frac{1}{6}\Big[	(\beta_1-\beta_2)^2	+(\beta_2-\beta_3)^2	+(\beta_3-\beta_1)^2	\Big]\,.
\end{equation}
For isotropic solids, recall that the eigenvalues of the Biot stress are related to those of the Cauchy and second Piola--Kirchhoff stresses as given in \eqref{Eigenvalues-sigma-B-S-iso}. Thus, the invariants may be written in terms of the principal values of the Cauchy stress as
\begin{equation}
\begin{aligned}
	\mathcal{I}_1
	&= \lambda_2\lambda_3\,\sigma_1
	+ \lambda_1\lambda_3\,\sigma_2
	+ \lambda_1\lambda_2\,\sigma_3\,,\\
	\mathcal{J}_2
	&= \frac{1}{3}\Big[
	\big(
	\lambda_2\lambda_3\,\sigma_1
	+ \lambda_1\lambda_3\,\sigma_2
	+ \lambda_1\lambda_2\,\sigma_3
	\big)^2 
	- (\lambda_2\lambda_3\,\sigma_1)^2
	- (\lambda_1\lambda_3\,\sigma_2)^2
	- (\lambda_1\lambda_2\,\sigma_3)^2
	\Big]\,.
\end{aligned}
\end{equation}
Similarly, in terms of the principal values of the second Piola--Kirchhoff stress one has
\begin{equation}
\begin{aligned}
	\mathcal{I}_1
	&= \lambda_1 S_1 + \lambda_2 S_2 + \lambda_3 S_3\,,\\
	\mathcal{J}_2
	&= \frac{1}{3}\Big[
	(\lambda_1 S_1 + \lambda_2 S_2 + \lambda_3 S_3)^2
	- (\lambda_1 S_1)^2
	- (\lambda_2 S_2)^2
	- (\lambda_3 S_3)^2
	\Big]\,.
\end{aligned}
\end{equation}
This example shows that a strength function expressed in terms of the Biot stress can be written in terms of the Cauchy stress or the second Piola--Kirchhoff stress only at the expense of introducing explicit dependence on the principal stretches. Thus, when the strength surface is taken to be a function of the Biot stress, it is independent of deformation only in that representation, while its representations in terms of $\boldsymbol{\sigma}$ or $\mathbf{S}$ are deformation dependent.
\end{example}

\begin{remark}
The preceding covariance argument should not be understood as excluding strength criteria that depend only on stress. One may postulate a strength function in terms of a particular stress measure and suppress any explicit strain dependence, as is commonly done in applications. However, when the same physical criterion is rewritten in terms of another stress measure, the transformation between the two stress measures generally introduces the corresponding strain. Thus, stress-only criteria are admissible special representations, while the stress--strain pair provides the natural general domain for a formulation that is independent of the particular stress measure chosen. The covariance argument establishes this general transformation property; it does not assert that every strength function must depend nontrivially on both stress and strain in every representation.
\end{remark}

\begin{remark}
Spatial covariance allows one to pass from a material representation of constitutive relations in terms of $(\mathbf{S},\mathbf{C}^\flat)$ to a spatial representation in terms of $(\boldsymbol{\sigma},\mathbf{c}^\flat)$ via push-forward, where $\mathbf{c}^\flat=\varphi_*\mathbf{G}$ is the push-forward of the material metric. One may equivalently use $(\boldsymbol{\sigma},\mathbf{b}^\sharp)$, since $\mathbf{c}=\mathbf{b}^{-1}$.
However, neither $\mathbf{c}$ nor $\mathbf{b}$ is work-conjugate to the Cauchy stress $\boldsymbol{\sigma}$.\footnote{As a matter of fact, $\boldsymbol{\tau}=J \boldsymbol{\sigma}$ is work conjugate to Hencky’s logarithmic strain \citep{Gurtin1983,Hoger1986,Xiao1997,Neff2020Axiomatic}.} This is in contrast with the pairs $(\mathbf{P},\mathbf{F})$ and $(\mathbf{S},\mathbf{C}^\flat)$, which are work-conjugate by construction. There is no inconsistency here. Covariance is a geometric requirement: it prescribes how tensor fields transform under changes of configuration via push-forward and pull-back operations. Work-conjugacy is an energetic requirement: it is dictated by the stress power identity and singles out pairs that preserve the bilinear pairing of power. In general, push-forward and pull-back operations do not preserve this bilinear pairing, and hence they do not preserve work-conjugacy. Consequently, while $(\boldsymbol{\sigma},\mathbf{c}^\flat)$ is natural from the standpoint of covariance, it is not a work-conjugate pair.
\end{remark}

\begin{remark}
There are strain-energy-based strength functions in the literature. \citet{Beltrami1885} proposed a total strain-energy criterion, in which failure occurs when the stored elastic energy reaches a critical value. A similar strength function has been adopted for elastomers, where nucleation in a flaw-free material is defined through a critical work of fracture \citep{chen2017}. However, such strength functions face two difficulties: (i) they are inconsistent with experimental observations, as failure under different multiaxial loading conditions does not occur at a unique value of the stored strain energy. In particular, for hard brittle solids, where a broad range of multiaxial stress states can be probed experimentally, the observed strength surface is generally not centro-symmetric \citep{KumarBourdinFrancfortLopezPamies2020}, and (ii) they cannot capture cavitation-like nucleation under dominant hydrostatic tensile stresses, as mentioned in the Introduction and discussed explicitly below in Example~\ref{Ex:Hydro-Static-Tension}.
The first limitation has been widely recognized in the literature. \citet{Huber1904} (see the English translation \citep{Huber2004}) formulated the distortion-energy criterion, proposing that failure is governed by the distortional (deviatoric) part of the strain energy rather than the total strain energy. The von Mises--Hencky theory \citep{Mises1913,Hencky1924} is a commonly utilized distortion-energy criterion. 
There are many other similar criteria in the literature, all based on the idea that only a part of the strain energy governs failure.
These criteria have been widely utilized in the computational modeling of fracture with continuum damage models \citep{mazars1986, mazars1989, comi2001, badel2007} and phase-field models \citep{amor2009, miehe2010}.
Strain-energy-based criteria typically do not agree with the experimental observations on strength in the entire stress space \citep{LopezPamies2025}. Nevertheless, all such strength functions are a special case of the general framework discussed here. Let us denote the relevant energy measure by $\mathbb{E}=\mathbb{E}(\mathbf{S},\mathbf{C}^\flat,\mathbf{G})$. Then the corresponding strength function has the form $\mathfrak{F}=\mathfrak{F}(\mathbb{E})$, and hence
\begin{equation}
	\hat{\mathsf{F}}(\mathbf{S},\mathbf{C}^\flat,\mathbf{G})
	=\mathfrak{F}\big(\mathbb{E}(\mathbf{S},\mathbf{C}^\flat,\mathbf{G})\big)
	\,.
\end{equation}
Thus, energy-based strength criteria are a special case of stress--strain-based strength criteria.
\end{remark}

\begin{remark}
The strength hypersurface represents a physical condition for fracture and therefore must be independent of the particular configuration used to describe the body. 
A deformation merely changes the representation of stress, strain, and metric tensors through push-forward and pull-back operations, but it does not alter the underlying physical state of the material.
Consequently, if a given stress state causes fracture when expressed in the reference configuration, the same physical state expressed in the current configuration must also correspond to fracture. This requirement implies that the strength function must be spatially covariant under the change of configuration induced by the deformation map. In other words, the functional relation defining the strength hypersurface must be preserved when the arguments are transformed by the natural push-forward and pull-back associated with the motion. Spatial covariance therefore expresses the representation-independence of the fracture criterion: it guarantees that the same set of physical failure states is described consistently in both the reference and the spatial configurations.
\end{remark}

\subsubsection{The manifold of constitutively admissible homogeneous stresses}

Let us assume that strength is defined in terms of uniform Cauchy stress $\boldsymbol{\sigma}$. It should be emphasized that not all homogeneous stress fields are constitutively admissible. First, note that the set of all symmetric $3\times 3$ matrices, $\mathrm{Sym}(3)$, is a $6$-dimensional real vector space and hence a smooth $6$-dimensional manifold. In the present explicit isotropic compressible setting, the set of constitutively admissible homogeneous stress fields is a subset of $\mathrm{Sym}(3)$, i.e., the set of all homogeneous stresses $\boldsymbol{\sigma}$ such that there exists a deformation $\varphi$ with associated $\mathbf{b}^\sharp$ satisfying the compatibility conditions and the constitutive relation $\boldsymbol{\sigma}=\hat{\boldsymbol{\sigma}}(\mathbf{b}^\sharp,\mathbf{g})$. Let us denote this set by $\mathcal{A}$. Note that the set $\mathcal{A}$ is non-empty. To show this, consider a homogeneous dilatational deformation $\varphi(X)=\alpha\,X$, $\alpha>0$. Then the deformation gradient is uniform with components $F^a{}_A=\alpha\,\delta^a_A$, and hence the left Cauchy--Green tensor is written as $\mathbf{b}^\sharp=\mathbf{F}\mathbf{G}^\sharp\mathbf{F}^\star=\alpha^{2}\,\mathbf{g}^{\sharp}$. The constitutive relation $\boldsymbol{\sigma}=\hat{\boldsymbol{\sigma}}(\mathbf{b}^\sharp,\mathbf{g})$ then implies that the corresponding Cauchy stress is uniform. Therefore, $\boldsymbol{\sigma}\in\mathcal{A}$, and hence $\mathcal{A}\neq\varnothing$.

\begin{prop}
Let $\hat{\boldsymbol{\sigma}}:\mathrm{Sym}^{+}(3)\to\mathrm{Sym}(3)$ be the constitutive map of an isotropic compressible elastic solid, and let
\begin{equation}
	\mathcal{A}=\hat{\boldsymbol{\sigma}}\big(\mathrm{Sym}^{+}(3)\big)\subset \mathrm{Sym}(3)\,,
\end{equation}
denote the set of constitutively admissible homogeneous Cauchy stresses. Suppose that, for every $\mathbf{b}^\sharp\in \mathrm{Sym}^{+}(3)$ under consideration, the derivative map
\begin{equation}
	D\hat{\boldsymbol{\sigma}}(\mathbf{b}^\sharp,\mathbf{g}):T_{\mathbf{b}^\sharp}\mathrm{Sym}^{+}(3)
	\to T_{\hat{\boldsymbol{\sigma}}(\mathbf{b}^\sharp,\mathbf{g})}\mathrm{Sym}(3)\,,
\end{equation}
has full rank $6$. Then $\mathcal{A}$ is a smooth $6$-dimensional manifold.
\end{prop}

\begin{proof}
Recall that $\mathrm{Sym}(3)$ is a $6$-dimensional smooth manifold, and that $\mathrm{Sym}^{+}(3)$, being an open subset of $\mathrm{Sym}(3)$, is also a $6$-dimensional smooth manifold. Since the derivative map
\begin{equation}
	D\hat{\boldsymbol{\sigma}}(\mathbf{b}^\sharp,\mathbf{g}):T_{\mathbf{b}^\sharp}\mathrm{Sym}^{+}(3)
	\to T_{\hat{\boldsymbol{\sigma}}(\mathbf{b}^\sharp,\mathbf{g})}\mathrm{Sym}(3)\,,
\end{equation}
has full rank $6$ for every $\mathbf{b}^\sharp\in\mathrm{Sym}^{+}(3)$ under consideration, it is an isomorphism. Equivalently, $\det D\hat{\boldsymbol{\sigma}}(\mathbf{b}^\sharp,\mathbf{g})\neq 0$.\footnote{A closely related non-degeneracy condition appears in \citep{Neff2025}, where the authors use the condition $\det D_{\mathbf B}\boldsymbol{\sigma}(\mathbf B)\neq 0$ and relate it to the invertibility of the associated tangent operator; see in particular their equations (1.10) and (2.69).} Therefore, by the inverse function theorem \citep{Lee2013}, for every $\mathbf{b}^\sharp\in\mathrm{Sym}^{+}(3)$ there exist open neighborhoods $\mathcal{U}_{\mathbf{b}^\sharp}\subset\mathrm{Sym}^{+}(3)$ and $\mathcal{V}_{\hat{\boldsymbol{\sigma}}(\mathbf{b}^\sharp,\mathbf{g})}\subset\mathrm{Sym}(3)$ such that the restriction
\begin{equation}
	\hat{\boldsymbol{\sigma}}\big|_{\mathcal{U}_{\mathbf{b}^\sharp}}:\mathcal{U}_{\mathbf{b}^\sharp}
	\to \mathcal{V}_{\hat{\boldsymbol{\sigma}}(\mathbf{b}^\sharp,\mathbf{g})}\,,
\end{equation}
is a diffeomorphism. In particular, for every $\mathbf{b}^\sharp\in\mathrm{Sym}^{+}(3)$, the image of a neighborhood of $\mathbf{b}^\sharp$ is an open subset of $\mathrm{Sym}(3)$. Hence, the image
\begin{equation}
	\mathcal{A}=\hat{\boldsymbol{\sigma}}\big(\mathrm{Sym}^{+}(3)\big)\subset\mathrm{Sym}(3)\,,
\end{equation}
is locally open in $\mathrm{Sym}(3)$. Equivalently, every point of $\mathcal{A}$ has a neighborhood, open in $\mathcal{A}$, that is also open in $\mathrm{Sym}(3)$. Since $\mathrm{Sym}(3)$ is a smooth $6$-dimensional manifold, it follows that $\mathcal{A}$ inherits the structure of a smooth $6$-dimensional manifold.
\end{proof}

This proposition shows that, under the non-degeneracy assumption $\det D\hat{\boldsymbol{\sigma}}(\mathbf{b}^\sharp,\mathbf{g})\neq 0$, the set of constitutively admissible homogeneous Cauchy stresses is not an arbitrary subset of $\mathrm{Sym}(3)$, but a smooth $6$-dimensional manifold. Thus, the strength hypersurface should be understood as a hypersurface in this constitutively admissible stress manifold. The non-degeneracy assumption is equivalent to local invertibility of the constitutive map $\mathbf{b}^\sharp\mapsto \hat{\boldsymbol{\sigma}}(\mathbf{b}^\sharp,\mathbf{g})$, and has a clear physical meaning: the elastic response is locally non-degenerate, in the sense that small changes in strain produce unique changes in stress, and there are no non-trivial nearby strain states that leave the stress unchanged. This may be viewed as a natural constitutive regularity condition for a reasonable elastic solid.

Assuming that the tangent map $D_{\mathbf{b}^\sharp}\boldsymbol{\sigma}$ is invertible is a strong local non-degeneracy condition, as it implies that no nontrivial perturbation of $\mathbf{b}^\sharp$ produces a vanishing increment of stress. In particular, it guarantees local uniqueness of the constitutive response.
Let $\Lin(T_X\mathcal{B},T_x\mathcal{C})$ denote the space of linear maps from $T_X\mathcal{B}$ to $T_x\mathcal{C}$, and $\Lin^+(T_X\mathcal{B},T_x\mathcal{C})$ denote the subset of orientation-preserving linear isomorphisms. An elastic energy $W(\mathbf{F},\mathbf{G},\mathbf{g})$ is said to be rank-one convex if, for every $\mathbf{F}\in \Lin^+(T_X\mathcal{B},T_x\mathcal{C})$ and every rank-one tensor $\mathbf{a}\otimes\mathbf{N}^\flat$, the function $t\mapsto W(\mathbf{F}+t\,\mathbf{a}\otimes\mathbf{N}^\flat,\mathbf{G},\mathbf{g})$ is convex for all $t$ for which $\mathbf{F}+t\,\mathbf{a}\otimes\mathbf{N}^\flat\in \Lin^+(T_X\mathcal{B},T_x\mathcal{C})$ \citep{Ball1976}. If $W$ is of class $C^2$, this is equivalent to the Legendre--Hadamard condition, or strong ellipticity, namely
\begin{equation}
	D^2 W(\mathbf{F},\mathbf{G},\mathbf{g}) 
	\left[\mathbf{a}\otimes\mathbf{N}^\flat,\mathbf{a}\otimes\mathbf{N}^\flat \right] 
	\geq 0\,,
	\qquad\forall\,\mathbf{a}\in T_x\mathcal{C}\setminus\{\mathbf{0}\}\,,\ 
	\forall\,\mathbf{N}^\flat\in T_X^*\mathcal{B}\setminus\{\mathbf{0}\}\,,
\end{equation}
and in components
\begin{equation}
	\frac{\partial^2 W}{\partial F^a{}_A \partial F^b{}_B}\,a^a \,N_A\, a^b \,N_B \geq 0\,,
	\qquad\forall\,\mathbf{a}\in T_x\mathcal{C}\setminus\{\mathbf{0}\}\,,\ 
	\forall\,\mathbf{N}^\flat\in T_X^*\mathcal{B}\setminus\{\mathbf{0}\}\,.
\end{equation}
Thus, strong ellipticity requires positivity only along rank-one directions.

The restriction to rank-one directions is not arbitrary. Rank-one perturbations correspond to deformation gradients that are compatible across surfaces and hence describe the simplest admissible discontinuities through the Hadamard jump condition. Loss of strong ellipticity along such directions signals loss of ellipticity, the onset of material instability, and the possible formation of microstructure. Thus, strong ellipticity and local invertibility of $D_{\mathbf{b}^\sharp}\boldsymbol{\sigma}$ are distinct constitutive conditions. Strong ellipticity concerns the strict positivity of the elasticity quadratic form on nonzero rank-one tensors, whereas invertibility of $D_{\mathbf{b}^\sharp}\boldsymbol{\sigma}$ concerns the local injectivity of the stress--strain relation. In general, neither condition implies the other. For a $C^2$ strain-energy function, rank-one convexity is equivalent to the nonnegativity of the elasticity quadratic form on rank-one tensors, while strong ellipticity requires strict positivity on every nonzero rank-one tensor. Strong ellipticity is a fundamental local stability requirement in elasticity.\footnote{The local invertibility of constitutive stress--strain relations and related monotonicity properties have been investigated extensively by \citet{Neff2015I,Neff2015II}. In particular, the true-stress--true-strain monotonicity condition considered therein implies invertibility of the corresponding stress--strain relation.}

\subsubsection{Constitutively admissible homogeneous stresses in isotropic solids}

For isotropic solids, all one needs for describing the strength properties are the principal stresses. For a hyperelastic solid the Cauchy principal stresses $(\sigma_1,\sigma_2,\sigma_3)$ are related to principal stretches $(\lambda_1,\lambda_2,\lambda_3)$ as \citep{Ogden1984}
\begin{equation}
	\sigma_i=\frac{\lambda_i}{\lambda_1\lambda_2\lambda_3}\,\frac{\partial W}{\partial \lambda_i}\,,\qquad i=1,2,3\,.
\end{equation}
Thus, one has the principal stress map
\begin{equation}
\begin{aligned}
	\hat{\boldsymbol{\sigma}}:\mathbb{R}_{+}^{3} &\to \mathbb{R}^{3} \\
	(\lambda_1,\lambda_2,\lambda_3) &\mapsto (\sigma_1,\sigma_2,\sigma_3)\,.
\end{aligned}
\end{equation}
Suppose that, for every $(\lambda_1,\lambda_2,\lambda_3)\in\mathbb{R}_{+}^{3}$ under consideration, the derivative map
\begin{equation}
	D\hat{\boldsymbol{\sigma}}(\lambda_1,\lambda_2,\lambda_3)
	:T_{(\lambda_1,\lambda_2,\lambda_3)}\mathbb{R}_{+}^{3}
	\to T_{\hat{\boldsymbol{\sigma}}(\lambda_1,\lambda_2,\lambda_3)}\mathbb{R}^{3}\,,
\end{equation}
has full rank $3$. Equivalently, the Jacobian matrix
\begin{equation}
	\left[\frac{\partial \sigma_i}{\partial \lambda_j}\right]\,,
\end{equation}
is invertible. This is precisely the condition that the map $(\lambda_1,\lambda_2,\lambda_3)\mapsto(\sigma_1,\sigma_2,\sigma_3)$ is locally invertible. By the inverse function theorem \citep{Lee2013}, its image is locally open in $\mathbb{R}^{3}$. Therefore, the set of constitutively admissible triples of principal stresses is a smooth $3$-dimensional manifold. Consequently, for isotropic solids, the strength hypersurface may be viewed equivalently as a hypersurface in this $3$-dimensional manifold of admissible principal stresses.
\begin{figure}[t!]
\centering
\includegraphics[width=0.5\textwidth]{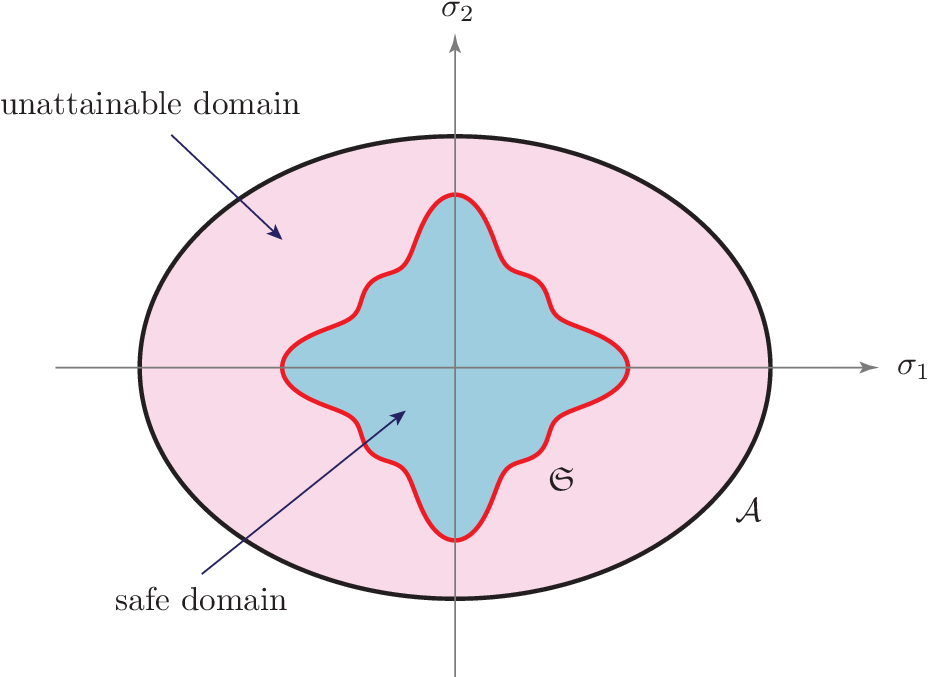}
\vspace*{0.10in}
\caption{Schematic illustration of the constitutively admissible stress manifold $\mathcal{A}$ (bounded by the black ellipse) and the strength hypersurface $\mathfrak{S}$ (red curve) in principal stress space. The strength hypersurface is a hypersurface of the constitutively admissible stress manifold and separates safe and unsafe stress states.}
\label{Strength-Schematics}
\end{figure}

\subsubsection{The strength hypersurface}

In the following we restrict attention to a subclass of strength functions that depend only on stress, since all existing examples in the literature belong to this class.
Let us consider a homogeneous solid in a state of spatially uniform Cauchy stress $\boldsymbol{\sigma}$. We use the Cauchy stress without loss of generality; the same conclusions hold if one starts with any other stress measure. Suppose the set of all uniform stresses under which the body fractures is described by the equation $\mathsf{f}(\boldsymbol{\sigma},\mathbf{g})=0$.\footnote{Note that a scalar function of a spatial second-order tensor necessarily involves the spatial metric $\mathbf{g}$ through index contractions.} Assume that $\mathsf{f}:\mathbb{R}^6\to\mathbb{R}$ is continuous and define the corresponding strength hypersurface:
\begin{equation} \label{Strength-Function}
	\mathfrak{S} =\left\{\boldsymbol{\sigma}\in \mathcal{A}
	:\mathsf{f}(\boldsymbol{\sigma},\mathbf{g})=0 \right\}\,.
\end{equation}

Any physical state of stress in the body is assumed to satisfy the constraint $\mathsf{f}(\boldsymbol{\sigma},\mathbf{g})\le 0$, and the set $\{\mathsf{f}<0\}$ is the set of \textit{safe stresses}. The domain $\{\mathsf{f}>0\}$ is assumed to be physically unattainable, see Fig.~\ref{Strength-Schematics}. Clearly, the domains $\{\mathsf{f}<0\}$ and $\{\mathsf{f}>0\}$ are disjoint. Consequently, the strength hypersurface $\mathfrak{S}$ partitions the stress space into two disjoint sets corresponding to safe and unattainable stresses. Note that every stress state on the strength hypersurface must be a threshold state separating safe stresses from failing stresses. In physical terms this means the following: Let $\boldsymbol{\sigma}_0\in\mathfrak{S}$ be a critical stress state satisfying $\mathsf{f}(\boldsymbol{\sigma}_0,\mathbf{g})=0$. Then arbitrarily small perturbations of the stress must produce both safe and failing states. More precisely, for every $\varepsilon>0$ there exist stresses $\boldsymbol{\sigma}_-\in\mathbb{R}^6$ and $\boldsymbol{\sigma}_+\in\mathbb{R}^6$ such that\footnote{We denote by $\|\cdot\|$ any norm on $\mathrm{Sym}(3)\cong\mathbb{R}^6$. Since $\mathrm{Sym}(3)$ is a finite-dimensional vector space, all norms are equivalent. Therefore, the particular choice of norm is immaterial, as it induces the same topology and hence the same notion of neighborhoods.}
\begin{equation}
	\|\boldsymbol{\sigma}_--\boldsymbol{\sigma}_0\|<\varepsilon,\qquad
	\|\boldsymbol{\sigma}_+-\boldsymbol{\sigma}_0\|<\varepsilon\,,
\end{equation}
with
\begin{equation}
	\mathsf{f}(\boldsymbol{\sigma}_-,\mathbf{g})<0\,,\qquad
	\mathsf{f}(\boldsymbol{\sigma}_+,\mathbf{g})>0\,.
\end{equation}
In other words, every neighborhood of $\boldsymbol{\sigma}_0$ contains stresses that do not fracture the material and stresses that do fracture the material. This expresses the physical idea that $\mathfrak{S}$ is a true threshold between safe and failing stresses.

\begin{prop}\label{prop:strength-surface-boundary}
Let $\mathsf{f}:\mathbb{R}^6\to\mathbb{R}$ be continuous and let $\mathfrak{S}$ be defined by $\mathsf{f}(\boldsymbol{\sigma},\mathbf{g})=0$. Then the sets $\{\mathsf{f}<0\}$ and $\{\mathsf{f}>0\}$ are open and disjoint, and
\begin{equation}
	\partial\{\mathsf{f}<0\}\subset\mathfrak{S}\,,\qquad \partial\{\mathsf{f}>0\}\subset\mathfrak{S}\,.
\end{equation}
If, in addition, the threshold property stated above holds for every $\boldsymbol{\sigma}_0\in\mathfrak{S}$, then
\begin{equation}
	\partial\{\mathsf{f}<0\}=\partial\{\mathsf{f}>0\}=\mathfrak{S}\,.
\end{equation}
\end{prop}

\begin{proof}
Since $\mathsf{f}$ is continuous, the preimages $\{\mathsf{f}<0\}=\mathsf{f}^{-1}((-\infty,0))$ and $\{\mathsf{f}>0\}=\mathsf{f}^{-1}((0,\infty))$ are open. They are disjoint because $\mathsf{f}$ cannot be simultaneously negative and positive.

To prove $\partial\{\mathsf{f}<0\}\subset\mathfrak{S}$, let $\boldsymbol{\sigma}_0\in\partial\{\mathsf{f}<0\}$. Suppose, for contradiction, that $\mathsf{f}(\boldsymbol{\sigma}_0)>0$. By continuity of $\mathsf{f}$ at $\boldsymbol{\sigma}_0$, there exists $\varepsilon>0$ such that $\mathsf{f}(\boldsymbol{\sigma},\mathbf{g})>0$ for all $\boldsymbol{\sigma}$ with $\|\boldsymbol{\sigma}-\boldsymbol{\sigma}_0\|<\varepsilon$. Hence, $B_\varepsilon(\boldsymbol{\sigma}_0)\cap\{\mathsf{f}<0\}=\emptyset$, where $B_\varepsilon(\boldsymbol{\sigma}_0)=\{\boldsymbol{\sigma}\in\mathbb{R}^n:\|\boldsymbol{\sigma}-\boldsymbol{\sigma}_0\|<\varepsilon\}$ denotes the open ball of radius $\varepsilon$ centered at $\boldsymbol{\sigma}_0$. This implies that $\boldsymbol{\sigma}_0\notin\partial\{\mathsf{f}<0\}$, a contradiction. Therefore $\mathsf{f}(\boldsymbol{\sigma}_0)\le0$.
Next suppose, again for contradiction, that $\mathsf{f}(\boldsymbol{\sigma}_0)<0$. By continuity, there exists $\varepsilon>0$ such that $\mathsf{f}(\boldsymbol{\sigma},\mathbf{g})<0$ whenever $\|\boldsymbol{\sigma}-\boldsymbol{\sigma}_0\|<\varepsilon$. Hence $B_\varepsilon(\boldsymbol{\sigma}_0)\subset\{\mathsf{f}<0\}$, so $\boldsymbol{\sigma}_0$ is an interior point of $\{\mathsf{f}<0\}$ and cannot lie in its boundary, again a contradiction. Thus $\mathsf{f}(\boldsymbol{\sigma}_0)=0$, i.e. $\boldsymbol{\sigma}_0\in\mathfrak{S}$, and therefore $\partial\{\mathsf{f}<0\}\subset\mathfrak{S}$. The inclusion $\partial\{\mathsf{f}>0\}\subset\mathfrak{S}$ is proved identically, with the inequalities reversed.

Finally, assume the threshold property. Let $\boldsymbol{\sigma}_0\in\mathfrak{S}$. Then every neighborhood of $\boldsymbol{\sigma}_0$ contains a point of $\{\mathsf{f}<0\}$ and a point of $\{\mathsf{f}>0\}$. Hence every neighborhood of $\boldsymbol{\sigma}_0$ intersects both $\{\mathsf{f}<0\}$ and its complement, so $\boldsymbol{\sigma}_0\in\partial\{\mathsf{f}<0\}$. Similarly, $\boldsymbol{\sigma}_0\in\partial\{\mathsf{f}>0\}$. Therefore $\mathfrak{S}\subset\partial\{\mathsf{f}<0\}$ and $\mathfrak{S}\subset\partial\{\mathsf{f}>0\}$. Combined with the previously proved inclusions $\partial\{\mathsf{f}<0\}\subset\mathfrak{S}$ and $\partial\{\mathsf{f}>0\}\subset\mathfrak{S}$, we conclude that $\partial\{\mathsf{f}<0\}=\partial\{\mathsf{f}>0\}=\mathfrak{S}$.
\end{proof}

\begin{example}
This example illustrates that, even for spatially uniform stress fields, constitutive admissibility imposes nontrivial constraints. In particular, we show that not every uniform stress field can be realized by a deformation.
The implicit constitutive equations for an isotropic solid are rewritten as \citep{RivlinEricksen1955,Morgan1966,YavariGoriely2024}
\begin{equation} \label{Implicit-Constitutive-Equation-Isotropic}
\begin{aligned}
	\boldsymbol{\mathsf{f}}(\boldsymbol{\sigma},\mathbf{b})
	& = \alpha_0\,\mathbf{g}^\sharp
	+\alpha_1 \boldsymbol{\sigma}+\alpha_2 \boldsymbol{\sigma}^2
	+\alpha_3\,\mathbf{b}^\sharp+\alpha_4\,\mathbf{c}^{\sharp} 
	+\alpha_5 \left(\boldsymbol{\sigma}\mathbf{b}^\sharp
	+\mathbf{b}^\sharp\boldsymbol{\sigma}\right)
	+\alpha_6 \left(\boldsymbol{\sigma}^2\mathbf{b}^\sharp
	+\mathbf{b}^\sharp\boldsymbol{\sigma}^2\right)\\
	& \quad
	+\alpha_7 \left(\boldsymbol{\sigma}\mathbf{c}^{\sharp}
	+\mathbf{c}^{\sharp}\boldsymbol{\sigma}\right)
	+\alpha_8 \left(\boldsymbol{\sigma}^2\mathbf{c}^{\sharp}
	+\mathbf{c}^{\sharp}\boldsymbol{\sigma}^2\right)
	=\mathbf{0}\,,
\end{aligned}
\end{equation}
where $\alpha_i$, $i=0,\hdots,8$, are functions of the following ten invariants 
\begin{equation} \label{I-Invariants-b}
\begin{aligned}
	& I_1=\operatorname{tr}\boldsymbol{\sigma}\,,\quad
	I_2=\operatorname{tr}\boldsymbol{\sigma}^2\,,\quad
	I_3=\operatorname{tr}\boldsymbol{\sigma}^3\,,\quad
	I_4=\operatorname{tr}\mathbf{b}^{\sharp} \,,\quad
	I_5=\frac{1}{2}\left[ I_4^2- \operatorname{tr}\mathbf{b}^{2\sharp}\right]\,,\quad
	I_6=\det\mathbf{b}^{\sharp} \,,\quad \\
	& I_7=\operatorname{tr}\left(\boldsymbol{\sigma}\mathbf{b}^{\sharp}\right)\,,\quad
	I_8=\operatorname{tr}\left(\boldsymbol{\sigma}\mathbf{b}^{2\sharp}\right)\,,\quad
	I_9=\operatorname{tr}\left(\boldsymbol{\sigma}^2\mathbf{b}^{\sharp}\right)\,,\quad
	I_{10}=\operatorname{tr}\left(\boldsymbol{\sigma}^2\mathbf{b}^{2\sharp}\right)
	\,.
\end{aligned}
\end{equation}
Suppose that the Cauchy stress field is uniform, and hence, its principal values $\sigma_i$, $i=1,2,3$, are constants throughout the body. Assume $\alpha_5=\alpha_6=\alpha_7=\alpha_8=0$, i.e.,
\begin{equation}
	\alpha_2 \boldsymbol{\sigma}^2+\alpha_1 \boldsymbol{\sigma}
	=-\alpha_0\,\mathbf{g}^\sharp-\alpha_3\,\mathbf{b}^\sharp-\alpha_4\,\mathbf{c}^{\sharp} \,.
\end{equation}
We ask whether there exists a deformation corresponding to this uniform stress field. 
Let $\lambda_i=\lambda_i(x)$, $i=1,2,3$, denote the principal stretches of the deformation at a point $x\in\mathcal{C}$.
We denote the eigenvalues of $\mathbf{b}^\sharp(x)$ by $\lambda_i^2(x)$.
For this class of materials, $\boldsymbol{\sigma}$ and $\mathbf{b}^\sharp$ have the same eigenvectors and can be diagonalized simultaneously.\footnote{It should be emphasized that for the more general class of elastic solids with constitutive equations \eqref{Implicit-Constitutive-Equation-Isotropic}, $\boldsymbol{\sigma}$ and $\mathbf{b}^\sharp$ do not have the same eigenvectors, in general.}
Thus, at each point $x\in\mathcal{C}$, with respect to the common eigenbasis one has
\begin{equation}
	\alpha_3\,\lambda_i^2(x)+\alpha_4\,\lambda_i^{-2}(x)
	= -\alpha_0 - \alpha_1 \sigma_i - \alpha_2 \sigma_i^2
	\,,\qquad i=1,2,3 \,.
\end{equation}
Multiplying by $\lambda_i^2(x)$, we obtain
\begin{equation}
	\alpha_3\,\lambda_i^4(x)
	+\left(\alpha_0 + \alpha_1 \sigma_i + \alpha_2 \sigma_i^2\right)\lambda_i^2(x)
	+\alpha_4 = 0
	\,,\qquad i=1,2,3 \,.
\end{equation}
Note that the coefficients $\alpha_j=\alpha_j(I_1,\hdots,I_{10})$ are not constants in general, because the invariants $I_4,\hdots,I_{10}$ depend on $\mathbf{b}^\sharp(x)$ and hence on the position-dependent stretches $\lambda_i(x)$. Therefore, the above equations are not quadratic equations with constant coefficients for $\lambda_i^2(x)$; rather, they form a coupled nonlinear algebraic system that must be satisfied pointwise. More precisely, for each point $x\in\mathcal{C}$ one must solve
\begin{equation}
	\Psi_i \big(\lambda_1^2(x),\lambda_2^2(x),\lambda_3^2(x);\sigma_1,\sigma_2,\sigma_3 \big)=0
	\,,\qquad i=1,2,3 \,,
\end{equation}
where
\begin{equation}
	\Psi_i
	:=\alpha_3\,\lambda_i^4
	+\left(\alpha_0 + \alpha_1 \sigma_i + \alpha_2 \sigma_i^2\right)\lambda_i^2
	+\alpha_4 \,,
\end{equation}
and each $\alpha_j$ is evaluated at the invariants \eqref{I-Invariants-b}. Thus, constitutive admissibility of a uniform stress field requires that, at every point $x\in\mathcal{C}$, this nonlinear algebraic system admit a real positive solution $(\lambda_1^2(x),\lambda_2^2(x),\lambda_3^2(x))$.
If, for a given uniform stress field $\boldsymbol{\sigma}$, there exists a point $x\in\mathcal{C}$ such that the above system has no real positive solution, then no deformation can realize this stress field. Hence, not every uniform stress field is constitutively admissible.
A concrete subclass is obtained by assuming that $\alpha_i$, $i=0,\hdots,4$, are constants and that $\alpha_1=\alpha_2=0$, $\alpha_3=\alpha_4=1$, and $\alpha_0=-c$, where $c$ is a constant. In this case, the constitutive equations reduce to
\begin{equation}
	\lambda_i^2(x)+\lambda_i^{-2}(x)=c	\,,\qquad i=1,2,3 \,.
\end{equation}
Multiplying by $\lambda_i^2(x)$, one obtains
\begin{equation}
	\lambda_i^4(x)-c\,\lambda_i^2(x)+1=0	\,,\qquad i=1,2,3 \,.
\end{equation}
This quadratic equation admits real positive solutions for $\lambda_i^2(x)$ if and only if $c^2-4\geq 0$, i.e., if and only if $c\geq 2$ or $c\leq -2$. Since $\lambda_i^2(x)+\lambda_i^{-2}(x)\geq 2$ for every positive $\lambda_i^2(x)$, only the case $c\geq 2$ is physically admissible. Therefore, if $0<c<2$, there is no real positive solution for $\lambda_i^2(x)$ at any point $x\in\mathcal{C}$. Thus, for this subclass, no deformation can realize such a uniform stress field, and hence not every uniform stress field is constitutively admissible.
\end{example}

\subsubsection{Star-shapedness of the safe domain}

Next, we show that, under a natural assumption, the safe domain is star-shaped with respect to the zero-stress state (the origin in stress space).

\begin{defi}[Star-shaped domain]
Let $\Omega\subset\mathbb{R}^n$ with $n\in\{2,3\}$. We say that $\Omega$ is \textit{star-shaped with respect to the origin} if $\mathbf{0}\in\Omega$ and if for every $\boldsymbol{x}\in\Omega$ and every $t\in[0,1]$ one has $t\boldsymbol{x}\in\Omega$ \citep{Rockafellar1970}. Equivalently, for every $\boldsymbol{x}\in\Omega$ the line segment $\{t\boldsymbol{x}:t\in[0,1]\}$ is contained in $\Omega$.
In $n=2$ we write $\boldsymbol{x}=(\sigma_1,\sigma_2)$, while in $n=3$ we write $\boldsymbol{x}=(\sigma_1,\sigma_2,\sigma_3)$.
\end{defi}

\begin{defi}[Proportional reduction safety]
Let $\mathsf{f}:\mathbb{R}^n\to\mathbb{R}$ and let $\Omega=\{\boldsymbol{\sigma}\in\mathbb{R}^n:\mathsf{f}(\boldsymbol{\sigma},\mathbf{g})<0\}$ denote the safe domain. We say that $\mathsf{f}$ satisfies the \textit{proportional reduction safety property} if for every safe stress $\boldsymbol{\sigma}\in\Omega$ and every $t\in[0,1)$, the proportionally reduced stress $t\boldsymbol{\sigma}$ is also safe, i.e.,
\begin{equation}\label{eq:prop-red-safe}
	\mathsf{f}(\boldsymbol{\sigma},\mathbf{g})<0
	~\Rightarrow~
	\mathsf{f}(t\boldsymbol{\sigma},\mathbf{g})<0\,,
	\qquad \forall\,\boldsymbol{\sigma}\in\mathbb{R}^n\,,\ \forall\,t\in[0,1)\,.
\end{equation}
\end{defi}

\begin{prop}\label{prop:star-shaped-safe}
Let $\mathsf{f}:\mathbb{R}^n\to\mathbb{R}$ with $n\in\{2,3\}$ be a (continuous) strength function, and let $\Omega$ be its safe domain. If $\mathsf{f}$ satisfies the proportional reduction safety property, then $\Omega$ is star-shaped with respect to the origin.
\end{prop}

\begin{proof}
Let $\boldsymbol{\sigma}\in\Omega$. Since $\boldsymbol{\sigma}\in\Omega$, one has $\mathsf{f}(\boldsymbol{\sigma},\mathbf{g})<0$. Hence, by the proportional reduction safety property, $t\boldsymbol{\sigma}\in\Omega$ for every $t\in[0,1)$. Since also $\boldsymbol{\sigma}\in\Omega$, it follows that $t\boldsymbol{\sigma}\in\Omega$ for every $t\in[0,1]$. Thus, $\Omega$ is star-shaped with respect to the origin.
\end{proof}

One should note that this conclusion cannot be deduced from continuity of $\mathsf{f}$ alone. Star-shapedness encodes a constitutive requirement on strength: if a stress state is safe, then every proportionally reduced stress state along the segment joining that state to the origin must also be safe. The proportional reduction safety property is a mathematical expression of this requirement.

\begin{remark}
Every convex domain is star-shaped with respect to each of its points. The converse is false: a star-shaped domain need not be convex. It should be emphasized that star-shapedness is a much weaker condition than convexity. Convexity is therefore not a consequence of the proportional reduction safety property. Indeed, experimental evidence indicates that non-convex safe domains may occur \citep{Ashkenazi1965}. However, such domains are still star-shaped with respect to the origin. See also Fig.~\ref{fig:strength-surfaces}.
\end{remark}

\subsubsection{Topology of the strength hypersurface}

Let $\mathsf{f}:\mathcal{A}\to\mathbb{R}$ be a smooth strength function.
From a physical point of view, the strength hypersurface \eqref{Strength-Function} represents the set of critical stress states separating safe and failing states. It is therefore natural to expect it to be a smooth hypersurface and that it be bounded, since sufficiently large admissible stresses should not remain safe. We now make these requirements precise.

\begin{prop}
Let $\mathsf{f}:\mathcal{A}\to\mathbb{R}$ be a smooth function, where $\mathcal{A}\subset\mathbb{R}^6$ is the manifold of constitutively admissible homogeneous stresses. Suppose that $0$ is a regular value of $\mathsf{f}$ and that either
\begin{equation} \label{Large-Stress-Condition-1}
	\mathsf{f}(\boldsymbol{\sigma})>0\,,\qquad
	\forall\,\boldsymbol{\sigma}\in\mathcal{A}\ \text{with}\ \|\boldsymbol{\sigma}\|>R\,,
\end{equation}
for some $R>0$, or, more strongly,\footnote{This assumption is the standard coercivity condition. It is stronger than the first, since $\mathsf{f}(\boldsymbol{\sigma})\to +\infty$ as $\|\boldsymbol{\sigma}\|\to\infty$ implies that there exists $R>0$ such that $\mathsf{f}(\boldsymbol{\sigma})>0$ whenever $\|\boldsymbol{\sigma}\|>R$.}
\begin{equation} \label{Large-Stress-Condition-2}
	\mathsf{f}(\boldsymbol{\sigma})\to +\infty
	\quad\text{as}\quad
	\|\boldsymbol{\sigma}\|\to\infty\,,\qquad
	\boldsymbol{\sigma}\in\mathcal{A}\,.
\end{equation}
Then the strength hypersurface $\mathfrak{S}=\mathsf{f}^{-1}(0)$ is a smooth compact embedded hypersurface in $\mathcal{A}$, and hence also in $\mathbb{R}^6\,.$
\end{prop}

\begin{proof}
A value $c\in\mathbb{R}$ is a regular value of $\mathsf{f}$ if, for every $\boldsymbol{\sigma}\in \mathsf{f}^{-1}(c)$, the derivative map
\begin{equation}
	d\mathsf{f}_{\boldsymbol{\sigma}}:T_{\boldsymbol{\sigma}}\mathcal{A}\to 
	T_c\mathbb{R}\simeq \mathbb{R}\,,
\end{equation}
is surjective, which is equivalent to $d\mathsf{f}_{\boldsymbol{\sigma}}\neq \mathbf{0}$. By the preimage theorem \citep{Lee2013}, if $0$ is a regular value of $\mathsf{f}$, then $\mathfrak{S}=\mathsf{f}^{-1}(0)$ is a smooth embedded submanifold of $\mathcal{A}$ of codimension one. Hence, $\mathfrak{S}$ is a smooth embedded hypersurface in $\mathcal{A}$, and since $\mathcal{A}$ is an embedded submanifold of $\mathbb{R}^6$, it follows that $\mathfrak{S}$ is also an embedded hypersurface in $\mathbb{R}^6\,.$

Under either assumption \eqref{Large-Stress-Condition-1} or \eqref{Large-Stress-Condition-2}, there exists $R>0$ such that $\mathsf{f}(\boldsymbol{\sigma})>0$ for all $\boldsymbol{\sigma}\in\mathcal{A}$ with $\|\boldsymbol{\sigma}\|>R$. Therefore,
\begin{equation}
	\mathfrak{S}\subset \overline{B_R(0)}\,,
\end{equation}
and hence $\mathfrak{S}$ is bounded. Since $\mathsf{f}$ is continuous and $\{0\}$ is closed, the set $\mathfrak{S}=\mathsf{f}^{-1}(0)$ is closed in $\mathcal{A}$, and therefore closed in $\mathbb{R}^6$. Thus, $\mathfrak{S}$ is closed and bounded in $\mathbb{R}^6$, and hence compact by the Heine--Borel theorem \citep{Rudin1976}\,.
\end{proof}

In summary, if $\mathsf{f}$ is smooth, if $0$ is a regular value, and if sufficiently large admissible stresses are unsafe, then the strength hypersurface $\mathfrak{S}$ is a smooth compact embedded hypersurface that separates safe and failing stresses.

\begin{remark}[Strength criteria in linear elasticity]
Linearization is performed about an initial motion $\mathring{\varphi}:\mathcal{B}\rightarrow\mathcal{S}$, with initial deformation gradient $\mathring{\mathbf{F}}=T\mathring{\varphi}$. For a one-parameter family of motions $\varphi_\epsilon$ satisfying $\varphi_0=\mathring{\varphi}$, the deformation gradient has the expansion $\mathbf{F}_\epsilon=\mathring{\mathbf{F}}+\epsilon\,\delta\mathbf{F}+o(\epsilon)$. The initial motion need not coincide with the inclusion of the reference configuration, and the initial state may be finitely deformed and prestressed. The resulting linear theory is an incremental theory about the state associated with $\mathring{\varphi}$. The incremental strain and stress measures depend on the initial deformation gradient and, when the initial state is prestressed, on the initial stress. In particular, the different incremental stress measures do not generally coincide. They agree to leading order in the special case of linearization about an undeformed, stress-free configuration with $\mathring{\mathbf{F}}=\mathbf{I}$. Thus, even within a linearized theory, the representation of strength may retain information about the initial deformation and stress. A strength hypersurface specifies the critical states forming the boundary of the safe domain and constitutes material information beyond the incremental elastic moduli. The incremental constitutive equations obtained by linearization characterize the elastic response to infinitesimal perturbations of the initial motion; they do not by themselves determine the strength hypersurface. Classical criteria such as Tresca, von Mises, Mohr--Coulomb, and Drucker--Prager are therefore prescribed directly as strength hypersurfaces using material symmetry, physical considerations, and experimental observations. A finite-elasticity strength function may also be expanded about the initial state associated with $\mathring{\varphi}$. Such an expansion provides a local approximation of the strength function near that state. Its coefficients depend on the initial deformation, the initial stress, and the derivatives of the strength function evaluated at the initial state. This local expansion is distinct from classical strength criteria, which are intended to describe the entire boundary of the safe domain rather than the behavior of the strength function only in a neighborhood of a single state. When the initial motion is the undeformed, stress-free configuration, the classical stress measures agree to first order, and a strength criterion may therefore be expressed unambiguously in terms of the linearized stress tensor. This explains the conventional use of stress-based strength criteria in classical linear elasticity. For a general finitely deformed or prestressed initial state, however, the choice of stress measure and its relation to the initial deformation must remain explicit.
\end{remark}

\subsection{Material strength in the presence of internal constraints}

In the presence of an internal constraint, the stress has a reactive part and a constitutive part. 
As an example, let us consider incompressible solids, for which $J=1$. In this case, the second Piola--Kirchhoff stress admits the following standard additive decomposition\footnote{In terms of the first Piola--Kirchhoff and Cauchy stresses this is rewritten as
\begin{equation}
	\mathbf{P} = -p_0\,\mathbf{g}^{\sharp}\,\mathbf{F}^{-\star} + \bar{\mathbf{P}}\,, \qquad
	\boldsymbol{\sigma} = -p\,\mathbf{g}^{\sharp} + \bar{\boldsymbol{\sigma}}\,,
\end{equation}
where $p=J^{-1} p_0$.}

\begin{equation}
	\mathbf{S}  = -\,p_0\,\mathbf{C}^{-\sharp} + \bar{\mathbf{S}}\,,
\end{equation}
where $p_0$ is an undetermined scalar Lagrange multiplier field and $\bar{\mathbf{S}}$ is the constitutive stress. The scalar field $p_0$ is not determined by the constitutive equations; it is fixed only through equilibrium and boundary conditions together with the incompressibility constraint $\det \mathbf{C}^\flat=\det \mathbf{G}$. Consequently, in any constitutive equation, whether implicit or explicit, it is the constitutive part of the stress $\bar{\mathbf{S}}$ that enters directly, rather than the total stress $\mathbf{S}$.

If the strength function is assumed to depend only on stress, i.e., $\hat{\mathsf{F}}(\mathbf{S},\mathbf{G})=0$, then the internal constraint does not appear explicitly in the equation defining the strength hypersurface. However, the admissible strength states must satisfy
\begin{equation} \label{Strength-System}
\begin{dcases}
	\hat{\mathsf{F}}(\mathbf{S},\mathbf{G})=0\,,\\
	\mathbf{S}	= -\,p_0\,\mathbf{C}^{-\sharp} +	\bar{\mathbf{S}}\,,\\
	\boldsymbol{\mathfrak{f}}(\bar{\mathbf{S}},\mathbf{C}^\flat,\mathbf{G})=\mathbf{0}\,,\\
	\det \mathbf{C}^\flat=\det \mathbf{G}\,,
\end{dcases}
\end{equation}
where \eqref{Strength-System}$_3$ is the implicit constitutive equation.
Thus, the internal constraint restricts the set of admissible stresses without modifying the form of the strength function.
On the other hand, if the strength function depends on both stress and strain, i.e., $\hat{\mathsf{F}}(\mathbf{S},\mathbf{C}^\flat,\mathbf{G})=0$, then the admissible strength states are determined by
\begin{equation}
\begin{dcases}
	\hat{\mathsf{F}}(\mathbf{S},\mathbf{C}^\flat,\mathbf{G})=0\,,\\
	\mathbf{S}= -\,p_0\,\mathbf{C}^{-\sharp} +	\bar{\mathbf{S}}\,,\\
	\boldsymbol{\mathfrak{f}}(\bar{\mathbf{S}},\mathbf{C}^\flat,\mathbf{G})=\mathbf{0}\,,\\
	\det \mathbf{C}^\flat=\det \mathbf{G}\,.
\end{dcases}
\end{equation}
Thus, for constrained materials, a stress-based strength criterion is restricted only indirectly through constitutive admissibility, whereas a strength function depending on both stress and strain is restricted directly through the admissible stress-strain pairs.

\begin{example}[Uniform hydrostatic tensile stress in an incompressible solid]\label{Ex:Hydro-Static-Tension}
Consider a homogeneous incompressible solid for which the total first Piola--Kirchhoff stress is uniform. In particular, there are no body forces, and the loading is applied only through boundary tractions. Let the reference configuration be a ball of radius $R_0$. We use spherical coordinates $(R,\Theta,\Phi)$ and $(r,\theta,\phi)$ in the reference and current configurations, respectively. The material and spatial metrics have the classical representations $\mathbf{G}=\mathrm{diag}(1,\,R^{2},\,R^{2}\sin^{2}\Theta)$ and $\mathbf{g}=\mathrm{diag}(1,\,r^{2},\,r^{2}\sin^{2}\theta)$. We consider radial deformations:
\begin{equation} \label{Deformation-Family4}
	r(R,\Theta,\Phi)=r(R)\,,\qquad \theta(R,\Theta,\Phi)=\Theta\,,\qquad \phi(R,\Theta,\Phi)=\Phi
	\,.
\end{equation}
The deformation gradient has the representation $\mathbf{F}=\mathrm{diag}(r'(R),\,1,\,1)$. Incompressibility implies that
\begin{equation} \label{Jacobian}
	J=\sqrt{\frac{\det\mathbf{g}}{\det\mathbf{G}}}\,\det\mathbf{F}
	=\dfrac{r^2(R)\,r'(R)}{R^2}=1 \,.
\end{equation}
Assuming that $r(0)=0$, one concludes that $r(R)=R$. 
Thus, under radial deformations the spherical ball does not deform, i.e., the deformation is the identity map, and the body behaves effectively as a rigid solid under this loading.
The principal invariants are calculated as
\begin{equation}
	I_1 = \frac{R^6 + 2 r^6(R)}{R^2 r^4(R)}=3\,,\qquad
	I_2 = \frac{2 R^6 + r^6(R)}{R^4 r^2(R)}=3\,.
\end{equation}
The non-zero Cauchy stress components are
\begin{equation}
\begin{aligned}
	\sigma^{rr}(R) &
	= -p + \frac{2 R^{4} W_1}{r^{4}(R)} - \frac{2 W_2\, r^{4}(R)}{R^{4}} = -p + 2 (W_1 - W_2)\,,\\[4pt]
	\sigma^{\theta\theta}(R) &
	= -\frac{p}{r^{2}(R)} + \frac{2 W_1}{R^{2}} - \frac{2 R^{2} W_2}{r^{4}(R)}
	= -\frac{p}{R^2} + \frac{2 (W_1-W_2)}{R^{2}} \,,\\[4pt]
	\sigma^{\phi\phi}(R) &
	= \left[-\frac{p}{r^{2}(R)} + \frac{2 W_1}{R^{2}} - \frac{2 R^{2} W_2}{r^{4}(R)}\right] \csc^{2}\Theta
	=\left[ -\frac{p}{R^2} + \frac{2 (W_1-W_2)}{R^{2}} \right] \csc^{2}\Theta
	\,.
\end{aligned}
\end{equation}
The only non-trivial equilibrium equation is
\begin{equation} \label{Radial-Equilibrium}
	\frac{\partial \sigma^{rr}}{\partial r}
	+\frac{2}{r}\sigma^{rr}-r\,\sigma^{\theta\theta}-(r\sin^2\theta)\,\sigma^{\phi\phi}=0 \,,
\end{equation}
which simplifies to $\frac{\partial \sigma^{rr}(R)}{\partial R}=0$. Therefore, if $\sigma^{rr}(R_0)=\sigma_0$, then $\sigma^{rr}(R)=\sigma_0$ throughout the body. It follows that $p(R)=-\sigma_0+2(W_1-W_2)$, and hence the corresponding physical components of the other two normal stresses are also equal to $\sigma_0$. Thus, the deformation map is the identity and the stress field is a uniform hydrostatic tension, i.e., $\boldsymbol{\sigma}=\sigma_0\,\mathbf{g}^\sharp$.
This example demonstrates that strength functions formulated solely in terms of strain, strain energy, or its distortional component are not, in general, adequate.
It was shown in Proposition~\ref{prop:star-shaped-safe} that the safe domain is star-shaped with respect to the origin. Thus, every ray emanating from the origin intersects the strength hypersurface exactly once. In particular, along the hydrostatic tensile ray $\sigma_1=\sigma_2=\sigma_3=\sigma_0$ with $\sigma_0>0$, there exists a unique intersection with the strength surface.
\end{example}

\section{Material Strength of Isotropic Solids} \label{Material-Strength-Isotropic-Solids}

In this section, we specialize the general framework to isotropic solids and examine several classical strength criteria as examples, with particular attention to the geometry of the associated strength surfaces and safe domains.

For an isotropic material, three principal values $(\sigma_1,\sigma_2,\sigma_3)$ of the Cauchy stress control fracture. In this case the strength hypersurface $\mathsf{f}(\boldsymbol{\sigma},\mathbf{g})=\hat{\mathsf{f}}(\sigma_1,\sigma_2,\sigma_3)=0$ is a surface in a $3$-manifold. Isotropy implies that we must have the following permutation symmetries
\begin{equation}
	\hat{\mathsf{f}}(\sigma_1,\sigma_2,\sigma_3)
	=\hat{\mathsf{f}}(\sigma_1,\sigma_3,\sigma_2)
	=\hat{\mathsf{f}}(\sigma_2,\sigma_1,\sigma_3)
	=\hat{\mathsf{f}}(\sigma_2,\sigma_3,\sigma_1)
	=\hat{\mathsf{f}}(\sigma_3,\sigma_1,\sigma_2)
	=\hat{\mathsf{f}}(\sigma_3,\sigma_2,\sigma_1)
	\,.
\end{equation}
Obviously, $\hat{\mathsf{f}}(0,0,0)<0$.
Each transposition of two principal stresses corresponds to a reflection across one of the planes $\sigma_1=\sigma_2$, $\sigma_1=\sigma_3$, or $\sigma_2=\sigma_3$. Therefore, the strength hypersurface is mirror-symmetric with respect to each of these three planes. In particular, if $(\sigma_1,\sigma_2,\sigma_3)$ lies on the strength hypersurface, then its mirror image with respect to any of these planes also lies on the strength hypersurface. 
Consequently, the surface is completely determined by its restriction to any one of the six regions corresponding to the six possible orderings of the principal stresses, namely
$\sigma_1\ge\sigma_2\ge\sigma_3$, $\sigma_1\ge\sigma_3\ge\sigma_2$, $\sigma_2\ge\sigma_1\ge\sigma_3$, $\sigma_2\ge\sigma_3\ge\sigma_1$, $\sigma_3\ge\sigma_1\ge\sigma_2$, and $\sigma_3\ge\sigma_2\ge\sigma_1$. For example, one may restrict attention to the region $\sigma_1\ge\sigma_2\ge\sigma_3$, and the remaining parts of the surface then follow by symmetry.

In the case of states of plane stress, the strength hypersurface reduces to a curve defined by $\mathsf{f}(\boldsymbol{\sigma},\mathbf{g})=\hat{\mathsf{f}}(\sigma_1,\sigma_2)=0$ with the symmetry $\hat{\mathsf{f}}(\sigma_1,\sigma_2)=\hat{\mathsf{f}}(\sigma_2,\sigma_1)$. Let us define the reflection map $S(\sigma_1,\sigma_2)=(\sigma_2,\sigma_1)$. The symmetry condition can be written as $\mathsf{f}=\mathsf{f}\circ S$. Therefore, if $(\sigma_1,\sigma_2)$ belongs to the strength curve, i.e., if $\hat{\mathsf{f}}(\sigma_1,\sigma_2)=0$, then $\mathsf{f}(S(\sigma_1,\sigma_2))=\hat{\mathsf{f}}(\sigma_2,\sigma_1)=0$ as well. Hence, the set $\{(\sigma_1,\sigma_2):\hat{\mathsf{f}}(\sigma_1,\sigma_2)=0\}$ is invariant under the reflection $(\sigma_1,\sigma_2)\mapsto(\sigma_2,\sigma_1)$, and the strength curve is mirror-symmetric with respect to the line $\sigma_1=\sigma_2$. Consequently, the strength curve is completely determined by its restriction to either of the two regions $\sigma_1\ge\sigma_2$ or $\sigma_2\ge\sigma_1$. For example, one may restrict attention to the region $\sigma_1\ge\sigma_2$, and the remaining part of the curve then follows by symmetry.

\subsection{Examples of Isotropic Strength Surfaces}

In this section, several classical isotropic strength criteria \citep{Chockalingam2026} are written explicitly in terms of the principal Cauchy stresses and analyzed within the framework introduced earlier. For each criterion we first rewrite the strength surface in a form compatible with our stress convention and parameter notation. We then examine whether the corresponding safe domain satisfies the proportional reduction safety property, and hence whether the admissible stress region is star-shaped with respect to the origin.

\subsubsection{Mohr--Coulomb Strength Surface}

The Mohr--Coulomb strength surface depends only on the maximum and minimum principal stresses. Let us define
\begin{equation}
	\sigma_{\max} = \max\{\sigma_1,\sigma_2,\sigma_3\}\,,\qquad
	\sigma_{\min} = \min\{\sigma_1,\sigma_2,\sigma_3\}\,.
\end{equation}
Then the Mohr--Coulomb strength surface can be written as
\begin{equation}
	\mathsf{f}_{\mathrm{MC}}(\sigma_1,\sigma_2,\sigma_3)
	=\beta_1\,\sigma_{\max}+\beta_2\,\sigma_{\min}-1=0\,.
\end{equation}
Here $\beta_1$ and $\beta_2$ are material constants. They are determined by requiring that the strength surface pass through the uniaxial tensile and uniaxial compressive strength states.
Under uniaxial tension with tensile strength $\sigma_{\mathrm{ts}}^{\mathrm{MC}}$, the principal stresses are $(\sigma_{\mathrm{ts}}^{\mathrm{MC}},0,0)$, and hence $\sigma_{\max}=\sigma_{\mathrm{ts}}^{\mathrm{MC}}$ and $\sigma_{\min}=0$. Substituting into the Mohr--Coulomb strength surface gives $\beta_1\,\sigma_{\mathrm{ts}}^{\mathrm{MC}}-1=0$, and therefore
\begin{equation}
	\beta_1=\frac{1}{\sigma_{\mathrm{ts}}^{\mathrm{MC}}}\,.
\end{equation}
Under uniaxial compression with compressive strength $\sigma_{\mathrm{cs}}^{\mathrm{MC}}$, the principal stresses are $(-\sigma_{\mathrm{cs}}^{\mathrm{MC}},0,0)$, and hence $\sigma_{\max}=0$ and $\sigma_{\min}=-\sigma_{\mathrm{cs}}^{\mathrm{MC}}$. Substituting into the Mohr--Coulomb strength surface gives $-\beta_2\,\sigma_{\mathrm{cs}}^{\mathrm{MC}}-1=0$, and therefore
\begin{equation}
	\beta_2=-\frac{1}{\sigma_{\mathrm{cs}}^{\mathrm{MC}}}\,.
\end{equation}
Thus the Mohr--Coulomb strength surface may be written as
\begin{equation}
	\mathsf{f}_{\mathrm{MC}}(\sigma_1,\sigma_2,\sigma_3)
	=\frac{\sigma_{\max}}{\sigma_{\mathrm{ts}}^{\mathrm{MC}}}
	-\frac{\sigma_{\min}}{\sigma_{\mathrm{cs}}^{\mathrm{MC}}}
	-1=0\,.
\end{equation}
The surface is plotted in Fig.~\ref{fig:strength-surfaces}(a). This criterion does not depend on the intermediate principal stress.

We next show that the safe domain of the Mohr--Coulomb criterion is star-shaped with respect to the origin. To this end, let $(\sigma_1,\sigma_2,\sigma_3)$ be a safe stress state, so that $\mathsf{f}_{\mathrm{MC}}(\sigma_1,\sigma_2,\sigma_3)<0$, and let $t\in[0,1)$. Since $t\ge 0$, one has $(t\boldsymbol{\sigma})_{\max}=t\,\sigma_{\max}$ and $(t\boldsymbol{\sigma})_{\min}=t\,\sigma_{\min}$, and therefore
\begin{equation}
	\mathsf{f}_{\mathrm{MC}}(t\sigma_1,t\sigma_2,t\sigma_3)
	=t\left(
	\frac{\sigma_{\max}}{\sigma_{\mathrm{ts}}^{\mathrm{MC}}}
	-\frac{\sigma_{\min}}{\sigma_{\mathrm{cs}}^{\mathrm{MC}}}
	\right)-1\,.
\end{equation}
Since $\mathsf{f}_{\mathrm{MC}}(\sigma_1,\sigma_2,\sigma_3)<0$, one has $\frac{\sigma_{\max}}{\sigma_{\mathrm{ts}}^{\mathrm{MC}}}-\frac{\sigma_{\min}}{\sigma_{\mathrm{cs}}^{\mathrm{MC}}}<1$. Multiplying by $t\in[0,1)$ gives $t\left(\frac{\sigma_{\max}}{\sigma_{\mathrm{ts}}^{\mathrm{MC}}}-\frac{\sigma_{\min}}{\sigma_{\mathrm{cs}}^{\mathrm{MC}}}\right)<t<1$, and hence $\mathsf{f}_{\mathrm{MC}}(t\sigma_1,t\sigma_2,t\sigma_3)<0$. Thus the Mohr--Coulomb criterion satisfies the proportional reduction safety property, and consequently its safe domain is star-shaped with respect to the origin, see Fig.~\ref{fig:strength-surfaces}(a).

\subsubsection{Hoek--Brown Strength Surface}

The Hoek--Brown strength surface is usually written in the rock-mechanics convention in which compressive stress is taken to be positive. Thus, if $(\sigma_1,\sigma_2,\sigma_3)$ denote the principal Cauchy stresses in our convention, then one sets $s_i=-\sigma_i$. It follows that $s_{\max}=-\sigma_{\min}$ and $s_{\min}=-\sigma_{\max}$. In the notation of the paper, the generalized Hoek--Brown criterion is written as
\begin{equation}
	s_{\max}=s_{\min}+\sigma_c\left(m_b\frac{s_{\min}}{\sigma_c}+s_{\mathrm{HB}}\right)^a\,.
\end{equation}
Here $\sigma_c$, $m_b$, $s_{\mathrm{HB}}$, and $a$ are material parameters. In the usual Hoek--Brown model one has $0<a\le 1$, and therefore the expression is well defined only when $m_b\frac{s_{\min}}{\sigma_c}+s_{\mathrm{HB}}\ge 0$.
For intact rock one has $m_b=m_i$, $s_{\mathrm{HB}}=1$, and $a=0.5$. Rewriting the criterion in terms of $(\sigma_1,\sigma_2,\sigma_3)$ gives
\begin{equation}
	-\sigma_{\min}
	=-\sigma_{\max}
	+\sigma_c\left(-m_b\frac{\sigma_{\max}}{\sigma_c}+s_{\mathrm{HB}}\right)^a\,.
\end{equation}
Hence
\begin{equation}
	\sigma_{\max}-\sigma_{\min}
	=\sigma_c\left(s_{\mathrm{HB}}-m_b\frac{\sigma_{\max}}{\sigma_c}\right)^a\,.
\end{equation}
Therefore, in our notation, the Hoek--Brown strength surface may be written as
\begin{equation}
	\mathsf{f}_{\mathrm{HB}}(\sigma_1,\sigma_2,\sigma_3)
	=\sigma_{\max}-\sigma_{\min}
	-\sigma_c\left(s_{\mathrm{HB}}-m_b\frac{\sigma_{\max}}{\sigma_c}\right)^a=0\,.
\end{equation}
The surface is plotted in Fig.~\ref{fig:strength-surfaces}(b). Unlike the Mohr--Coulomb criterion, this relation is nonlinear in the principal stresses. It is calibrated through the parameters $\sigma_c$, $m_b$, $s_{\mathrm{HB}}$, and $a$, and the corresponding uniaxial tensile and compressive strengths are obtained by substituting the uniaxial states into the criterion, which yields nonlinear relations among these parameters. Nevertheless, like the Mohr--Coulomb criterion, the Hoek--Brown surface depends only on $\sigma_{\max}$ and $\sigma_{\min}$ and therefore ignores the intermediate principal stress.

The safe domain associated with the Hoek--Brown strength surface is also star-shaped with respect to the origin, under the usual assumption $0<a\le 1$ and provided the fractional power remains well-defined along the proportional reduction path. To see this, let $(\sigma_1,\sigma_2,\sigma_3)$ be a safe stress state, so that $\mathsf{f}_{\mathrm{HB}}(\sigma_1,\sigma_2,\sigma_3)<0$, and let $t\in[0,1)$. Since $t\ge 0$, one has $(t\boldsymbol{\sigma})_{\max}=t\,\sigma_{\max}$ and $(t\boldsymbol{\sigma})_{\min}=t\,\sigma_{\min}$. Therefore
\begin{equation}
	\mathsf{f}_{\mathrm{HB}}(t\sigma_1,t\sigma_2,t\sigma_3)
	=t(\sigma_{\max}-\sigma_{\min})
	-\sigma_c\left(s_{\mathrm{HB}}-m_b\frac{t\,\sigma_{\max}}{\sigma_c}\right)^a\,.
\end{equation}
Next observe that
$s_{\mathrm{HB}}-m_b\frac{t\,\sigma_{\max}}{\sigma_c}=(1-t)s_{\mathrm{HB}}+t\left(s_{\mathrm{HB}}-m_b\frac{\sigma_{\max}}{\sigma_c}\right)$.
Since $0<a\le 1$, the function $f(x)=x^a$ is concave on $[0,\infty)$, and therefore for any $x_0,x_1\ge 0$ and $t\in[0,1]$ one has
\begin{equation}
	f((1-t)x_0+t x_1)\ge (1-t)f(x_0)+t f(x_1).
\end{equation}
Applying this with $x_0=s_{\mathrm{HB}}$ and $x_1=s_{\mathrm{HB}}-m_b\frac{\sigma_{\max}}{\sigma_c}$ yields
\begin{equation}
\begin{aligned}
	\left(s_{\mathrm{HB}}-m_b\frac{t\,\sigma_{\max}}{\sigma_c}\right)^a
	&\ge
	(1-t)s_{\mathrm{HB}}^a
	+
	t\left(s_{\mathrm{HB}}-m_b\frac{\sigma_{\max}}{\sigma_c}\right)^a.
\end{aligned}
\end{equation}
Since $(1-t)s_{\mathrm{HB}}^a\ge 0$, it follows that
\begin{equation}
	\left(s_{\mathrm{HB}}-m_b\frac{t\,\sigma_{\max}}{\sigma_c}\right)^a
	\ge
	t\left(s_{\mathrm{HB}}-m_b\frac{\sigma_{\max}}{\sigma_c}\right)^a.
\end{equation}
Substituting this estimate into $\mathsf{f}_{\mathrm{HB}}(t\sigma_1,t\sigma_2,t\sigma_3)$ gives
\begin{equation}
\begin{aligned}
	\mathsf{f}_{\mathrm{HB}}(t\sigma_1,t\sigma_2,t\sigma_3)
	&=t(\sigma_{\max}-\sigma_{\min})
	-\sigma_c\left(s_{\mathrm{HB}}-m_b\frac{t\,\sigma_{\max}}{\sigma_c}\right)^a\\
	&\le
	t(\sigma_{\max}-\sigma_{\min})
	-t\,\sigma_c\left(s_{\mathrm{HB}}-m_b\frac{\sigma_{\max}}{\sigma_c}\right)^a\\
	&=
	t\left[\sigma_{\max}-\sigma_{\min}
	-\sigma_c\left(s_{\mathrm{HB}}-m_b\frac{\sigma_{\max}}{\sigma_c}\right)^a\right]
	=
	t\,\mathsf{f}_{\mathrm{HB}}(\sigma_1,\sigma_2,\sigma_3)<0\,.
\end{aligned}
\end{equation}
Thus, the Hoek--Brown criterion satisfies the proportional reduction safety property, and consequently its safe domain is star-shaped with respect to the origin.

\subsubsection{Generalized Three-Dimensional Hoek--Brown Strength Surface}

The generalized three-dimensional Hoek--Brown surface is written in terms of $\tau_{\mathrm{oct}}$ and $s_{m,2}$, where
\begin{equation}
	\tau_{\mathrm{oct}}
	=\sqrt{\frac{2}{3}\,J_2}\,,\qquad
	J_2
	=\frac{1}{6}\Big[(\sigma_1-\sigma_2)^2+(\sigma_2-\sigma_3)^2+(\sigma_3-\sigma_1)^2\Big]\,,
\end{equation}
and
\begin{equation}
	s_{m,2}=\frac{1}{2}(s_{\max}+s_{\min})
	=-\frac{1}{2}(\sigma_{\max}+\sigma_{\min})\,.
\end{equation}
Thus
\begin{equation}
	\tau_{\mathrm{oct}}
	=\frac{1}{3}\sqrt{(\sigma_1-\sigma_2)^2+(\sigma_2-\sigma_3)^2+(\sigma_3-\sigma_1)^2}\,.
\end{equation}
The generalized three-dimensional Hoek--Brown surface is described as 
\begin{equation}
	\mathsf{f}_{\mathrm{GHB}}(\sigma_1,\sigma_2,\sigma_3)
	=\beta_1\big(3\sqrt{2}\,\tau_{\mathrm{oct}}\big)^{\frac{1}{a}}
	+\beta_2\left(\frac{3}{2}\sqrt{2}\,\tau_{\mathrm{oct}}-s_{m,2}\right)-1=0\,.
\end{equation}
The surface is plotted in Fig.~\ref{fig:strength-surfaces}(c).
Using $s_{m,2}=-(\sigma_{\max}+\sigma_{\min})/2$, this becomes
\begin{equation}
	\mathsf{f}_{\mathrm{GHB}}(\sigma_1,\sigma_2,\sigma_3)
	=\beta_1\big(3\sqrt{2}\,\tau_{\mathrm{oct}}\big)^{\frac{1}{a}}
	+\beta_2\left(\frac{3}{2}\sqrt{2}\,\tau_{\mathrm{oct}}
	+\frac{1}{2}(\sigma_{\max}+\sigma_{\min})\right)-1=0\,.
\end{equation}
The material constants are written as
\begin{equation}
	\beta_1=\frac{1}{\big(\sigma_{\mathrm{cs}}^{\mathrm{HB}}\big)^{\frac{1}{a}}}\,,\qquad
	\beta_2
	=\frac{1}{\sigma_{\mathrm{ts}}^{\mathrm{HB}}}
	\left[
	1-\left(\frac{\sigma_{\mathrm{ts}}^{\mathrm{HB}}}{\sigma_{\mathrm{cs}}^{\mathrm{HB}}}\right)^{\frac{1}{a}}
	\right]\,.
\end{equation}
Unlike the preceding two criteria, this one depends on all three principal stresses through $J_2$, and therefore includes the effect of the intermediate principal stress.

The safe domain associated with the generalized three-dimensional Hoek--Brown surface is also star-shaped with respect to the origin, provided $\beta_1\ge 0$, $\beta_2\ge 0$, and $0<a\le 1$. To see this, let $(\sigma_1,\sigma_2,\sigma_3)$ be a safe stress state, so that $\mathsf{f}_{\mathrm{GHB}}(\sigma_1,\sigma_2,\sigma_3)<0$, and let $t\in[0,1)$. Since $\tau_{\mathrm{oct}}$ is homogeneous of degree one in the principal stresses, one has $\tau_{\mathrm{oct}}(t\sigma_1,t\sigma_2,t\sigma_3)=t\,\tau_{\mathrm{oct}}$, while $(t\boldsymbol{\sigma})_{\max}=t\,\sigma_{\max}$ and $(t\boldsymbol{\sigma})_{\min}=t\,\sigma_{\min}$. Therefore
\begin{equation}
	\mathsf{f}_{\mathrm{GHB}}(t\sigma_1,t\sigma_2,t\sigma_3)
	=t^{\frac{1}{a}}\beta_1\left(3\sqrt{2}\,\tau_{\mathrm{oct}}\right)^{\frac{1}{a}}
	+t\beta_2\left[\frac{3}{2}\sqrt{2}\,\tau_{\mathrm{oct}}
	+\frac{1}{2}(\sigma_{\max}+\sigma_{\min})\right]-1\,.
\end{equation}
Since $0<a\le 1$, one has $\frac{1}{a}\ge 1$, and hence $t^{\frac{1}{a}}\le t\le 1$ for every $t\in[0,1)$. Using also $\beta_1\ge 0$ and $\beta_2\ge 0$, it follows that
\begin{equation}
\begin{aligned}
	\mathsf{f}_{\mathrm{GHB}}(t\sigma_1,t\sigma_2,t\sigma_3)
	&\le
	t\beta_1\left(3\sqrt{2}\,\tau_{\mathrm{oct}}\right)^{\frac{1}{a}}
	+t\beta_2\left[\frac{3}{2}\sqrt{2}\,\tau_{\mathrm{oct}}
	+\frac{1}{2}(\sigma_{\max}+\sigma_{\min})\right]-1\\
	&=
	t\left\{
	\beta_1\left(3\sqrt{2}\,\tau_{\mathrm{oct}}\right)^{\frac{1}{a}}
	+\beta_2\left[\frac{3}{2}\sqrt{2}\,\tau_{\mathrm{oct}}
	+\frac{1}{2}(\sigma_{\max}+\sigma_{\min})\right]
	\right\}-1\\
	&<
	\beta_1\left(3\sqrt{2}\,\tau_{\mathrm{oct}}\right)^{\frac{1}{a}}
	+\beta_2\left[\frac{3}{2}\sqrt{2}\,\tau_{\mathrm{oct}}
	+\frac{1}{2}(\sigma_{\max}+\sigma_{\min})\right]-1
	=\mathsf{f}_{\mathrm{GHB}}(\sigma_1,\sigma_2,\sigma_3)<0\,.
\end{aligned}
\end{equation}
Thus, the generalized three-dimensional Hoek--Brown criterion satisfies the proportional reduction safety property, and consequently its safe domain is star-shaped with respect to the origin.

\subsubsection{Mogi--Coulomb Strength Surface}

The Mogi--Coulomb surface is written in the paper as
\begin{equation}
	\mathsf{f}_{\mathrm{MgC}}(\sigma_1,\sigma_2,\sigma_3)
	=\beta_1\,s_{m,2}+\beta_2\,\tau_{\mathrm{oct}}-1=0\,.
\end{equation}
Using $s_{m,2}=-\frac{1}{2}(\sigma_{\max}+\sigma_{\min})$, this becomes
\begin{equation}
	\mathsf{f}_{\mathrm{MgC}}(\sigma_1,\sigma_2,\sigma_3)
	=-\frac{\beta_1}{2}(\sigma_{\max}+\sigma_{\min})
	+\beta_2\,\tau_{\mathrm{oct}}-1=0\,.
\end{equation}
The parameters are given by
\begin{equation}
	\beta_1
	=\frac{1}{\sigma_{\mathrm{cs}}^{\mathrm{MC}}}-\frac{1}{\sigma_{\mathrm{ts}}^{\mathrm{MC}}}\,,\qquad
	\beta_2
	=\frac{3\sqrt{2}}{4}
	\left(
	\frac{1}{\sigma_{\mathrm{ts}}^{\mathrm{MC}}}+\frac{1}{\sigma_{\mathrm{cs}}^{\mathrm{MC}}}
	\right)\,.
\end{equation}
Therefore
\begin{equation}
	\mathsf{f}_{\mathrm{MgC}}(\sigma_1,\sigma_2,\sigma_3)
	=-\frac{1}{2}
	\left(
	\frac{1}{\sigma_{\mathrm{cs}}^{\mathrm{MC}}}-\frac{1}{\sigma_{\mathrm{ts}}^{\mathrm{MC}}}
	\right)
	(\sigma_{\max}+\sigma_{\min})
	+\frac{3\sqrt{2}}{4}
	\left(
	\frac{1}{\sigma_{\mathrm{ts}}^{\mathrm{MC}}}+\frac{1}{\sigma_{\mathrm{cs}}^{\mathrm{MC}}}
	\right)\tau_{\mathrm{oct}}
	-1=0\,.
\end{equation}
The surface is plotted in Fig.~\ref{fig:strength-surfaces}(d). This criterion depends on all three principal stresses through $\tau_{\mathrm{oct}}$ and hence accounts for the intermediate principal stress.

The safe domain associated with the Mogi--Coulomb strength surface is star-shaped with respect to the origin. To see this, let $(\sigma_1,\sigma_2,\sigma_3)$ be a safe stress state, so that $\mathsf{f}_{\mathrm{MgC}}(\sigma_1,\sigma_2,\sigma_3)<0$, and let $t\in[0,1)$. Since $\tau_{\mathrm{oct}}$ is homogeneous of degree one in the principal stresses, one has $\tau_{\mathrm{oct}}(t\sigma_1,t\sigma_2,t\sigma_3)=t\,\tau_{\mathrm{oct}}$, while $(t\boldsymbol{\sigma})_{\max}=t\,\sigma_{\max}$ and $(t\boldsymbol{\sigma})_{\min}=t\,\sigma_{\min}$. Therefore
\begin{equation}
	\mathsf{f}_{\mathrm{MgC}}(t\sigma_1,t\sigma_2,t\sigma_3)
	=t\left[
	-\frac{\beta_1}{2}(\sigma_{\max}+\sigma_{\min})
	+\beta_2\,\tau_{\mathrm{oct}}
	\right]-1\,.
\end{equation}
Since $\mathsf{f}_{\mathrm{MgC}}(\sigma_1,\sigma_2,\sigma_3)<0$, one has $-\frac{\beta_1}{2}(\sigma_{\max}+\sigma_{\min})+\beta_2\,\tau_{\mathrm{oct}}<1$. Multiplying by $t\in[0,1)$ gives $t\big[-\frac{\beta_1}{2}(\sigma_{\max}+\sigma_{\min})+\beta_2\,\tau_{\mathrm{oct}}\big]<t<1$, and hence, $\mathsf{f}_{\mathrm{MgC}}(t\sigma_1,t\sigma_2,t\sigma_3)<0$. Thus, the Mogi--Coulomb criterion satisfies the proportional reduction safety property, and consequently its safe domain is star-shaped with respect to the origin.

Fig.~\ref{fig:strength-surfaces} compares the four isotropic strength surfaces discussed in this section in principal stress space under plane stress conditions. Although all four criteria are isotropic, their geometric structures are markedly different, reflecting different assumptions about the dependence of strength on the principal stresses. All four surfaces are star-shaped with respect to the origin. In addition, the surfaces in panels (a) and (b) are convex, whereas those in panels (c) and (d) are not.
\begin{figure}[t!]
\centering
\includegraphics[width=0.75\textwidth]{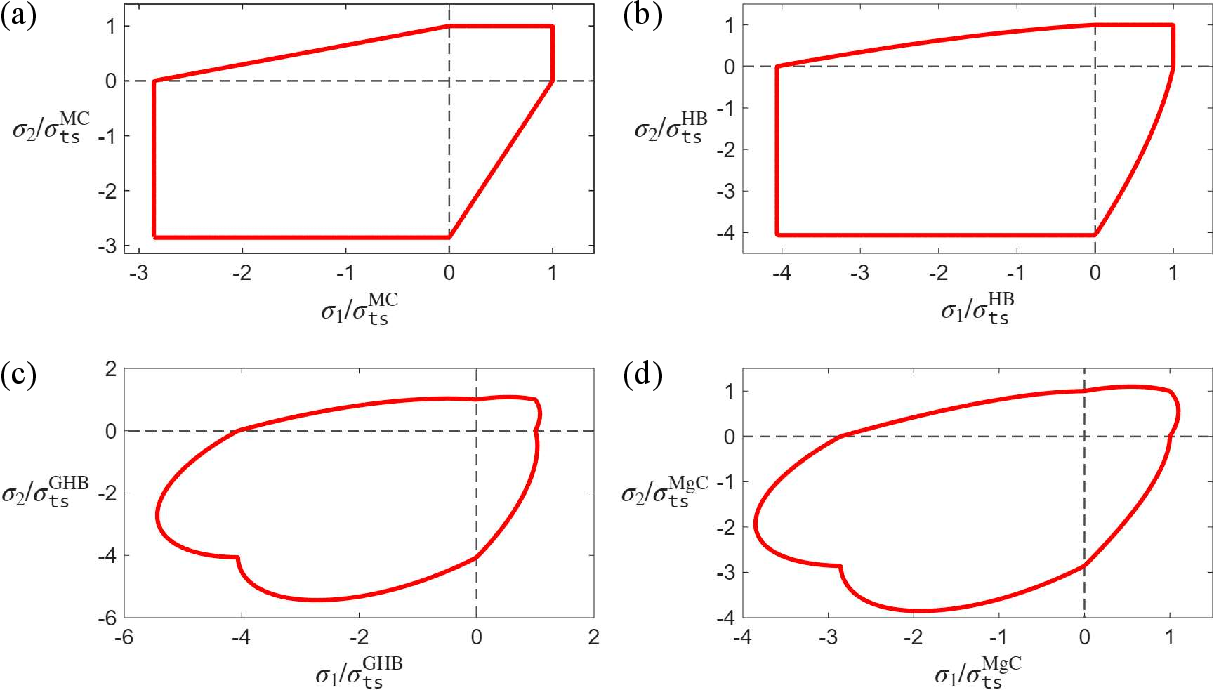}
\vspace*{0.10in}
\caption{Plot of various isotropic strength surfaces in principal stress space under plane stress conditions. (a) Mohr-Coulomb surface, (b) Hoek-Brown surface, (c) 3D Hoek-Brown or GZZ surface, and (d) Mogi-Coulomb surface. }
\label{fig:strength-surfaces}
\end{figure}

\section{Material Strength in the Presence of Residual Stresses} \label{Material-Strength-Residual-Stresses}

We next consider solids with residual stresses and anelastic distortions, and examine how the definition and representation of material strength are modified in anelasticity. 
Anelasticity, in the sense of \citet{Eckart1948}, is characterized by the presence of internal distortions that define a local relaxed configuration, so that the elastic response is measured relative to this evolving stress-free state rather than a global Euclidean reference configuration.
One may ask whether material strength can be affected by residual stresses. It should be emphasized that, in the present framework, material strength is a local material property associated with a homogeneous body under boundary tractions and in a spatially uniform state of stress. In a traction-free body with no body forces, a spatially uniform residual stress must be trivial, i.e., it must vanish \citep{Hoger1985}. In this sense, residual stress does not enter the definition of material strength at a given point. Rather, material strength is defined with respect to a stress-free state, namely for a homogeneous body in its relaxed state before a homogeneous state of stress is imposed that induces fracture.

\subsection{Material metric}

Suppose a body $\mathcal{B}$ has a distribution of eigenstrains.\footnote{\citet{Reissner1931} formulated a theory of \emph{Eigenspannungen} (eigenstresses) generated by incompatible initial distortions acting as internal stress sources. In modern terminology, these initial distortions may be interpreted as eigenstrains. The term \emph{eigenstrain} was later popularized by Mura \citep{Kinoshita1971,Mura1982}. Closely related notions appear in the literature under a variety of names, including \emph{initial strain} \citep{Kondo1949}, \emph{nuclei of strain} \citep{Mindlin1950}, \emph{transformation strain} \citep{Eshelby1957}, \emph{inherent strain} \citep{Ueda1975}, and \emph{residual strain} \citep{ambrosi2019growth}; see also \citep{Jun2010,Zhou2013}.}
Let us consider a material point $X\in\mathcal{B}$ that is mapped to $x=\varphi(X)$ in the current configuration $\mathcal{C}$. An infinitesimal line element $\mathbf{U}\in T_X\mathcal{B}$ in the reference configuration is mapped to $\mathbf{F}\mathbf{U}$ after deformation. Now imagine that independently of the rest of the body the deformed line element is allowed to elastically relax and denote it by $\Fe^{-1}\mathbf{F}\mathbf{U}$, where $\Fe^{-1}:T_x\mathcal{C}\to T_X\mathcal{B}$ is a local elastic unloading map. The resulting line element is different from that in the reference configuration. The two are related as $\Fe^{-1}\mathbf{F}\mathbf{U}=\Fa\mathbf{U}$, where $\Fa(X):T_X\mathcal{B}\to T_X \mathcal{B}$ is the anelastic distortion. As $\mathbf{U}$ is an arbitrary vector, one obtains $\mathbf{F}=\Fe\Fa$---the Bilby-Kr\"oner-Lee multiplicative decomposition of the deformation gradient \citep{bilby1957continuous,kroner1959allgemeine,lee1967finite,lee1969elastic}, see also \citep{Sadik2017,YavariSozio2023}.
Let us denote the induced flat metric of the reference configuration in the absence of eigenstrains by $\mathring{\mathbf{G}}$. The square relaxed line element length is calculated as follows
\begin{equation}
	\llangle \Fa\mathbf{U},\Fa\mathbf{U} \rrangle_{\mathring{\mathbf{G}}}
	=\llangle \mathbf{U},\mathbf{U} \rrangle_{\Fa^*\mathring{\mathbf{G}}}
	=\llangle \mathbf{U},\mathbf{U} \rrangle_{\mathbf{G}}
	\,,
\end{equation}
where $\mathbf{G}=\Fa^*\mathring{\mathbf{G}}=\Fa^\star\mathring{\mathbf{G}}\Fa$ is the material metric \citep{YavariGoriely2013,Yavari2021Eshelby}.

\subsection{Elastic energy function in anelasticity}

Let us consider an isotropic hyperelastic solid with energy function $W$, which explicitly depends on the elastic distortion: $W=W(X,\Fe,\mathring{\mathbf{G}},\mathbf{g})$. For an isotropic solid the energy function is materially covariant \citep{MarsdenHughes1983}, and hence,
\begin{equation} \label{Elastic-Energy-Anelasticity}
	W(X,\Fe,\mathring{\mathbf{G}},\mathbf{g})
	=W(X,\Fa^*\Fe,\Fa^*\mathring{\mathbf{G}},\mathbf{g})
	=W(X,\Fe\Fa,\mathbf{G},\mathbf{g})
	=W(X,\mathbf{F},\mathbf{G},\mathbf{g})
	\,.
\end{equation}
Therefore, in the presence of eigenstrains, one can work with the total strain as long as the material metric is used instead of the induced flat Euclidean metric. This can be generalized to anisotropic solids as well \citep{YavariSozio2023}.

\subsection{Material strength function of an anelastic brittle solid}

In defining material strength, one considers a thought experiment in which a small material neighborhood is isolated and allowed to relax to a stress-free configuration. The specimen is then loaded starting from this relaxed state until fracture occurs. Thus, the notion of strength is tied to the relaxed (or intermediate) configuration rather than to the current configuration viewed through the total deformation. In particular, the relevant kinematic quantity is the elastic distortion $\Fe$, and the corresponding energetically conjugate stress is
\begin{equation}
	\mathring{\mathbf{P}}= \mathbf{g}^\sharp\,\frac{\partial W}{\partial \Fe}\,.
\end{equation}
Consequently, the stress entering the strength function is not the first Piola--Kirchhoff stress $\mathbf{P}$ associated with the total deformation $\mathbf{F}$, but rather $\mathring{\mathbf{P}}$, which is defined with respect to the intermediate configuration. Therefore, the strength function is written in the following form 
\begin{equation}
	\mathsf{F}(X,\mathring{\mathbf{P}},\Fe,\mathbf{g},\mathring{\mathbf{G}})=0	\,.
\end{equation}
The explicit dependence on $X$ emphasizes that, in general, the material may be inhomogeneous.\footnote{The safe domain associated with material strength is conceptually analogous to the elastic domain of elastoplasticity. In both cases one identifies a domain in stress space together with a boundary separating admissible states from critical states. The two notions are fundamentally different, however: the boundary of the safe domain is a fracture criterion, whereas the boundary of the elastic domain is a yield criterion.}

For an isotropic material the strength function is materially covariant, i.e.,
\begin{equation} 
	\mathsf{F}(X,\mathring{\mathbf{P}},\Fe,\mathbf{g},\mathring{\mathbf{G}})
	=\mathsf{F}(X,\Fa^*\mathring{\mathbf{P}},\Fa^*\Fe,\mathbf{g},\Fa^*\mathring{\mathbf{G}})
	=\mathsf{F}(X,\Fa^*\mathring{\mathbf{P}},\mathbf{F},\mathbf{g},\mathbf{G})\,.
\end{equation}
From \eqref{Elastic-Energy-Anelasticity}, one can write
\begin{equation}
	\frac{\partial W}{\partial \Fe}
	=\frac{\partial W}{\partial \mathbf{F}}\frac{\partial \mathbf{F}}{\partial \Fe} 
	=\frac{\partial W}{\partial \mathbf{F}}\,\Fa^\star\,.
\end{equation}
We know that
\begin{equation}
	\mathring{\mathbf{P}}=\mathbf{g}^\sharp \,\frac{\partial W}{\partial \Fe}\,,\qquad
	\mathbf{P}=\mathbf{g}^\sharp \,\frac{\partial W}{\partial \mathbf{F}}\,.
\end{equation}
Hence, $\mathring{\mathbf{P}}=\mathbf{P}\,\Fa^\star$. Note that $\mathring{\mathbf{P}}(X)\in T_{\varphi(X)}\mathcal{S}\otimes T_X\mathcal{B}$, and therefore, under the pull-back by $\Fa:T_X\mathcal{B}\to T_X\mathcal{B}$, the material slot transforms with the inverse map. In coordinate-independent form, one writes $\Fa^*\mathring{\mathbf{P}}	=(\mathrm{id}_{T_{\varphi(X)}\mathcal{S}}\otimes \Fa^{-1})\,\mathring{\mathbf{P}} =\mathring{\mathbf{P}}\Fa^{-\star}$.
In components, this reads $(\Fa^*\mathring{\mathbf{P}})^{aA}=\mathring{P}^{aB}\,\cFa^{-A}{}_B$.
Thus, we have
\begin{equation}
	\Fa^*\mathring{\mathbf{P}}= \mathbf{P}\,.
\end{equation}
Therefore,
\begin{equation} 
	\mathsf{F}(X,\mathring{\mathbf{P}},\Fe,\mathbf{g},\mathring{\mathbf{G}})
	=\mathsf{F}(X,\mathbf{P},\mathbf{F},\mathbf{g},\mathbf{G})\,.
\end{equation}
Thus, similar to the elastic energy function, in the presence of eigenstrains the strength function may be expressed in terms of the total deformation and the corresponding first Piola--Kirchhoff stress, provided that the material metric $\mathbf{G}$ is used in place of the Euclidean reference metric $\mathring{\mathbf{G}}$.

\begin{remark}[Intrinsic strength, residual stress, and evolution of the strength hypersurface]
For a brittle solid with distributed eigenstrains, the strength hypersurface is represented in terms of the total deformation and the material metric as $\mathsf{F}(X,\mathbf{P},\mathbf{F},\mathbf{g},\mathbf{G})=0$.
This representation should be distinguished from the intrinsic definition of material strength. The intrinsic strength criterion is determined by isolating a material neighborhood, allowing it to relax elastically, and loading it from its local stress-free state until fracture occurs; it is therefore defined in the relaxed configuration by $\mathsf{F}(X,\mathring{\mathbf{P}},\Fe,\mathbf{g},\mathring{\mathbf{G}})=0$. The equation written in terms of $(\mathbf{P},\mathbf{F},\mathbf{G})$ is the materially covariant representation of this same criterion in the reference configuration. In the absence of eigenstrains, $\Fa=\mathrm{id}_{T\mathcal{B}}$ and $\mathbf{G}=\mathring{\mathbf{G}}$, and the strength hypersurface reduces to $\mathsf{F}(X,\mathbf{P},\mathbf{F},\mathbf{g},\mathring{\mathbf{G}})=0$. In a boundary-value problem with a prescribed distribution of eigenstrains, the stress $\mathbf{P}$ is not independent of the anelastic distortion. It is determined by the deformation, the constitutive response, the eigenstrain distribution, the equilibrium equations, and the boundary conditions. Consequently, residual stresses affect the current value of the strength function and may bring a material point closer to or farther from its strength hypersurface without altering the intrinsic criterion defined in the relaxed configuration. If the anelastic distortion is time-dependent, $\Fa=\Fa(X,t)$, and evolves according to a kinetic equation, then the material metric $\mathbf{G}(X,t)=\Fa^*(X,t)\,\mathring{\mathbf{G}}(X)=\Fa^\star(X,t)\,\mathring{\mathbf{G}}(X)\,\Fa(X,t)$ evolves as well. Hence, the representation $\mathsf{F}(X,\mathbf{P},\mathbf{F},\mathbf{g},\mathbf{G}(t))=0$ defines a time-dependent strength hypersurface in the constitutively admissible stress space. Thus, evolution of the eigenstrain field has two distinct effects: it changes the actual stress and deformation fields obtained from the boundary-value problem, and it changes the strength hypersurface through the evolving material metric. Both effects must be distinguished from an independent evolution of the intrinsic material strength caused by changes in material constitution, such as degradation, aging, or other microstructural alterations. Such intrinsic strengthening or weakening would require evolution of the material parameters entering $\mathsf{F}$, or an explicit dependence of $\mathsf{F}$ on additional internal variables.
\end{remark}

\begin{example}[Drucker--Prager strength function for a residually stressed solid]
The Drucker--Prager strength function has the form given in \eqref{Drucker-Prager-Strength-Function}. More specifically, we have $\mathcal{F}(X,\SOne,\mathbf{G})=0$, where $\mathbf{G}$ is the material metric.
This strength function is written in terms of the first invariant of the Biot stress and the second invariant of its deviatoric part. More precisely,
$\mathcal{I}_1=\operatorname{tr}_{\mathbf{G}}\SOne$ is the trace of the Biot stress, while $\mathcal{J}_2$ is the second invariant of the deviatoric Biot stress. Thus,
\begin{equation}
	\mathcal{I}_1=\operatorname{tr}_{\mathbf{G}}\SOne	=\SOneC^{AB}G_{AB}\,.
\end{equation}
The deviatoric Biot stress is defined by
\begin{equation}
	\operatorname{dev}_{\mathbf{G}}\SOne = \SOne-\frac{1}{3}\mathcal{I}_1\,\mathbf{G}^{\sharp}\,,
	\qquad
	(\operatorname{dev}_{\mathbf{G}}\SOne)^{AB}
	=\SOneC^{AB}-\frac{1}{3} \mathcal{I}_1 G^{AB}\,.
\end{equation}
Therefore,
\begin{equation}
	\mathcal{J}_2
	=\frac{1}{2}\,\operatorname{dev}_{\mathbf{G}}\SOne\!:\!\operatorname{dev}_{\mathbf{G}}\SOne
	=\frac{1}{2}\,(\operatorname{dev}_{\mathbf{G}}\SOne)^{AB}
	(\operatorname{dev}_{\mathbf{G}}\SOne)^{MN}G_{AM}G_{BN}\,.
\end{equation}
It is seen that the two invariants explicitly depend on the material metric, and hence, eigenstrains. As an explicit example of a residually stressed solid, consider a spherically symmetric body with a radially symmetric distribution of pure dilatational eigenstrain for which in spherical coordinates the material metric has the following representation \citep{Yavari2021Eshelby}
\begin{equation}\label{Material-Metric-Radial}
	\mathbf{G}=\mathbf{G}(R)=e^{2\omega(R)}
\begin{bmatrix}
  1 & 0  & 0  \\
  0 & R^2  & 0  \\
  0 & 0  &  R^2\sin^2\Theta 
\end{bmatrix}\,,
\end{equation}
where $\omega(R)$ is some function that quantifies the eigenstrain. 
The Ricci curvature of the material metric has the matrix representation\footnote{With respect to a coordinate chart $\{X^A\}$, the Levi--Civita connection coefficients are given as $\Gamma^{A}{}_{BC}=\frac{1}{2}G^{AD}(G_{CD,B}+G_{BD,C}-G_{BC,D})$. The components of the Riemann curvature tensor read $\mathcal{R}^{A}{}_{BCD}=\Gamma^{A}{}_{BD,C}-\Gamma^{A}{}_{BC,D}+\Gamma^{A}{}_{CE}\Gamma^{E}{}_{BD}-\Gamma^{A}{}_{DE}\Gamma^{E}{}_{BC}$. Lowering the first index gives $\mathcal{R}_{ABCD}=G_{AE}\mathcal{R}^{E}{}_{BCD}$, and the Ricci tensor is defined as $R_{AB}=\mathcal{R}^{C}{}_{ACB}=G^{CD}\mathcal{R}_{CADB}$. Equivalently, $R_{AB}=\Gamma^{C}{}_{AB,C}-\Gamma^{C}{}_{AC,B}+\Gamma^{C}{}_{CD}\Gamma^{D}{}_{AB}-\Gamma^{C}{}_{BD}\Gamma^{D}{}_{AC}$.}
\begin{equation}
	[\operatorname{Ric}_{\mathbf{G}}]
	=
	\begin{bmatrix}
	-2\left(\omega''+\dfrac{\omega'}{R}\right) & 0 & 0 \\[2mm]
	0 & -R^2\omega''-3R\omega'-R^2(\omega')^2 & 0 \\[2mm]
	0 & 0 & -\sin^2\Theta\left[R^2\omega''+3R\omega'+R^2(\omega')^2\right]
	\end{bmatrix}.
\end{equation}
Since a three-dimensional Riemannian manifold is locally flat if and only if its Ricci curvature vanishes, flatness of the material metric requires that
\begin{equation}
	\omega''+\frac{\omega'}{R}=0\,,\qquad
	\omega''+\frac{3\omega'}{R}+(\omega')^2=0\,.
\end{equation}
The first equation gives $\omega(R)=a\ln R+b$, and substitution into the second yields $a(a+2)=0$. Therefore,
\begin{equation}
	\omega(R)=b\,,\qquad\text{or}\qquad
	\omega(R)=b-2\ln R\,.
\end{equation}
For a solid ball containing $R=0$, regularity excludes the second solution. Hence, the material metric is flat only when $\omega$ is constant, corresponding to a uniform dilatational eigenstrain. Every nonconstant regular function $\omega(R)$ then induces a non-flat material metric and therefore an incompatible eigenstrain distribution.
For a thick spherical shell with $R\in[R_i,R_o]$ and $R_i>0$, both solutions are regular. Any function $\omega(R)$ other than these two forms produces a non-flat material metric and hence an incompatible eigenstrain distribution. Such a metric cannot be realized as a stress-free configuration in Euclidean space. Consequently, under the standard assumption that stress vanishes only when the spatial metric induced by the deformation agrees with the material metric, a traction-free equilibrium configuration must contain residual stresses.
We observe that both $\mathcal{I}_1$ and $\mathcal{J}_2$ depend explicitly on the function $\omega(R)$ through the components of the material metric and its inverse. Therefore, the Drucker--Prager strength function of a residually stressed solid depends explicitly on the underlying eigenstrain distribution.
If $\omega=\omega(R,t)$ evolves through the kinetics of the eigenstrain, then $\mathbf{G}$, $\mathcal{I}_1$, and $\mathcal{J}_2$, and hence the representation of the Drucker--Prager strength hypersurface in the reference stress space, evolve with time.
\end{example}

\section{Material Strength of Anisotropic Solids} \label{Material-Strength-Anisotropic-Solids}

In this section, we develop the invariant-theoretic representation of strength functions for anisotropic solids and analyze the role of material symmetry in determining their admissible forms.

\subsection{A brief history of strength surfaces for anisotropic solids} \label{Sec:history}

The modern foundations of failure and similarly motivated yield theories were laid in the classical works of \citet{Mises1913,Mises1928}, in which invariant-based criteria for isotropic materials were introduced.
\citet{Hill1948} proposed a macroscopic yield criterion for anisotropic metals with three mutually orthogonal planes of symmetry, i.e., orthotropic solids, by introducing a homogeneous quadratic plastic potential with six anisotropy parameters, together with associated stress--strain-increment relations. He applied the theory to problems such as rolled sheets, torsion of thin-walled cylinders, and deep drawing, and showed that the model captured experimentally observed directional effects. This work is historically relevant as an early formulation of anisotropic yield surfaces using quadratic forms, although it pertained to plastic yielding rather than brittle material strength.

Building on these earlier works on anisotropic yielding, a number of phenomenological strength and failure criteria were proposed for brittle solids and composites. \citet{Marin1957} proposed an early phenomenological extension of strength theories to anisotropic materials, allowing for different tensile and compressive strengths and formulating a criterion for combined stress states. His work already pointed to the empirical character of anisotropic strength criteria and to the need for additional information beyond simple uniaxial tests. \citet{GoldenblatKopnov1965} proposed a general tensorial strength criterion for anisotropic materials, especially glass-reinforced plastics, in which the failure function is written in terms of tensorial strength coefficients and stress components. The formulation is constructed so that the strength function is a scalar and therefore respects material frame indifference under changes of coordinates. They also showed that several earlier criteria arise as special cases of their general theory and compared its predictions with experiments under combined loading. \citet{AzziTsai1965} adapted Hill's quadratic anisotropic yield criterion to composites and showed how it could be used to predict failure when the loading is not aligned with the principal material directions, while also allowing for different tensile and compressive strengths in different quadrants. \citet{Hoffman1967} proposed a phenomenological failure criterion for orthotropic brittle materials, motivated by analogy with anisotropic yield criteria, especially Hill’s quadratic form, while emphasizing that brittle fracture is fundamentally different from yielding. His criterion introduced distinct tensile, compressive, and shear strength parameters in the material symmetry directions and, for plane stress in a unidirectional composite, led to a fracture surface that he compared with available experimental data.

\citet{TsaiWu1971} proposed a general polynomial failure criterion for anisotropic materials in which the strength surface is expressed as a quadratic function of the stress components, combining linear and quadratic terms to capture different failure modes. They imposed restrictions on the coefficients of this quadratic form so that the resulting strength surface is bounded and physically admissible. Their formulation unifies and generalizes several existing failure criteria and provides a systematic framework for incorporating material anisotropy through experimentally determined strength parameters. \citet{Cowin1986} incorporated anisotropy into a strength criterion by expressing the strength coefficients in terms of a second-order fabric tensor that characterizes the underlying material microstructure.\footnote{Cowin uses the term \emph{fabric tensor} for a symmetric second-order tensor characterizing the geometric arrangement of the microstructural components of a porous or multiphase material. Its precise definition depends on the material and on the method used to quantify its microstructure. In Cowin's formulation, anisotropy is assumed to arise solely from the geometry represented by the fabric tensor, and the strength function is taken to be an isotropic function of the stress, the fabric tensor, and the solid volume fraction. In this constitutive sense, the fabric tensor plays the role of a structural tensor, although the terms are not synonymous: a fabric tensor is a microstructural descriptor, whereas a structural tensor is a more general object used to represent material symmetry in constitutive equations. Cowin's second-order fabric-tensor representation can describe orthotropic, transversely isotropic, and isotropic symmetry, but not monoclinic or triclinic symmetry: three distinct eigenvalues of the fabric tensor yield orthotropy, two equal eigenvalues yield transverse isotropy, and three equal eigenvalues yield isotropy. Since a symmetric second-order tensor possesses mutually orthogonal principal directions, this representation cannot encode monoclinic or triclinic symmetry.} For historical surveys of macroscopic failure criteria for anisotropic and composite materials, including maximum-stress, tensor-polynomial, and quadratic criteria, see \citet{Franklin1968}, \citet{Tsai1984}, and \citet{Fan1987}; for a broader review of strength theories for materials under complex stress states, see also \citet{Yu2002}.

\subsection{Invariant representation of strength function}

Constitutive equations of nonlinear elasticity and anelasticity have been developed systematically using invariant theory. Invariant theory studies scalar and tensorial quantities that remain unchanged under the action of a symmetry group. In continuum mechanics, it provides the mathematical foundation for writing constitutive equations in a form consistent with material symmetry. For isotropic solids, this explains why constitutive equations may be written in terms of scalar invariants of tensors such as $\mathbf{C}$ or $\mathbf{b}$ rather than their individual components. The algebraic foundations of the subject go back to \citet{Hilbert1890}, whose finiteness results showed that invariant algebras admit finite generating sets. \citet{Weyl1939} later reformulated classical invariant theory in representation-theoretic terms, while \citet{Spencer1972,Spencer1982} systematically adapted these ideas to continuum mechanics and nonlinear elasticity, where structural tensors and integrity bases became standard tools.

In elasticity, the practical consequence of invariant theory is the representation theorem viewpoint: once the material symmetry group is specified, invariant theory identifies the scalar invariants that may appear in the elastic energy function or in constitutive response functions, and it also determines the corresponding tensorial representations. This was briefly discussed in \S\ref{Sec:Anisotropic-Elasticity}.

The fracture properties of a solid are ultimately determined by the distribution of defects and microstructural inhomogeneities in the material. In many situations, particularly in brittle solids, these defects are distributed in a manner that reflects the same material symmetries that govern the elastic response. Consequently, the macroscopic strength of the material is expected to inherit the same symmetry group as the underlying elastic constitutive equations.\footnote{\citet{Hoffman1967}, in his discussion of orthotropic brittle materials, observed that it is plausible to assume that the planes of elastic symmetry and the planes of strength symmetry coincide, but he also emphasized that this assumption should be made with caution. He noted that elastic properties are associated with an averaging of the underlying microstructure, whereas fracture strength is generally more sensitive to microstructural inhomogeneities and defects. This suggests that the symmetry group of strength may, in general, be a proper subgroup of the symmetry group of elasticity. In the present work, however, we assume for simplicity that the two symmetry groups coincide.}
Similar to elastic properties and their symmetries, the symmetry of material strength is therefore inherently a material notion and must be defined on the material manifold. In particular, the action of the material symmetry group is naturally expressed in terms of material stress and strain measures. Hence, in order to characterize the effect of symmetry on strength, one must start with a material description in terms of $(\mathbf{S},\mathbf{C}^\flat)$ or $(\mathbf{P},\mathbf{F})$. In view of \eqref{W-Material-Symmetry}, it is then natural to require that the strength function be invariant under the same material symmetry group. Thus, if $\mathcal{G}_X$ is the material symmetry group at $X\in\mathcal{B}$, then one must have
\begin{equation}\label{F-Material-Symmetry}
	\hat{\mathsf{F}}(\mathbf{K}^*\mathbf{S},\mathbf{K}^*\mathbf{C}^\flat,\mathbf{G})
	=\hat{\mathsf{F}}(\mathbf{S},\mathbf{C}^\flat,\mathbf{G})\,,\qquad
	\forall\,\mathbf{K}\in\mathcal{G}_X\leqslant \mathrm{Orth}(\mathbf{G})\,.
\end{equation}
Equation \eqref{F-Material-Symmetry} is the statement that the strength function has the same material symmetry group as the elastic energy function. In other words, the fracture criterion inherits the same material symmetries as the elastic response.

The material symmetry group can be characterized by a finite collection of structural tensors $\boldsymbol{\zeta}_i$ of order $\mu_i$, $i=1,\dots,N$ \citep{liu1982,boehler1987,zheng1993,zheng1994theory,Lu2000,MazzucatoRachele2006}. More precisely,
\begin{equation} \label{invar}
	\mathbf{Q}\in \mathcal{G}_X \leqslant\mathrm{Orth}(\mathbf{G}) \iff 
	\langle\mathbf{Q}\rangle_{\mu_1}\boldsymbol{\zeta}_1=\boldsymbol{\zeta}_1\,,
	\dots\,,
	\langle\mathbf{Q}\rangle_{\mu_N}\boldsymbol{\zeta}_N=\boldsymbol{\zeta}_N\,.
\end{equation}
Thus, the material symmetry group $\mathcal{G}_X$ is precisely the invariance group of the structural tensors $\boldsymbol{\zeta}_i$, $i=1,\dots,N$. In this sense, the structural tensors encode the material symmetry relevant to strength.

For a $\mathbf{G}$-orthogonal transformation $\mathbf{Q}$ and a tensor $\boldsymbol{\zeta}$ of order $\mu$, the $\mu$-th Kronecker power $\langle\mathbf{Q}\rangle_{\mu}$ is defined by
\begin{equation}
	(\langle\mathbf{Q}\rangle_{\mu}\boldsymbol{\zeta})^{\bar{A}_1\dots\bar{A}_{\mu}}
	=Q^{\bar{A}_1}{}_{A_1}\dots Q^{\bar{A}_{\mu}}{}_{A_{\mu}}\,
	\zeta^{A_1\dots A_{\mu}}\,.
\end{equation}
In particular, for arbitrary vectors $\mathbf{V}_i\in T_X\mathcal{B}$, $i=1,\dots,m$, one has
\begin{equation}
	\langle\mathbf{Q}\rangle_{m}\left(\mathbf{V}_1\otimes\dots\otimes\mathbf{V}_m\right)
	=\mathbf{Q}\mathbf{V}_1\otimes\dots\otimes\mathbf{Q}\mathbf{V}_m\,.
\end{equation}
Accordingly, one may write the strength function in the form
\begin{equation}\label{F-Structural-Tensors}
	\mathsf{F}=\widehat{\mathsf{F}}(X,\mathbf{S},\mathbf{C}^\flat,\mathbf{G},
	\boldsymbol{\zeta}_1,\dots,\boldsymbol{\zeta}_N)\,.
\end{equation}
Denoting the set of structural tensors by $\boldsymbol{\Lambda}=\{\boldsymbol{\zeta}_1,\dots,\boldsymbol{\zeta}_N\}$, one simply writes $\mathsf{F}=\widehat{\mathsf{F}}(X,\mathbf{S},\mathbf{C}^\flat,\mathbf{G}, \boldsymbol{\Lambda})$.
When the structural tensors are included among its arguments, the strength function becomes an isotropic scalar-valued function of its arguments, in the same sense as in the classical principle of isotropy of space \citep{Boehler1979}.

Instead of working directly with the tensors $\{\mathbf{S},\mathbf{C}^\flat,\mathbf{G},\boldsymbol{\zeta}_1,\dots,\boldsymbol{\zeta}_N\}$, one may equivalently use a corresponding set of isotropic invariants. By Hilbert's theorem, any finite collection of tensors admits a finite integrity basis for the algebra of isotropic invariants \citep{Spencer1971}. Denoting such an integrity basis by $I_j$, $j=1,\dots,m$, one may therefore write
\begin{equation}
	\mathsf{F}=\mathsf{F}(X,I_1,\dots,I_m)\,.
\end{equation}

\begin{example}
For a transversely isotropic solid, the material symmetry group is characterized by a single material preferred direction at each material point. Let $\mathbf{N}(X)\in T_X\mathcal{B}$ be a unit vector normal to the plane of isotropy, i.e., $\llangle \mathbf{N},\mathbf{N} \rrangle_{\mathbf{G}}=1$. In this case, a single structural tensor is sufficient to characterize the symmetry. It should be emphasized, however, that the choice of structural tensor is not unique. One convenient choice is
\begin{equation}
	\boldsymbol{\zeta}=\mathbf{N}\otimes\mathbf{N}\,.
\end{equation}
Accordingly, the strength function may be written as $\mathsf{F}=\widehat{\mathsf{F}}(X,\mathbf{S},\mathbf{C}^\flat,\mathbf{G},\mathbf{N}\otimes\mathbf{N})$. In this case, the integrity basis consists of five invariants, namely the three principal invariants of $\mathbf{C}$ together with the additional invariants $I_4=\mathbf{C}^\flat(\mathbf{N},\mathbf{N})$ and $I_5=(\mathbf{C}^\flat\mathbf{G}^{\sharp}\mathbf{C}^\flat)(\mathbf{N},\mathbf{N})$.
\end{example}

\subsubsection{Isotropic solids}

First, let us consider isotropic solids. An isotropic function of two symmetric second-order tensors can be expressed in terms of the following ten invariants \citep{RivlinEricksen1955}:
\begin{equation}
	\operatorname{tr}\mathbf{S}\,,~~
	\operatorname{tr}\mathbf{S}^2\,,~~
	\operatorname{tr}\mathbf{S}^3\,,~~
	\operatorname{tr}\mathbf{C} \,,~~
	\operatorname{tr}\mathbf{C}^2 \,,~~
	\operatorname{tr}\mathbf{C}^3 \,,~~
	\operatorname{tr}\left(\mathbf{S}\mathbf{C}\right)\,,~~
	\operatorname{tr}\left(\mathbf{S}\mathbf{C}^2\right)\,,~~
	\operatorname{tr}\left(\mathbf{S}^2\mathbf{C}\right)\,,~~
	\operatorname{tr}\left(\mathbf{S}^2\mathbf{C}^2\right)\,.
\end{equation}
Equivalently, we can  use the following invariants\footnote{Clearly, one can use the pair $I_4$ and $I_5$ instead of $I_4$ and $\operatorname{tr}\mathbf{C}^2$. Using the Cayley-Hamilton theorem, one can  show that $\operatorname{tr}\mathbf{C}^3=I_4^3-3I_4I_5+3I_6$ \citep{Spencer1971}.}
\begin{equation} \label{I-Invariants}
\begin{aligned}
	& \mathcal{I}_1=\operatorname{tr}\mathbf{S}\,,\quad
	\mathcal{I}_2=\operatorname{tr}\mathbf{S}^2\,,\quad
	\mathcal{I}_3=\operatorname{tr}\mathbf{S}^3\,,\quad
	\mathcal{I}_4=\operatorname{tr}\mathbf{C} \,,\quad
	\mathcal{I}_5=\frac{1}{2}\left[ \mathcal{I}_4^2- \operatorname{tr}\mathbf{C}^2\right]\,,\quad
	\mathcal{I}_6=\det\mathbf{C} \,,\quad \\
	& \mathcal{I}_7=\operatorname{tr}\left(\mathbf{S}\mathbf{C}\right)\,,\quad
	\mathcal{I}_8=\operatorname{tr}\left(\mathbf{S}\mathbf{C}^2 \right)\,,\quad
	\mathcal{I}_9=\operatorname{tr}\left(\mathbf{S}^2\mathbf{C} \right)\,,\quad
	\mathcal{I}_{10}=\operatorname{tr}\left(\mathbf{S}^2\mathbf{C}^2 \right)
	\,.
\end{aligned}
\end{equation}
Therefore, the material strength function has the following functional form: $\hat{\mathsf{F}}(\mathbf{S},\mathbf{C}^\flat,\mathbf{G})=\bar{\mathsf{F}}(\mathcal{I}_1,\cdots,\mathcal{I}_{10})$.

Even for isotropic solids, once the strength function is allowed to depend on both stress and strain, the corresponding integrity basis becomes rather large. For this reason, and also in preparation for the anisotropic case where the number of invariants increases further, we restrict attention here to the special case in which the strength function depends only on stress. Thus, for isotropic solids the strength function admits the representation $\hat{\mathsf{F}}(\mathbf{S},\mathbf{G})=\bar{\mathsf{F}}(\mathcal{I}_1,\mathcal{I}_2,\mathcal{I}_3)$.

\subsubsection{Transversely isotropic solids}

Recall that a transversely isotropic solid has, at each material point, a single preferred direction orthogonal to the plane of isotropy. This preferred direction at $X\in\mathcal{B}$ is denoted as $\mathbf{N}(X)$. In addition to $\mathcal{I}_1$, $\mathcal{I}_2$, and $\mathcal{I}_3$, we have the following two extra invariants
\begin{equation} 
	\mathcal{I}_4=\mathbf{N}\cdot\mathbf{S}\cdot\mathbf{N}\,,\qquad 
	\mathcal{I}_5=\mathbf{N}\cdot\mathbf{S}^2\cdot\mathbf{N}\,.
\end{equation}
Therefore, for a transversely isotropic solid, the strength function has the following invariant representation 
\begin{equation}
	\hat{\mathsf{F}} = \bar{\mathsf{F}}(\mathcal{I}_1,\mathcal{I}_2,\mathcal{I}_3,
	\mathcal{I}_4,\mathcal{I}_5)\,.
\end{equation}

\subsubsection{Orthotropic solids}

An orthotropic solid is characterized by symmetry with respect to three mutually orthogonal reflection planes. Let $\mathbf{N}_1(X)$, $\mathbf{N}_2(X)$, and $\mathbf{N}_3(X)$ denote a $\mathbf{G}$-orthonormal triad that defines the corresponding principal material directions at $X$.
Beyond $\mathcal{I}_1$, $\mathcal{I}_2$, and $\mathcal{I}_3$, four additional invariants arise:
\begin{equation} 
	\mathcal{I}_4=\mathbf{N}_1\cdot\mathbf{S}\cdot\mathbf{N}_1\,,\qquad 
	\mathcal{I}_5=\mathbf{N}_1\cdot\mathbf{S}^2\cdot\mathbf{N}_1\,,\qquad
	\mathcal{I}_6=\mathbf{N}_2\cdot\mathbf{S}\cdot\mathbf{N}_2\,,\qquad
	\mathcal{I}_7=\mathbf{N}_2\cdot\mathbf{S}^2\cdot\mathbf{N}_2\,.
\end{equation}
Accordingly, for an orthotropic solid the strength function admits the following invariant representation:
\begin{equation}
	\hat{\mathsf{F}} = \bar{\mathsf{F}}(\mathcal{I}_1,\cdots,\mathcal{I}_7)\,.
\end{equation}

\subsubsection{Monoclinic solids}

A monoclinic solid is characterized by a set of three preferred material directions described by unit vectors $\{\mathbf{N}_1,\mathbf{N}_2,\mathbf{N}_3\}$, where $\mathbf{N}_1$ and $\mathbf{N}_2$ are not orthogonal and $\mathbf{N}_3$ is orthogonal to the plane they span.
In addition to the seven invariants associated with orthotropic solids, monoclinic solids admit two further invariants:
\begin{equation}
	\mathcal{I}_8=\mathfrak{g}\,\mathbf{N}_1\cdot\mathbf{S}\cdot\mathbf{N}_2,\qquad
	\mathcal{I}_9=\mathfrak{g}^2\,,
\end{equation}  
where $\mathfrak{g}=\mathbf{N}_1\cdot\mathbf{N}_2$. Thus, for a monoclinic solid the strength function can be represented as
\begin{equation}
	\hat{\mathsf{F}} = \bar{\mathsf{F}}(\mathcal{I}_1,\cdots,\mathcal{I}_9)
	\,.
\end{equation}

\subsubsection{Strength function for anisotropic anelastic solids}

Let us denote the set of structural tensors with respect to the intermediate configuration by $\mathring{\boldsymbol{\Lambda}}=\{\mathring{\boldsymbol{\zeta}}_1,\dots,\mathring{\boldsymbol{\zeta}}_N\}$. 
With respect to this local stress-free configuration the strength function has the representation 
\begin{equation}
	\mathsf{F}(X,\mathring{\mathbf{P}},\Fe,\mathbf{g},\mathring{\mathbf{G}},\mathring{\boldsymbol{\Lambda}})=0\,.
\end{equation}
Using covariance under the anelastic distortion $\Fa$, one obtains
\begin{equation} 
	\mathsf{F}(X,\mathring{\mathbf{P}},\Fe,\mathbf{g},\mathring{\mathbf{G}},\mathring{\boldsymbol{\Lambda}})
	=\mathsf{F}(X,\Fa^*\mathring{\mathbf{P}},\Fa^*\Fe,\mathbf{g},\Fa^*\mathring{\mathbf{G}},\Fa^*\mathring{\boldsymbol{\Lambda}})
	=\mathsf{F}(X,\mathbf{P},\mathbf{F},\mathbf{g},\mathbf{G},\boldsymbol{\Lambda})\,,
\end{equation}
where $\boldsymbol{\Lambda}=\Fa^*\mathring{\boldsymbol{\Lambda}}$ is the set of anelastic structural tensors. 
The treatment of anisotropy then proceeds as for anisotropic elastic solids without eigenstrains, except that the structural tensors are replaced by their anelastically transformed counterparts and all invariants are formed using the material metric $\mathbf{G}$ rather than the flat metric $\mathring{\mathbf{G}}$.

\section{Conclusions}  \label{Sec:Conclusions}

In this paper, we formulated material strength for brittle solids in the setting of finite elasticity. We showed that a strength criterion expressed solely in terms of a chosen stress measure generally acquires dependence on the corresponding strain when rewritten using another stress measure. The pair $(\text{stress},\text{strain})$ therefore provides the natural domain for a spatially covariant formulation that is independent of the stress measure used to represent the criterion. Purely stress-based strength criteria remain admissible as an important special class for which the strain dependence is absent in the chosen representation.

We showed that the representation of a strength function depends on the choice of stress measure. Although the various stress measures are mechanically equivalent in the sense that they represent the same physical traction vector when paired with the appropriate area element, the notion of strength depends, in general, on which stress measure is assumed to be homogeneous at the onset of fracture. This leads naturally to the problem of relating strength functions written in terms of the first and second Piola--Kirchhoff, and Cauchy stresses, and we showed that spatial covariance is the fundamental principle governing the relation among these different representations.

Restricting attention to the stress-based criteria that appear in the existing literature, we defined the corresponding strength hypersurface and discussed the associated safe domain. We distinguished constitutive admissibility from fracture and clarified that the strength hypersurface should be understood as a hypersurface in the constitutively admissible stress manifold. 
We showed that, for stress-based strength functions satisfying standard regularity conditions and the requirement that sufficiently large stresses are inadmissible, the strength surface is a smooth compact hypersurface of the constitutively admissible stress manifold, defined as the zero level set of the strength function separating admissible from inadmissible stress states.
We showed that, under a natural proportional-reduction hypothesis, the safe domain is star-shaped with respect to the origin. We verified this property for several classical isotropic strength criteria, including the Mohr--Coulomb, Hoek--Brown, generalized three-dimensional Hoek--Brown, and Mogi--Coulomb surfaces.

We also discussed the role of material symmetry in strength. Motivated by the fact that fracture properties are determined by the same underlying microstructure that governs elastic response, we argued that the strength function should inherit the same material symmetry group as the constitutive equations. This provides a natural framework for characterizing anisotropic strength hypersurfaces and their invariant-theoretic representations.

Finally, we showed that this framework extends naturally to materials with internal constraints and to solids with residual stress or distributed eigenstrains, for which the appropriate strength function must be referred to the local relaxed configuration. In this setting, residual stresses do not enter merely as external parameters, but are encoded through the material metric, which modifies the invariants entering the strength function. Consequently, even when the strength function is expressed in terms of stress invariants, the corresponding strength hypersurface in stress space depends on the internal anelastic state. In particular, eigenstrains alter both the admissible set of stresses and the geometry of the strength surface, so that strength can no longer be characterized by a universal stress-based criterion independent of the material state. This highlights the essential role of the relaxed configuration and shows that, in the presence of residual stress, strength must be understood as a covariant relation involving both stress and the underlying anelastic structure. Taken together, these results provide a systematic and covariant continuum-mechanical formulation of material strength for brittle solids in both isotropic and anisotropic settings.

The present work provides a rational formulation of material strength for brittle solids within the framework of finite elasticity. Several natural directions for future work arise. A first is the extension to dissipative and rate-dependent solids, especially viscoelastic materials, where the strength function is expected to depend on additional internal variables and on the loading history. A second important direction is the study of fracture in residually stressed solids, where the interaction of strength, incompatibility, and the material metric is expected to play a fundamental role. Another natural direction is the extension of the present framework to more general classes of anisotropic solids and to evolving material symmetries.

\section*{Acknowledgement}

AY benefited from discussions with Patrizio Neff. AK acknowledges support from the National Science Foundation, United States, through the grant CMMI-2404808.

\bibliographystyle{abbrvnat}
\bibliography{ref}

\begin{thebibliography}{156}
\providecommand{\natexlab}[1]{#1}
\providecommand{\url}[1]{\texttt{#1}}
\expandafter\ifx\csname urlstyle\endcsname\relax
  \providecommand{\doi}[1]{doi: #1}\else
  \providecommand{\doi}{doi: \begingroup \urlstyle{rm}\Url}\fi

\bibitem[Ambrosi et~al.(2019)Ambrosi, Ben~Amar, Cyron, DeSimone, Goriely,
  Humphrey, and Kuhl]{ambrosi2019growth}
D.~Ambrosi, M.~Ben~Amar, C.~J. Cyron, A.~DeSimone, A.~Goriely, J.~D. Humphrey,
  and E.~Kuhl.
\newblock Growth and remodelling of living tissues: perspectives, challenges
  and opportunities.
\newblock \emph{Journal of the Royal Society Interface}, 16\penalty0
  (157):\penalty0 20190233, 2019.

\bibitem[Amor et~al.(2009)Amor, Marigo, and Maurini]{amor2009}
H.~Amor, J.-J. Marigo, and C.~Maurini.
\newblock Regularized formulation of the variational brittle fracture with
  unilateral contact: Numerical experiments.
\newblock \emph{Journal of the Mechanics and Physics of Solids}, 57\penalty0
  (8):\penalty0 1209--1229, 2009.

\bibitem[Anderson(2005)]{anderson2005}
T.~L. Anderson.
\newblock \emph{Fracture Mechanics: Fundamentals and Applications}.
\newblock CRC press, 2005.

\bibitem[Ashkenazi(1965)]{Ashkenazi1965}
E.~K. Ashkenazi.
\newblock Problems of the anisotropy of strength.
\newblock \emph{Polymer Mechanics}, 1\penalty0 (2):\penalty0 60--70, 1965.

\bibitem[Azzi and Tsai(1965)]{AzziTsai1965}
V.~D. Azzi and S.~W. Tsai.
\newblock Anisotropic strength of composites: Investigation aimed at developing
  a theory applicable to laminated as well as unidirectional composites,
  employing simple material properties derived from unidirectional specimens
  alone.
\newblock \emph{Experimental Mechanics}, 5\penalty0 (9):\penalty0 283--288,
  1965.

\bibitem[Badel et~al.(2007)Badel, Godard, and Leblond]{badel2007}
P.~Badel, V.~Godard, and J.-B. Leblond.
\newblock Application of some anisotropic damage model to the prediction of the
  failure of some complex industrial concrete structure.
\newblock \emph{International Journal of Solids and Structures}, 44\penalty0
  (18-19):\penalty0 5848--5874, 2007.

\bibitem[Ball(1976)]{Ball1976}
J.~M. Ball.
\newblock Convexity conditions and existence theorems in nonlinear elasticity.
\newblock \emph{Archive for Rational Mechanics and Analysis}, 63\penalty0
  (4):\penalty0 337--403, 1976.

\bibitem[Ball and James(1987)]{BallJames1987}
J.~M. Ball and R.~D. James.
\newblock Fine phase mixtures as minimizers of energy.
\newblock \emph{Archive for Rational Mechanics and Analysis}, 100\penalty0
  (1):\penalty0 13--52, 1987.

\bibitem[Barenblatt(1959)]{barenblatt1959}
G.~I. Barenblatt.
\newblock The formation of equilibrium cracks during brittle fracture. general
  ideas and hypotheses. axially-symmetric cracks.
\newblock \emph{Journal of Applied Mathematics and Mechanics}, 23\penalty0
  (3):\penalty0 622--636, 1959.

\bibitem[Beltrami(1885)]{Beltrami1885}
E.~Beltrami.
\newblock Sull' interpretazione meccanica delle formule di {M}axwell.
\newblock \emph{Rendiconti del Reale Istituto Lombardo di Scienze e Lettere},
  18:\penalty0 141--148, 1885.

\bibitem[Bilby et~al.(1957)Bilby, Lardner, and Stroh]{bilby1957continuous}
B.~A. Bilby, L.~R.~T. Lardner, and A.~N. Stroh.
\newblock Continuous distributions of dislocations and the theory of
  plasticity.
\newblock In \emph{Actes du IXe congr\`es international de m\'ecanique
  appliqu\'ee, (Bruxelles, 1956)}, volume~8, pages 35--44, 1957.

\bibitem[Boehler(1979)]{Boehler1979}
J.-P. Boehler.
\newblock A simple derivation of representations for non-polynomial
  constitutive equations in some cases of anisotropy.
\newblock \emph{Zeitschrift f{\"u}r Angewandte Mathematik und Mechanik},
  59\penalty0 (4):\penalty0 157--167, 1979.

\bibitem[Boehler(1987)]{boehler1987}
J.-P. Boehler.
\newblock \emph{Applications of Tensor Functions in Solid Mechanics}, volume
  292.
\newblock Springer, 1987.

\bibitem[Bonacci et~al.(2026)Bonacci, Dolbow, and
  Guilleminot]{bonacci2026stochastic}
A.~Bonacci, J.~Dolbow, and J.~Guilleminot.
\newblock Stochastic modeling of anisotropic strength surfaces from atomistic
  simulations.
\newblock \emph{Theoretical and Applied Fracture Mechanics}, page 105594, 2026.

\bibitem[Bourdin et~al.(2000)Bourdin, Francfort, and Marigo]{Bourdin00}
B.~Bourdin, G.~A. Francfort, and J.-J. Marigo.
\newblock Numerical experiments in revisited brittle fracture.
\newblock \emph{Journal of the Mechanics and Physics of Solids}, 48\penalty0
  (4):\penalty0 797--826, 2000.

\bibitem[Bourdin et~al.(2008)Bourdin, Francfort, and Marigo]{Bourdin08}
B.~Bourdin, G.~A. Francfort, and J.-J. Marigo.
\newblock The variational approach to fracture.
\newblock \emph{Journal of elasticity}, 91:\penalty0 5--148, 2008.

\bibitem[Bourdin et~al.(2025)Bourdin, Marigo, Maurini, and
  Zolesi]{bourdin2025variational}
B.~Bourdin, J.-J. Marigo, C.~Maurini, and C.~Zolesi.
\newblock A variational approach to fracture incorporating any convex strength
  criterion.
\newblock \emph{arXiv preprint arXiv:2506.22558}, 2025.

\bibitem[Carroll(1973{\natexlab{a}})]{Carroll1973}
M.~M. Carroll.
\newblock Controllable states of stress for compressible elastic solids.
\newblock \emph{Journal of Elasticity}, 3:\penalty0 57--61, 1973{\natexlab{a}}.

\bibitem[Carroll(1973{\natexlab{b}})]{Carroll1973ControllableIncompressible}
M.~M. Carroll.
\newblock {Controllable states of stress for incompressible elastic solids}.
\newblock \emph{Journal of Elasticity}, 3\penalty0 (2):\penalty0 147--153,
  1973{\natexlab{b}}.

\bibitem[Cauchy(1828)]{Cauchy1828}
A.-L. Cauchy.
\newblock Sur les {\'e}quations qui expriment les conditions d'{\'e}quilibre ou
  les lois du mouvement int{\'e}rieur d'un corps solide, {\'e}lastique ou non
  {\'e}lastique.
\newblock \emph{Exercises de Math{\'e}matiques}, 3:\penalty0 160--187, 1828.

\bibitem[Chambolle et~al.(2009)Chambolle, Francfort, and Marigo]{chambolle2009}
A.~Chambolle, G.~A. Francfort, and J.-J. Marigo.
\newblock When and how do cracks propagate?
\newblock \emph{Journal of the Mechanics and Physics of Solids}, 57\penalty0
  (9):\penalty0 1614--1622, 2009.

\bibitem[Chen et~al.(2017)Chen, Wang, and Suo]{chen2017}
C.~Chen, Z.~Wang, and Z.~Suo.
\newblock Flaw sensitivity of highly stretchable materials.
\newblock \emph{Extreme Mechanics Letters}, 10:\penalty0 50--57, 2017.

\bibitem[Chockalingam et~al.(2026)Chockalingam, Tepole, and
  Kumar]{Chockalingam2026}
S.~Chockalingam, A.~B. Tepole, and A.~Kumar.
\newblock The phase-field model of fracture incorporating {M}ohr–{C}oulomb,
  {M}ogi–{C}oulomb, and {H}oek–{B}rown strength surfaces.
\newblock \emph{Engineering Fracture Mechanics}, 340:\penalty0 112108, 2026.

\bibitem[Ciarlet(1988)]{Ciarlet1988}
P.~G. Ciarlet.
\newblock \emph{Mathematical Elasticity, Volume I: Three-Dimensional
  Elasticity}.
\newblock North-Holland, Amsterdam, 1988.

\bibitem[Comi(2001)]{comi2001}
C.~Comi.
\newblock A non-local model with tension and compression damage mechanisms.
\newblock \emph{European Journal of Mechanics-A/Solids}, 20\penalty0
  (1):\penalty0 1--22, 2001.

\bibitem[Cowin(1986)]{Cowin1986}
S.~C. Cowin.
\newblock Fabric dependence of an anisotropic strength criterion.
\newblock \emph{Mechanics of Materials}, 5\penalty0 (3):\penalty0 251--260,
  1986.

\bibitem[Doyle and Ericksen(1956)]{Doyle1956}
T.~C. Doyle and J.~L. Ericksen.
\newblock Nonlinear elasticity.
\newblock \emph{Advances in Applied Mechanics}, 4:\penalty0 53--115, 1956.

\bibitem[Dugdale(1960)]{dugdale1960}
D.~S. Dugdale.
\newblock Yielding of steel sheets containing slits.
\newblock \emph{Journal of the Mechanics and Physics of Solids}, 8\penalty0
  (2):\penalty0 100--104, 1960.

\bibitem[Ebin and Marsden(1970)]{EbinMarsden1970}
D.~G. Ebin and J.~Marsden.
\newblock Groups of {D}iffeomorphisms and the motion of an incompressible
  fluid.
\newblock \emph{Annals of Mathematics}, 92\penalty0 (1):\penalty0 102--163,
  1970.

\bibitem[Eckart(1948)]{Eckart1948}
C.~Eckart.
\newblock The thermodynamics of irreversible processes. {IV}. {T}he theory of
  elasticity and anelasticity.
\newblock \emph{Physical Review}, 73\penalty0 (4):\penalty0 373, 1948.

\bibitem[Ericksen(1954)]{Ericksen1954}
J.~L. Ericksen.
\newblock Deformations possible in every isotropic, incompressible, perfectly
  elastic body.
\newblock \emph{Zeitschrift f{\"u}r Angewandte Mathematik und Physik},
  5\penalty0 (6):\penalty0 466--489, 1954.

\bibitem[Ericksen(1955)]{Ericksen1955}
J.~L. Ericksen.
\newblock Deformations possible in every compressible, isotropic, perfectly
  elastic material.
\newblock \emph{Journal of Mathematics and Physics}, 34\penalty0
  (1-4):\penalty0 126--128, 1955.

\bibitem[Eshelby(1957)]{Eshelby1957}
J.~D. Eshelby.
\newblock The determination of the elastic field of an ellipsoidal inclusion,
  and related problems.
\newblock \emph{Proceedings of the Royal Society of London A}, 241\penalty0
  (1226):\penalty0 376--396, 1957.

\bibitem[Fan(1987)]{Fan1987}
W.-x. Fan.
\newblock On phenomenological anisotropic failure criteria.
\newblock \emph{Composites Science and Technology}, 28:\penalty0 269--278,
  1987.

\bibitem[Francfort and Marigo(1998)]{FrancfortMarigo1998}
G.~A. Francfort and J.-J. Marigo.
\newblock Revisiting brittle fracture as an energy minimization problem.
\newblock \emph{Journal of the Mechanics and Physics of Solids}, 46\penalty0
  (8):\penalty0 1319--1342, 1998.

\bibitem[Franklin(1968)]{Franklin1968}
H.~G. Franklin.
\newblock Classic theories of failure of anisotropic materials.
\newblock \emph{Fibre Science and Technology}, 1\penalty0 (2):\penalty0
  137--150, 1968.

\bibitem[Giulini(2007)]{Giulini2007}
D.~Giulini.
\newblock Some remarks on the notions of general covariance and background
  independence.
\newblock \emph{Lecture Notes in Physics}, 721:\penalty0 105--120, 2007.

\bibitem[Gol'Denblat and Kopnov(1965)]{GoldenblatKopnov1965}
I.~I. Gol'Denblat and V.~A. Kopnov.
\newblock Strength of glass-reinforced plastics in the complex stress state.
\newblock \emph{Polymer Mechanics}, 1\penalty0 (2):\penalty0 54--59, 1965.

\bibitem[Green and Naghdi(1965)]{GreenNaghdi1965}
A.~E. Green and P.~M. Naghdi.
\newblock A general theory of an elastic-plastic continuum.
\newblock \emph{Archive for Rational Mechanics and Analysis}, 18:\penalty0
  251--281, 1965.

\bibitem[Green and Rivlin(1964)]{Green64}
A.~E. Green and R.~S. Rivlin.
\newblock On {C}auchy's equations of motion.
\newblock \emph{Zeitschrift f{\"u}r Angewandte Mathematik und Physik},
  15\penalty0 (3):\penalty0 290--292, 1964.

\bibitem[Green(1838)]{Green1838}
G.~Green.
\newblock On the laws of the reflexion and refraction of light at the common
  surface of two non-crystallized media.
\newblock \emph{Transactions of the Cambridge Philosophical Society},
  7:\penalty0 1, 1838.

\bibitem[Green(1839)]{Green1839}
G.~Green.
\newblock On the propagation of light in crystallized media.
\newblock \emph{Transactions of the Cambridge Philosophical Society},
  7:\penalty0 121, 1839.

\bibitem[Griffith(1921)]{Griffith1921}
A.~A. Griffith.
\newblock {VI}. {T}he phenomena of rupture and flow in solids.
\newblock \emph{Philosophical Transactions of the Royal Society of London.
  Series A}, 221\penalty0 (582-593):\penalty0 163--198, 1921.

\bibitem[Gurtin and Spear(1983)]{Gurtin1983}
M.~E. Gurtin and K.~Spear.
\newblock On the relationship between the logarithmic strain rate and the
  stretching tensor.
\newblock \emph{International Journal of Solids and Structures}, 19\penalty0
  (5):\penalty0 437--444, 1983.

\bibitem[Hamdi et~al.(2006)Hamdi, Abdelaziz, A{\"\i}t~Hocine, Heuillet, and
  Benseddiq]{hamdi2006}
A.~Hamdi, M.~N. Abdelaziz, N.~A{\"\i}t~Hocine, P.~Heuillet, and N.~Benseddiq.
\newblock A fracture criterion of rubber-like materials under plane stress
  conditions.
\newblock \emph{Polymer Testing}, 25\penalty0 (8):\penalty0 994--1005, 2006.

\bibitem[Hencky(1924)]{Hencky1924}
H.~Hencky.
\newblock Zur {T}heorie plastischer {D}eformationen und der hierdurch im
  {M}aterial hervorgerufenen {N}achspannungen.
\newblock \emph{Zeitschrift f{\"u}r Angewandte Mathematik und Mechanik},
  4:\penalty0 323--334, 1924.

\bibitem[Hilbert(1890)]{Hilbert1890}
D.~Hilbert.
\newblock {U}eber die theorie der algebraischen formen.
\newblock \emph{Mathematische Annalen}, 36:\penalty0 473--534, 1890.

\bibitem[Hill(1948)]{Hill1948}
R.~Hill.
\newblock A theory of the yielding and plastic flow of anisotropic metals.
\newblock \emph{Proceedings of the Royal Society of London. Series A.
  Mathematical and Physical Sciences}, 193\penalty0 (1033):\penalty0 281--297,
  1948.

\bibitem[Hill(1968)]{Hill1968}
R.~Hill.
\newblock On constitutive inequalities for simple materials--{I}.
\newblock \emph{Journal of the Mechanics and Physics of Solids}, 16\penalty0
  (4):\penalty0 229--242, 1968.

\bibitem[Hill(1970)]{Hill1970}
R.~Hill.
\newblock Constitutive inequalities for isotropic elastic solids under finite
  strain.
\newblock \emph{Proceedings of the Royal Society of London A}, 314\penalty0
  (1519):\penalty0 457--472, 1970.

\bibitem[Hill(1978)]{Hill1978}
R.~Hill.
\newblock Aspects of invariance in solid mechanics.
\newblock \emph{Advances in Applied Mechanics}, 18:\penalty0 1--75, 1978.

\bibitem[Hirsch(1976)]{Hirsch1976}
M.~W. Hirsch.
\newblock \emph{Differential Topology}, volume~33 of \emph{Graduate Texts in
  Mathematics}.
\newblock Springer-Verlag, New York, 1976.

\bibitem[Hoffman(1967)]{Hoffman1967}
O.~Hoffman.
\newblock The brittle strength of orthotropic materials.
\newblock \emph{Journal of Composite Materials}, 1\penalty0 (2):\penalty0
  200--206, 1967.

\bibitem[Hoger(1985)]{Hoger1985}
A.~Hoger.
\newblock On the residual stress possible in an elastic body with material
  symmetry.
\newblock \emph{Archive for Rational Mechanics and Analysis}, 88\penalty0
  (3):\penalty0 271--289, 1985.

\bibitem[Hoger(1986)]{Hoger1986}
A.~Hoger.
\newblock The material time derivative of logarithmic strain.
\newblock \emph{International Journal of Solids and Structures}, 22\penalty0
  (9):\penalty0 1019--1032, 1986.

\bibitem[Hoger(1987)]{Hoger1987}
A.~Hoger.
\newblock The stress conjugate to logarithmic strain.
\newblock \emph{International Journal of Solids and Structures}, 23\penalty0
  (12):\penalty0 1645--1656, 1987.

\bibitem[Huber(1904)]{Huber1904}
M.~T. Huber.
\newblock O podstawach wytrzyma{\l}o{\'s}ci materja{\l}{\'o}w.
\newblock \emph{Czasopismo Techniczne}, 1904.

\bibitem[Huber(2004)]{Huber2004}
M.~T. Huber.
\newblock Specific work of strain as a measure of material effort.
\newblock \emph{Archives of Mechanics}, 56\penalty0 (3):\penalty0 173--190,
  2004.
\newblock English translation of the 1904 Polish original.

\bibitem[Hughes and Marsden(1977)]{HuMa1977}
T.~J.~R. Hughes and J.~E. Marsden.
\newblock Some applications of geometry is continuum mechanics.
\newblock \emph{Reports on Mathematical Physics}, 12\penalty0 (1):\penalty0
  35--44, 1977.

\bibitem[Jun and Korsunsky(2010)]{Jun2010}
T.-S. Jun and A.~M. Korsunsky.
\newblock Evaluation of residual stresses and strains using the eigenstrain
  reconstruction method.
\newblock \emph{International Journal of Solids and Structures}, 47\penalty0
  (13):\penalty0 1678--1686, 2010.

\bibitem[Kamarei et~al.(2025)Kamarei, Breedlove, and
  Lopez-Pamies]{KamareiBreedloveLopezPamies2025}
F.~Kamarei, E.~Breedlove, and O.~Lopez-Pamies.
\newblock {N}ucleation and propagation of fracture in viscoelastic elastomers:
  {A} complete phase-field theory.
\newblock \emph{Computer Methods in Applied Mechanics and Engineering},
  446:\penalty0 118337, 2025.

\bibitem[Kamarei et~al.(2026)Kamarei, Zeng, Dolbow, and
  Lopez-Pamies]{Kamarei2026}
F.~Kamarei, B.~Zeng, J.~E. Dolbow, and O.~Lopez-Pamies.
\newblock Nine circles of elastic brittle fracture.
\newblock \emph{Computer Methods in Applied Mechanics and Engineering},
  448:\penalty0 118449, 2026.

\bibitem[Kinoshita and Mura(1971)]{Kinoshita1971}
N.~Kinoshita and T.~Mura.
\newblock Elastic fields of inclusions in anisotropic media.
\newblock \emph{Physica Status Solidi (a)}, 5\penalty0 (3):\penalty0 759--768,
  1971.

\bibitem[Knauss(1967)]{Knauss1967}
W.~G. Knauss.
\newblock {A}n upper bound of failure in viscoelastic materials subjected to
  multiaxial stress states.
\newblock \emph{International Journal of Fracture}, 3:\penalty0 267--277, 1967.

\bibitem[Kondo(1949)]{Kondo1949}
K.~Kondo.
\newblock A proposal of a new theory concerning the yielding of materials based
  on {R}iemannian geometry.
\newblock \emph{The Journal of the Japan Society of Aeronautical Engineering},
  2\penalty0 (8):\penalty0 29--31, 1949.

\bibitem[Kr{\"o}ner(1959)]{kroner1959allgemeine}
E.~Kr{\"o}ner.
\newblock Allgemeine kontinuumstheorie der versetzungen und eigenspannungen.
\newblock \emph{Archive for Rational Mechanics and Analysis}, 4\penalty0
  (1):\penalty0 273--334, 1959.

\bibitem[Kumar and Lopez-Pamies(2020)]{KumarLopezPamies2020}
A.~Kumar and O.~Lopez-Pamies.
\newblock The phase-field approach to self-healable fracture of elastomers: A
  model accounting for fracture nucleation at large, with application to a
  class of conspicuous experiments.
\newblock \emph{Theoretical and Applied Fracture Mechanics}, 107:\penalty0
  102550, 2020.

\bibitem[Kumar et~al.(2018)Kumar, Francfort, and Lopez-Pamies]{KFLP18}
A.~Kumar, G.~A. Francfort, and O.~Lopez-Pamies.
\newblock Fracture and healing of elastomers: A phase-transition theory and
  numerical implementation.
\newblock \emph{Journal of the Mechanics and Physics of Solids}, 112:\penalty0
  523--551, 2018.

\bibitem[Kumar et~al.(2020)Kumar, Bourdin, Francfort, and
  Lopez-Pamies]{KumarBourdinFrancfortLopezPamies2020}
A.~Kumar, B.~Bourdin, G.~A. Francfort, and O.~Lopez-Pamies.
\newblock Revisiting nucleation in the phase-field approach to brittle
  fracture.
\newblock \emph{Journal of the Mechanics and Physics of Solids}, 142:\penalty0
  104027, 2020.

\bibitem[Kumar et~al.(2024)Kumar, Liu, Dolbow, and Lopez-Pamies]{KLDLP24}
A.~Kumar, Y.~Liu, J.~E. Dolbow, and O.~Lopez-Pamies.
\newblock The strength of the brazilian fracture test.
\newblock \emph{Journal of the Mechanics and Physics of Solids}, 182:\penalty0
  105473, 2024.

\bibitem[Lam{\'e} and Clapeyron(1833)]{lame1833}
G.~Lam{\'e} and B.~P.~{\'E}. Clapeyron.
\newblock \emph{Memoir on the internal equilibrium of homogeneous solid
  bodies}.
\newblock Verlag nicht hermittelbar, 1833.

\bibitem[Lamont et~al.(2025)Lamont, Bouklas, and Vernerey]{lamont2025cohesive}
S.~C. Lamont, N.~Bouklas, and F.~J. Vernerey.
\newblock Cohesive instability in elastomers: {I}nsights from a crosslinked
  {V}an der {W}aals fluid model.
\newblock \emph{International Journal of Fracture}, 249\penalty0 (1):\penalty0
  20, 2025.

\bibitem[Larsen(2021)]{larsen2021}
C.~J. Larsen.
\newblock Variational fracture with boundary loads.
\newblock \emph{Applied Mathematics Letters}, 121:\penalty0 107437, 2021.

\bibitem[Larsen(2024)]{larsen2024}
C.~J. Larsen.
\newblock A local variational principle for fracture.
\newblock \emph{Journal of the Mechanics and Physics of Solids}, 187:\penalty0
  105625, 2024.

\bibitem[Lee and Liu(1967)]{lee1967finite}
E.~Lee and D.~Liu.
\newblock Finite-strain elastic-plastic theory with application to plane-wave
  analysis.
\newblock \emph{Journal of Applied Physics}, 38\penalty0 (1):\penalty0 19--27,
  1967.

\bibitem[Lee(1969)]{lee1969elastic}
E.~H. Lee.
\newblock Elastic-plastic deformation at finite strains.
\newblock \emph{Journal of Applied Mechanics}, 36\penalty0 (1):\penalty0 1--6,
  1969.

\bibitem[Lee(2013)]{Lee2013}
J.~M. Lee.
\newblock \emph{Introduction to Smooth Manifolds}, volume 218 of \emph{Graduate
  Texts in Mathematics}.
\newblock Springer, 2 edition, 2013.

\bibitem[Li and Wong(2013)]{li2013brazilian}
D.~Li and L.~N.~Y. Wong.
\newblock The brazilian disc test for rock mechanics applications: {R}eview and
  new insights.
\newblock \emph{Rock Mechanics and Rock Engineering}, 46\penalty0 (2):\penalty0
  269--287, 2013.

\bibitem[Liu(1982)]{liu1982}
I.~Liu.
\newblock On representations of anisotropic invariants.
\newblock \emph{International Journal of Engineering Science}, 20\penalty0
  (10):\penalty0 1099--1109, 1982.

\bibitem[Lopez-Pamies and Kamarei(2025)]{lopezpamieskamarei2025}
O.~Lopez-Pamies and F.~Kamarei.
\newblock When and where do large cracks grow? {G}riffith energy competition
  constrained by material strength.
\newblock \emph{Extreme Mechanics Letters}, 81:\penalty0 102417, 2025.

\bibitem[Lopez-Pamies et~al.(2025)Lopez-Pamies, Dolbow, Francfort, and
  Larsen]{LopezPamies2025}
O.~Lopez-Pamies, J.~E. Dolbow, G.~A. Francfort, and C.~J. Larsen.
\newblock Classical variational models cannot predict fracture nucleation.
\newblock \emph{Computer Methods in Applied Mechanics and Engineering},
  433:\penalty0 117520, 2025.

\bibitem[Lu(2012)]{Lu2012}
J.~Lu.
\newblock A covariant constitutive theory for anisotropic hyperelastic solids
  with initial strains.
\newblock \emph{Mathematics and Mechanics of Solids}, 17\penalty0 (2):\penalty0
  104--119, 2012.

\bibitem[Lu and Papadopoulos(2000)]{Lu2000}
J.~Lu and P.~Papadopoulos.
\newblock A covariant constitutive description of anisotropic non-linear
  elasticity.
\newblock \emph{Zeitschrift f{\"u}r Angewandte Mathematik und Physik},
  51\penalty0 (2):\penalty0 204--217, 2000.

\bibitem[Malmeister(1966)]{Malmeister1966}
A.~K. Malmeister.
\newblock Geometry of theories of strength.
\newblock \emph{Polymer Mechanics}, 2\penalty0 (4):\penalty0 324--331, 1966.

\bibitem[Marin(1957)]{Marin1957}
J.~Marin.
\newblock Theories of strength for combined stresses and nonisotropic
  materials.
\newblock \emph{Journal of the Aeronautical Sciences}, 24\penalty0
  (4):\penalty0 265--268, 1957.

\bibitem[Marsden and Hughes(1983)]{MarsdenHughes1983}
J.~E. Marsden and T.~J.~R. Hughes.
\newblock \emph{Mathematical Foundations of Elasticity}.
\newblock Prentice-Hall, Inc., Englewood Cliffs, New Jersey, 1983.

\bibitem[Maso et~al.(2005)Maso, Francfort, and Toader]{francfort2005nonlinear}
G.~D. Maso, G.~A. Francfort, and R.~Toader.
\newblock Quasistatic crack growth in nonlinear elasticity.
\newblock \emph{Archive for Rational Mechanics and Analysis}, 176\penalty0
  (2):\penalty0 165--225, 2005.

\bibitem[Mazars(1986)]{mazars1986}
J.~Mazars.
\newblock A model of a unilateral elastic damageable material and its
  application to concrete.
\newblock \emph{Fracture toughness and fracture energy of concrete}, pages
  61--71, 1986.

\bibitem[Mazars and Pijaudier-Cabot(1989)]{mazars1989}
J.~Mazars and G.~Pijaudier-Cabot.
\newblock Continuum damage theory—application to concrete.
\newblock \emph{Journal of engineering mechanics}, 115\penalty0 (2):\penalty0
  345--365, 1989.

\bibitem[Mazzucato and Rachele(2006)]{MazzucatoRachele2006}
A.~L. Mazzucato and L.~V. Rachele.
\newblock Partial uniqueness and obstruction to uniqueness in inverse problems
  for anisotropic elastic media.
\newblock \emph{Journal of Elasticity}, 83\penalty0 (3):\penalty0 205--245,
  2006.

\bibitem[Merodio and Ogden(2020)]{merodio2020finite}
J.~Merodio and R.~W. Ogden.
\newblock Finite deformation elasticity theory.
\newblock In \emph{Constitutive Modelling of Solid Continua}, pages 17--52.
  Springer, 2020.

\bibitem[Miehe et~al.(2010)Miehe, Welschinger, and Hofacker]{miehe2010}
C.~Miehe, F.~Welschinger, and M.~Hofacker.
\newblock Thermodynamically consistent phase-field models of fracture:
  Variational principles and multi-field fe implementations.
\newblock \emph{International journal for numerical methods in engineering},
  83\penalty0 (10):\penalty0 1273--1311, 2010.

\bibitem[Mihai and Neff(2017)]{MihaiNeff2017}
L.~A. Mihai and P.~Neff.
\newblock Hyperelastic bodies under homogeneous {C}auchy stress induced by
  non-homogeneous finite deformations.
\newblock \emph{International Journal of Non-Linear Mechanics}, 89:\penalty0
  93--100, 2017.

\bibitem[{Mihai} and {Neff}(2018)]{MihaiNeff2018}
L.~A. {Mihai} and P.~{Neff}.
\newblock Hyperelastic bodies under homogeneous {C}auchy stress induced by
  three-dimensional non-homogeneous deformations.
\newblock \emph{Mathematics and Mechanics of Solids}, 23\penalty0 (4):\penalty0
  606--616, 2018.

\bibitem[Mindlin and Cheng(1950)]{Mindlin1950}
R.~D. Mindlin and D.~H. Cheng.
\newblock Nuclei of strain in the semi-infinite solid.
\newblock \emph{Journal of Applied Physics}, 21\penalty0 (9):\penalty0
  926--930, 1950.

\bibitem[Morgan(1966)]{Morgan1966}
A.~J.~A. Morgan.
\newblock Some properties of media defined by constitutive equations in
  implicit form.
\newblock \emph{International Journal of Engineering Science}, 4\penalty0
  (2):\penalty0 155--178, 1966.

\bibitem[Mura(1982)]{Mura1982}
T.~Mura.
\newblock \emph{Micromechanics of Defects in Solids}.
\newblock Martinus Nijhoff, 1982.

\bibitem[Neff and Mihai(2017)]{NeffMihai2017}
P.~Neff and L.~A. Mihai.
\newblock Injectivity of the {C}auchy-stress tensor along rank-one connected
  lines under strict rank-one convexity condition.
\newblock \emph{Journal of Elasticity}, 127\penalty0 (2):\penalty0 309--315,
  2017.

\bibitem[Neff et~al.(2015{\natexlab{a}})Neff, Ghiba, and Lankeit]{Neff2015I}
P.~Neff, I.-D. Ghiba, and J.~Lankeit.
\newblock The exponentiated {H}encky-logarithmic strain energy. {P}art {I}:
  {C}onstitutive issues and rank-one convexity.
\newblock \emph{Journal of Elasticity}, 121\penalty0 (2):\penalty0 143--234,
  2015{\natexlab{a}}.

\bibitem[Neff et~al.(2015{\natexlab{b}})Neff, Lankeit, Ghiba, Martin, and
  Steigmann]{Neff2015II}
P.~Neff, J.~Lankeit, I.-D. Ghiba, R.~Martin, and D.~Steigmann.
\newblock The exponentiated {H}encky-logarithmic strain energy. {P}art {II}:
  coercivity, planar polyconvexity and existence of minimizers.
\newblock \emph{Zeitschrift f{\"u}r angewandte Mathematik und Physik},
  66\penalty0 (4):\penalty0 1671--1693, 2015{\natexlab{b}}.

\bibitem[Neff et~al.(2016)Neff, Eidel, and Martin]{Neff2016}
P.~Neff, B.~Eidel, and R.~J. Martin.
\newblock Geometry of logarithmic strain measures in solid mechanics.
\newblock \emph{Archive for Rational Mechanics and Analysis}, 222\penalty0
  (2):\penalty0 507--572, 2016.

\bibitem[Neff et~al.(2020)Neff, Graban, Schweickert, and
  Martin]{Neff2020Axiomatic}
P.~Neff, K.~Graban, E.~Schweickert, and R.~J. Martin.
\newblock The axiomatic introduction of arbitrary strain tensors by {H}ans
  {R}ichter--a commented translation of `{S}train tensor, strain deviator and
  stress tensor for finite deformations'.
\newblock \emph{Mathematics and Mechanics of Solids}, 25\penalty0 (5):\penalty0
  1060--1080, 2020.

\bibitem[Neff et~al.(2025)Neff, Holthausen, d’Agostino, Bernardini, Sky,
  Ghiba, and Martin]{Neff2025}
P.~Neff, S.~Holthausen, M.~V. d’Agostino, D.~Bernardini, A.~Sky, I.-D. Ghiba,
  and R.~J. Martin.
\newblock Hypo-elasticity, {C}auchy-elasticity, corotational stability and
  monotonicity in the logarithmic strain.
\newblock \emph{Journal of the Mechanics and Physics of Solids}, 202:\penalty0
  106074, 2025.

\bibitem[Noll(1963)]{Noll1963}
W.~Noll.
\newblock La mecanique classique, basee sur un axiome d'objectivite.
\newblock In \emph{La Methode Axiomatique dans les Mecaniques Classiques et
  NouveIies}, pages 47--56, Paris, 1963.

\bibitem[Norton(1993)]{Norton1993}
J.~D. Norton.
\newblock General covariance and the foundations of general relativity: Eight
  decades of dispute.
\newblock \emph{Reports on Progress in Physics}, 56\penalty0 (7):\penalty0
  791--858, 1993.

\bibitem[Ogden(1997)]{Ogden1984}
R.~W. Ogden.
\newblock \emph{Non-Linear Elastic Deformations}.
\newblock Dover, 1997.

\bibitem[Rajagopal(2003)]{Rajagopal2003}
K.~R. Rajagopal.
\newblock On implicit constitutive theories.
\newblock \emph{Applications of Mathematics}, 48:\penalty0 279--319, 2003.

\bibitem[Rajagopal(2007)]{Rajagopal2007}
K.~R. Rajagopal.
\newblock The elasticity of elasticity.
\newblock \emph{Zeitschrift f{\"u}r angewandte Mathematik und Physik},
  58:\penalty0 309--317, 2007.

\bibitem[Reissner(1931)]{Reissner1931}
H.~Reissner.
\newblock Eigenspannungen und {E}igenspannungsquellen.
\newblock \emph{Zeitschrift f{\"u}r Angewandte Mathematik und Mechanik},
  11\penalty0 (1):\penalty0 1--8, 1931.

\bibitem[Rivlin and Ericksen(1955)]{RivlinEricksen1955}
R.~S. Rivlin and J.~L. Ericksen.
\newblock Stress-deformation relations for isotropic materials.
\newblock \emph{Journal of Rational Mechanics and Analysis}, 4\penalty0
  (6):\penalty0 323--425, 1955.

\bibitem[Rockafellar(1970)]{Rockafellar1970}
R.~T. Rockafellar.
\newblock \emph{Convex Analysis}, volume~28 of \emph{Princeton Mathematical
  Series}.
\newblock Princeton University Press, Princeton, New Jersey, 1970.

\bibitem[Rudin(1976)]{Rudin1976}
W.~Rudin.
\newblock \emph{Principles of Mathematical Analysis}.
\newblock McGraw-Hill, 3 edition, 1976.

\bibitem[Sadik and Yavari(2017)]{Sadik2017}
S.~Sadik and A.~Yavari.
\newblock On the origins of the idea of the multiplicative decomposition of the
  deformation gradient.
\newblock \emph{Mathematics and Mechanics of Solids}, 22\penalty0 (4):\penalty0
  771--772, 2017.

\bibitem[Sato et~al.(1987)Sato, Awaji, Kawamata, Kurumada, and
  Oku]{sato1987graphite}
S.~Sato, H.~Awaji, K.~Kawamata, A.~Kurumada, and T.~Oku.
\newblock Fracture criteria of reactor graphite under multiaxial stesses.
\newblock \emph{Nuclear Engineering and Design}, 103\penalty0 (3):\penalty0
  291--300, 1987.

\bibitem[Schweickert et~al.(2018)Schweickert, Mihai, and Neff]{Schweickert2018}
E.~Schweickert, L.~A. Mihai, and P.~Neff.
\newblock Homogeneous {C}auchy stress induced by non-homogeneous deformations.
\newblock \emph{Proceedings in Applied Mathematics and Mechanics}, 18\penalty0
  (1):\penalty0 e201800185, 2018.

\bibitem[Sharma and Lim(1965)]{sharma1965experimental}
M.~Sharma and C.~Lim.
\newblock Experimental investigations on fracture of polymers.
\newblock \emph{Polymer Engineering \& Science}, 5\penalty0 (4):\penalty0
  254--262, 1965.

\bibitem[Simo and Marsden(1984)]{SimoMarsden1984}
J.~C. Simo and J.~E. Marsden.
\newblock On the rotated stress tensor and the material version of the
  {D}oyle-{E}ricksen formula.
\newblock \emph{Archive for Rational Mechanics and Analysis}, 86\penalty0
  (3):\penalty0 213--231, 1984.

\bibitem[Simo et~al.(1988)Simo, Marsden, and Krishnaprasad]{Simo1988}
J.~C. Simo, J.~E. Marsden, and P.~S. Krishnaprasad.
\newblock The {H}amiltonian structure of nonlinear elasticity: the material and
  convective representations of solids, rods, and plates.
\newblock \emph{Archive for Rational Mechanics and Analysis}, 104\penalty0
  (2):\penalty0 125--183, 1988.

\bibitem[Smith(1958)]{Smith1958}
T.~L. Smith.
\newblock {D}ependence of the ultimate properties of a {GR-S} rubber on strain
  rate and temperature.
\newblock \emph{Journal of Polymer Science}, 32:\penalty0 99--113, 1958.

\bibitem[Smith(1963)]{Smith1963}
T.~L. Smith.
\newblock {U}ltimate tensile properties of elastomers. {I}. {C}haracterization
  by a time and temperature independent failure envelope.
\newblock \emph{Journal of Polymer Science Part A}, 1:\penalty0 3597--3613,
  1963.

\bibitem[Smith(1964{\natexlab{a}})]{Smith1964a}
T.~L. Smith.
\newblock {R}elations between ultimate tensile properties of elastomers and
  their structure.
\newblock \emph{Proceedings of the Royal Society of London A}, 282:\penalty0
  102--113, 1964{\natexlab{a}}.

\bibitem[Smith(1964{\natexlab{b}})]{Smith1964b}
T.~L. Smith.
\newblock {U}ltimate tensile properties of elastomers. {II}. {C}omparison of
  failure envelopes for unfilled vulcanizates.
\newblock \emph{Journal of Applied Physics}, 35:\penalty0 27--36,
  1964{\natexlab{b}}.

\bibitem[Smith(1964{\natexlab{c}})]{smith1964}
T.~L. Smith.
\newblock Ultimate tensile properties of elastomers. ii. comparison of failure
  envelopes for unfilled vulcanizates.
\newblock \emph{Journal of Applied Physics}, 35\penalty0 (1):\penalty0 27--36,
  1964{\natexlab{c}}.

\bibitem[Smith and Rinde(1969)]{smith1969}
T.~L. Smith and J.~A. Rinde.
\newblock Ultimate tensile properties of elastomers. v. rupture in constrained
  biaxial tensions.
\newblock \emph{Journal of Polymer Science Part A-2: Polymer Physics},
  7\penalty0 (4):\penalty0 675--685, 1969.

\bibitem[Spencer(2015)]{Spencer2015}
A.~Spencer.
\newblock George {G}reen and the foundations of the theory of elasticity.
\newblock \emph{Journal of Engineering Mathematics}, 95\penalty0 (1):\penalty0
  5--6, 2015.

\bibitem[Spencer(1971)]{Spencer1971}
A.~J.~M. Spencer.
\newblock Part {III}. {T}heory of invariants.
\newblock \emph{Continuum Physics}, 1:\penalty0 239--353, 1971.

\bibitem[Spencer(1972)]{Spencer1972}
A.~J.~M. Spencer.
\newblock \emph{Deformations of Fibre-Reinforced Materials}.
\newblock Clarendon Press, Oxford, 1972.

\bibitem[Spencer(1982)]{Spencer1982}
A.~J.~M. Spencer.
\newblock The formulation of constitutive equation for anisotropic solids.
\newblock In \emph{Mechanical Behavior of Anisotropic Solids/Comportment
  M{\'e}chanique des Solides Anisotropes}, pages 3--26. Springer, 1982.

\bibitem[Spencer(1986)]{Spencer1986}
A.~J.~M. Spencer.
\newblock Modelling of finite deformations of anisotropic materials.
\newblock In \emph{Large Deformations of Solids: Physical Basis and
  Mathematical Modelling}, pages 41--52. Springer, 1986.

\bibitem[Truesdell(1952)]{Truesdell1952}
C.~Truesdell.
\newblock The mechanical foundations of elasticity and fluid dynamics.
\newblock \emph{Journal of Rational Mechanics and Analysis}, 1\penalty0
  (1):\penalty0 125--300, 1952.

\bibitem[Truesdell and Noll(2004)]{TruesdellNoll2004}
C.~Truesdell and W.~Noll.
\newblock \emph{The {N}onlinear {F}ield {T}heories of {M}echanics}.
\newblock Springer, Berlin, 2 edition, 2004.

\bibitem[Tsai(1984)]{Tsai1984}
S.~W. Tsai.
\newblock A survey of macroscopic failure criteria for composite materials.
\newblock \emph{Journal of Reinforced Plastics and Composites}, 3\penalty0
  (1):\penalty0 40--62, 1984.

\bibitem[Tsai and Wu(1971)]{TsaiWu1971}
S.~W. Tsai and E.~M. Wu.
\newblock A general theory of strength for anisotropic materials.
\newblock \emph{Journal of Composite Materials}, 5\penalty0 (1):\penalty0
  58--80, 1971.

\bibitem[Ueda et~al.(1975)Ueda, Fukuda, Nakacho, and Endo]{Ueda1975}
Y.~Ueda, K.~Fukuda, K.~Nakacho, and S.~Endo.
\newblock A new measuring method of residual stresses with the aid of finite
  element method and reliability of estimated values.
\newblock \emph{Transactions of JWRI}, 4\penalty0 (2):\penalty0 123--131, 1975.

\bibitem[Volokh(2007)]{volokh2007}
K.~Volokh.
\newblock Hyperelasticity with softening for modeling materials failure.
\newblock \emph{Journal of the Mechanics and Physics of Solids}, 55\penalty0
  (10):\penalty0 2237--2264, 2007.

\bibitem[von Mises(1913)]{Mises1913}
R.~von Mises.
\newblock Mechanik der festen {K{\"o}}rper im plastisch-deformablen zustand.
\newblock \emph{Nachrichten von der Gesellschaft der Wissenschaften zu
  {G{\"o}}ttingen, Mathematisch-Physikalische Klasse}, 1913:\penalty0 582--592,
  1913.

\bibitem[von Mises(1928)]{Mises1928}
R.~von Mises.
\newblock Mechanik der plastischen form{\"a}nderung von kristallen.
\newblock \emph{Zeitschrift f{\"u}r Angewandte Mathematik und Mechanik},
  8\penalty0 (3):\penalty0 161--185, 1928.

\bibitem[Ward and Kumar(2025)]{ward2025}
O.~Ward and A.~Kumar.
\newblock Why planar cracks fragment into echelon cracks.
\newblock \emph{arXiv preprint arXiv:2512.16053}, 2025.

\bibitem[Weyl(1939)]{Weyl1939}
H.~Weyl.
\newblock \emph{The Classical Groups: Their Invariants and Representations},
  volume~1 of \emph{Princeton Mathematical Series}.
\newblock Princeton University Press, Princeton, New Jersey, 1939.

\bibitem[Xiao et~al.(1997)Xiao, Bruhns, and Meyers]{Xiao1997}
H.~Xiao, O.~Bruhns, and A.~Meyers.
\newblock Logarithmic strain, logarithmic spin and logarithmic rate.
\newblock \emph{Acta Mechanica}, 124:\penalty0 89--105, 1997.

\bibitem[Yavari(2021)]{Yavari2021Eshelby}
A.~Yavari.
\newblock {On} {E}shelby’s inclusion problem in nonlinear anisotropic
  elasticity.
\newblock \emph{Journal of Micromechanics and Molecular Physics}, 6\penalty0
  (1):\penalty0 2150002, 2021.

\bibitem[Yavari(2025)]{Yavari2025Universal}
A.~Yavari.
\newblock On universal deformations of compressible {C}auchy elastic solids
  reinforced by inextensible fibers.
\newblock \emph{Journal of the Mechanics and Physics of Solids}, page 106340,
  2025.

\bibitem[Yavari and Goriely(2013)]{YavariGoriely2013}
A.~Yavari and A.~Goriely.
\newblock Nonlinear elastic inclusions in isotropic solids.
\newblock \emph{Proceedings of the Royal Society A}, 469\penalty0
  (2160):\penalty0 20130415, 2013.

\bibitem[Yavari and Goriely(2021)]{YavariGoriely2021}
A.~Yavari and A.~Goriely.
\newblock Universal deformations in anisotropic nonlinear elastic solids.
\newblock \emph{Journal of the Mechanics and Physics of Solids}, 156:\penalty0
  104598, 2021.

\bibitem[Yavari and Goriely(2023)]{YavariGoriely2023Universal}
A.~Yavari and A.~Goriely.
\newblock The universal program of nonlinear hyperelasticity.
\newblock \emph{Journal of Elasticity}, 154\penalty0 (1):\penalty0 91--146,
  2023.

\bibitem[Yavari and Goriely(2024)]{YavariGoriely2024}
A.~Yavari and A.~Goriely.
\newblock Controllable deformations in compressible isotropic implicit
  elasticity.
\newblock \emph{Zeitschrift f{\"u}r angewandte Mathematik und Physik},
  75\penalty0 (5):\penalty0 169, 2024.

\bibitem[Yavari and Goriely(2025)]{YavariGoriely2025}
A.~Yavari and A.~Goriely.
\newblock Nonlinear {C}auchy elasticity.
\newblock \emph{Archive for Rational Mechanics and Analysis}, 249\penalty0
  (5):\penalty0 57, 2025.

\bibitem[Yavari and Marsden(2012)]{YavariMarsden2012}
A.~Yavari and J.~E. Marsden.
\newblock Covariantization of nonlinear elasticity.
\newblock \emph{Zeitschrift f{\"u}r Angewandte Mathematik und Physik},
  63\penalty0 (5):\penalty0 921--927, 2012.

\bibitem[Yavari and Sozio(2023)]{YavariSozio2023}
A.~Yavari and F.~Sozio.
\newblock On the direct and reverse multiplicative decompositions of
  deformation gradient in nonlinear anisotropic anelasticity.
\newblock \emph{Journal of the Mechanics and Physics of Solids}, 170:\penalty0
  105101, 2023.

\bibitem[Yavari et~al.(2006)Yavari, Marsden, and Ortiz]{Yavari2006}
A.~Yavari, J.~E. Marsden, and M.~Ortiz.
\newblock On the spatial and material covariant balance laws in elasticity.
\newblock \emph{Journal of Mathematical Physics}, 47:\penalty0 85--112, 2006.

\bibitem[Yu(2002)]{Yu2002}
M.-H. Yu.
\newblock Advances in strength theories for materials under complex stress
  state in the 20th century.
\newblock \emph{Applied Mechanics Reviews}, 55\penalty0 (3):\penalty0 169--218,
  2002.

\bibitem[Yu(2004)]{Yu2004}
M.-H. Yu.
\newblock \emph{Unified Strength Theory and Its Applications}.
\newblock Springer, Berlin, 2004.

\bibitem[Zeng et~al.(2025)Zeng, Guilleminot, and Dolbow]{zeng2025stochastic}
B.~Zeng, J.~Guilleminot, and J.~E. Dolbow.
\newblock Examining crack nucleation under spatially uniform stress states with
  a complete phase-field model for fracture.
\newblock \emph{Theoretical and Applied Fracture Mechanics}, page 105170, 2025.

\bibitem[Zheng(1994)]{zheng1994theory}
Q.~S. Zheng.
\newblock Theory of representations for tensor functions.
\newblock \emph{Applied Mechanics Reviews}, 47\penalty0 (11):\penalty0
  545--587, 1994.

\bibitem[Zheng and Spencer(1993)]{zheng1993}
Q.-S. Zheng and A.~J.~M. Spencer.
\newblock Tensors which characterize anisotropies.
\newblock \emph{International Journal of Engineering Science}, 31\penalty0
  (5):\penalty0 679--693, 1993.

\bibitem[Zhou et~al.(2013)Zhou, Hoh, Wang, Keer, Pang, Song, and
  Wang]{Zhou2013}
K.~Zhou, H.~J. Hoh, X.~Wang, L.~M. Keer, J.~H. Pang, B.~Song, and Q.~J. Wang.
\newblock A review of recent works on inclusions.
\newblock \emph{Mechanics of Materials}, 60:\penalty0 144--158, 2013.

\end{thebibliography}

\end{document}